%% file: MU-OAMP/main.tex
\documentclass[12pt,draftclsnofoot,onecolumn]{IEEEtran}
\usepackage[utf8]{inputenc}
\input{macros}

\title{An Orthogonal Approximate Message Passing Framework for Multiuser Communications}

\pgfplotsset{compat=1.18}

\begin{document}
	\author{{Burak~\c{C}akmak}, {Hao Yan}, {Alexander Fengler},\\{Giuseppe~Caire}, and 
    {Lei~Liu}
\thanks{The work of B. \c{C}akmak was supported by the Gottfried Wilhelm Leibniz-Preis 2021 of DFG. The work of G. Caire was supported by BMBF Germany in the program of ``Souver\"an. Digital. Vernetzt.'' Joint Project 6G-RIC (Project IDs 16KISK030).
A. Fengler was funded by the Deutsche Forschungsgemeinschaft (DFG, German Research Foundation) – Grant 471512611.}
\thanks{Burak \c{C}akmak and Giuseppe Caire are with the Faculty of Electrical Engineering and Computer Science, Technical University of Berlin, Berlin 10587, Germany (e-mail: burak.cakmak@tu-berlin.de; caire@tu-berlin.de).}
\thanks{Hao Yan and Lei Liu are with the College of Information Science and Electronic Engineering, Zhejiang University, Hangzhou 310027, China (e-mail: hao\_yan@zju.edu.cn; lei\_liu@zju.edu.cn).}
\thanks{Alexander Fengler is with the Department of Electrical Engineering and Information Technology, Karlsruhe Institute of Technology, 76131 Karlsruhe, Germany (e-mail:  fengler@kit.edu).}
} 
	\maketitle
\begin{abstract}
We solve the open problem of constructing a Bayes-optimal iterative signal recovery algorithm 
for linear-Gaussian \emph{multiuser} communication systems with random precoding at the transmitters.
Specifically, we consider the received signal model
$\mathbf{y} = \sum_{u} \mathbf{H}_u \mathbf{\Xi}_u \mathbf{s}_u + \mathbf{n}$,
where $\mathbf{n}$ is white Gaussian noise,  $\{\mathbf{H}_u \in \mathbb{C}^{L 
\times L}\}$ are discrete-time channel matrices --- modeling a wide class of generally time-varying and dispersive linear channels with possibly multiple antennas --- and the precoding matrices $\{\boldsymbol{\Xi}_u \in \mathbb{C}^{L \times N_u}\}$ are drawn independently from a right-unitarily invariant random matrix ensemble. 
We consider generic \emph{non-separable} (coded) systems where the users' signals
$\{\mathbf{s}_u\}$ follow general (non-factorizing) distributions.
For this model, we introduce a novel orthogonal/vector approximate message passing
(OAMP/VAMP)-type framework, including an algorithm and its high-dimensional (but finite-sample) analysis. 
From an algorithmic standpoint, the proposed method can be interpreted as an \emph{interpolation} between
Minka's expectation propagation (EP)---a widely used method in machine
learning---and OAMP.
Our main theoretical contribution is the explicit finite-sample analysis of the
proposed algorithm. Furthermore, we analyze the associated inference problem via a replica-symmetric (RS) ansatz by using
a novel disorder-averaging technique. Both  the (rigorous) high-dimensional analysis of the algorithm and the RS ansatz reveal the same decoupling principle, establishing that the proposed algorithm is asymptotically
Bayes-optimal under the validity of the RS ansatz.
\end{abstract}

%%%%%%%%%%%%%%%%%%%%%%%%%%%%%%%%%%%%%%%%%%%%%%%%%%%%%%%%%%%%%%%%%%%%%%%%%%%%%%%%%
\section{Introduction}

We consider a linear-Gaussian multiaccess channel where $U$ users simultaneously transmit to a receiver. 
The received signal can be expressed as
\begin{equation}
\yv = \sum_{u=1}^{U} \Hm_u \xv_u + \nv,   \label{eq:system_model}
\end{equation}
where \( \nv \sim \mathcal{CN}(\mathbf{0}, \sigma^2 \Id_L) \) denotes additive white Gaussian noise,
\( \xv_u \in \mathbb{C}^{L} \) is the transmitted signal of the \(u\)-th user, and 
\( \Hm_u \in \mathbb{C}^{L \times L} \) is the discrete-time channel matrix of the  \(u\)-th user. This may include a physical 
linear dispersive possibly time-varying and possibly multi-input multi-output (MIMO) channel, e.g., as originated by multiple antenna transmission/reception in wireless communications\cite{tse2005fundamentals}, as well as the precoding and postprocessing for specific waveforms such as OFDM and OTFS (e.g., IDFT/DFT, insertion and removal of cyclic prefix, symplectic DFT, discrete Zak transform and so on). Specific relevant examples are discussed in the results of Section~\ref{sec_sim}. The channel matrices $\{\Hm_u\}$ may be deterministic or a realization of some random matrix. However, in the perspective of the receiver, they are fixed and known. 
%For example, under the OFDM\footnote{Abbreviation for Orthogonal Frequency-Division Multiple Access.} convention, we may consider \( L = nM \), where \( n \) denotes the number of OFDM symbols per frame and \( M \) the number of subcarriers per block.

Following \cite{lei26}, we consider the case where each user applies a linear random precoder before transmission \footnote{For an illustration of the advantages of random precoding, we refer to \cite[Section~III-D]{lei26}. Sparse regression (SR)\cite{barron12} encoded  signals $\xv_u$ naturally have a random precoding structure and can be studied within our framework. }. Specifically,
\begin{equation}
\xv_u = \Xim_u \sv_u ,
\label{eq:precoding}
\end{equation}
where \( \sv_u \in \mathbb{C}^{N_u} \) denotes the information-bearing signal vector and
\( \Xim_u \in \mathbb{C}^{L \times N_u} \) is a random precoding matrix.  Each precoder \( \Xim_u \) is drawn independently from a right-unitarily invariant ensemble, i.e., \( \Xim_u \sim \Xim_u \Vm \) for any unitary matrix \( \Vm \) of appropriate dimensions, statistically independent of \( \Xim_u \). Equivalently, \( \Xim_u \) admits the singular value decomposition
\begin{equation}
\Xim_u = \mathbf{U}_u \mathbf{\Lambda}_u \mathbf{O}_u ,
\label{svd}
\end{equation}
where \(\mathbf{U}_u \in \mathbb{C}^{L \times L} \) and \( \mathbf{O}_u \in \mathbb{C}^{N_u \times N_u} \) are unitary matrices, and \( \mathbf{\Lambda}_u \in \mathbb{R}^{L \times N_u} \) is diagonal. The matrices $\{ \mathbf{U}_u, \mathbf\Lambda_u\}$ are generic 
whereas \( \mathbf{O}_u\) are independent Haar-distributed matrices, i.e., chosen uniformly at random from the set of unitary $N_u\times N_u$ matrices. Notice that while the precoding matrices are generated at random by sampling the appropriate distribution, they are known to the decoder, as part of the transmission scheme. 
Our framework applies to general right-unitarily invariant precoders, which arise in a variety of contexts, e.g., i.i.d. Gaussian precoding matrices in high-dimensional time series~\cite{tieplova2025information}, as well as Haar unitary matrices and their row submatrices (i.e., random semi-unitary matrices~\cite{chikuse2003statistics}) used in our setting.

The goal of the receiver is to recover the signal vectors \( \{\sv_u\} \) given the received signal \( \yv \) and the {\em effective} matrices $\{\Hm_u \Xim_u\}$.

An important feature of the considered setting is that users may employ coded, non-separable signaling schemes, in which the components of each signal vector $\sv_u$ exhibit structured dependencies. As a concrete example, consider the recently proposed sparse-regression LDPC (SR-LDPC) codes~\cite{Ebe25}. These codes partition the signal into multiple sections, each containing exactly one nonzero component. The position of the nonzero component in each section is determined by an outer, non-binary LDPC code, thereby introducing global dependencies across sections.

In the single-user setting (\(U=1\)) with i.i.d.\ input symbols, the associated inference problem is by now well understood. Fundamental performance limits have been characterized using the replica-symmetric (RS) ansatz~\cite{tanaka2002statistical,Tulino13}, and these predictions have been rigorously validated~\cite{reeves2019replica,barbier2020mutual,10272997}. On the algorithmic side, Bayes-optimal recovery can be achieved using the Approximate Message Passing (AMP) framework (algorithm and its high-dimensional analysis) \cite{Kabashima,Donoha,opper2016theory,fan2022approximate}, as well as the orthogonal/vector AMP (OAMP/VAMP) framework~\cite{ma2017orthogonal,rangan2019vector,takeuchi2019rigorous}.

For non-separable input distributions, extensions of the AMP framework were initiated in~\cite{Berthier20}, relying on certain (implicit) limiting assumptions and establishing convergence in a weaker sense than almost sure convergence. These results were recently sharpened  in~\cite{reeves2025dimension}. Extensions of OAMP/VAMP to non-separable inputs were initiated in~\cite{fletcher2018plug}.

Existing AMP and OAMP frameworks, however, are limited to restricted models. In particular, AMP requires \( \Hm_u = \Id_L \) and i.i.d.\ Gaussian precoding matrices \( \Xim_u \). Even in simple static wireless channels, the matrices \( \Hm_u \) are generally diagonal with non-identical entries, thereby violating the assumptions required for AMP state evolution. Alternatively, one may stack all users into a joint linear model by defining
\begin{subequations}
\label{joint}
\begin{align}
\sv &\triangleq [\sv_1^\top, \sv_2^\top, \ldots, \sv_U^\top]^\top, \\
\Am &\triangleq [\Am_1, \Am_2, \ldots, \Am_U], \qquad \text{with } \Am_u \doteq \Hm_u \Xim_u, 
\end{align}     
such that we can rewrite \eqref{eq:system_model} as 
\begin{equation}
\yv = \Am \sv + \nv .  \label{channel2}
\end{equation}
\end{subequations}
Then, one may use the conventional OAMP framework, which is developed under the assumption that \( \Am \) is right-unitarily invariant. However, in the multiuser setting considered here, the right-unitarily invariant property generally fails due to the superposition of independently precoded channels. As a result, a direct application of classical OAMP typically leads to poor performance.

%%%%%%%%%%%%%%%%%%%%%%%%%%%%%%%%%%%%%%%%%%%%%%%%%%%%%%%%%
\subsection{Overview of Contributions}

In this study, we introduce a general theoretical framework for signal recovery in multisource linear observation models of the form~\eqref{eq:system_model} with independent random precoding \eqref{eq:precoding}. The framework accommodates generic channel matrices, right-unitarily invariant precoding matrices, and non-separable input signals. 

We propose a novel OAMP-type algorithm that jointly processes the user-specific effective matrices. From an algorithmic standpoint, the proposed method can be interpreted as an intermediate way between Minka’s expectation propagation (EP) \cite{minka2005divergence} (a commonly used method in machine learning) and OAMP.

We present an explicit (and rigorous) finite-sample analysis~\cite{rush2018,rush24,Cakmakisit24,cakmakit25} of the evolution of the iterative algorithm (referred to as ``dynamics'') applicable to generic non-separable systems.  While strong results in this respect have recently been established for AMP-type dynamics coupled with i.i.d.\ Gaussian random matrices~\cite{reeves2025dimension}, to the best of our knowledge, this is the first work that provides an explicit high-dimensional and finite-sample characterization of OAMP-type dynamics with non-separable inputs.

The proof technique may be of independent methodological interest and is therefore presented in the main body of the paper. Specifically, we leverage the ``Householder Dice'' representation of AMP-type dynamics with Haar-distributed matrices~\cite{lu2021householder} to construct a statistically equivalent representation that eliminates explicit dependencies on Haar matrices. We then apply concentration results from high-dimensional probability (summarized in Appendix~\ref{preliminariesop}) to characterize the resulting dynamics in the high-dimensional limit.

Furthermore, under the RS ansatz, we derive asymptotic expressions for the input–output mutual information and the per-user minimum mean-square error (MMSE). The RS analysis itself relies on a novel disorder-averaging technique (see Appendix~\ref{disorder_average}) rather than the asymptotic Itzykson–Zuber integral approach, which is routinely applied to the single-user case, $U = 1$. 

Both the finite-sample analysis and the RS predictions reveal the same asymptotic behaviour. This correspondence implies, under the validity of the RS ansatz, that the proposed algorithm achieves asymptotic Bayes optimality.

%%%%%%%%%%%%%%%%%%%%%%%%%%%%%%%%%%%%%%%%%%%%%%%%%%%
\subsection{Related Works}

Despite the practical relevance of the multiuser observation model~\eqref{eq:system_model}, its theoretical understanding remains limited. In the special case \( \Hm_u \equiv \Id_L \) and \( \Xim_u \equiv \Om_u \), where \( \Om_u \) is an \( L \times L \) Haar-distributed unitary matrix, the resulting observation model---referred to as the \emph{\(U\)-orthogonal form}---was analyzed via the RS ansatz in~\cite{vehkapera2016analysis}.
Recently, in the context of high-dimensional time series,~\cite{tieplova2025information} studied
   a similar problem where \( \mathbf{H}_u\) is a deterministic circulant matrix and
\( \boldsymbol{\Xi}_u \) is an i.i.d.\ Gaussian random matrix, and input signals
\( \sv_u\) with i.i.d. components. For this particular setting, the authors derived the replica predictions for the mutual information. However, no algorithm was proposed to achieve the Bayes-optimal performance predicted by the theory.

The multiuser OAMP (MU-OAMP) framework developed in this paper resolves this open issue.
From a technical standpoint, the following works are noted: (i) The asymptotic analysis of classical OAMP/VAMP with non-separable inputs was initiated in~\cite{fletcher2018plug}. That analysis relies on implicit limiting assumptions and establishes convergence in a weaker sense than almost sure convergence; (ii) our proof strategy is aligned with \cite{Cakmakisit24,cakmakit25}, which studies AMP-type dynamics involving i.i.d.\ Gaussian random matrices (rather than Haar matrices) and under i.i.d. (separable) input signals.

%%%%%%%%%%%%%%%%%%%%%%%%%%%%%%%%%%%%%

\subsection{Organization}
The remainder of the paper is organized as follows.
Section~\ref{derivation} discusses the algorithmic aspects of the proposed method, including the (self-contained) derivation of the proposed algorithm and its connections to classical EP and OAMP algorithms.
In Section~\ref{main_result}, we introduce the MU-OAMP framework and present its finite-sample analysis. Section~\ref{replica_sect} presents the RS ansatz and the algorithm's Bayesian optimality. Some numerical examples are presented in Section~\ref{sec_sim}. The proof of the main theoretical result is given in Section~\ref{Pth1}. Conclusions and future research directions are discussed in Section~\ref{sec_conc}. Additional technical derivations are deferred to the appendices.

\subsection{{Notations}}\label{NotDef}
We use $\Am > \bf 0$ (resp. $\Am\geq  \bf 0$) to indicate that $\Am$ is positive definite (resp. positive semi-definite). $\Vert \cdot \Vert_2$ is the spectral norm (largest singular value) of the matrix in the argument. For a vector $\av$, $\av>\matr 0$ (resp. $\av\geq  \matr 0$) indicates componentwise positivity (resp. non-negativity). We use $(\cdot)^\transp$ and $(\cdot)^\dagger$ to denote transpose and Hermitian transpose, and $\otimes$ to indicate Kronecker product. Also, $\Id_m$ indicates the $m \times m$ identity matrix. We use $\xv \sim \mathcal{CN}(\boldsymbol{\mu}, \boldsymbol{\Sigma})$ to indicate that $\xv$ is circularly-symmetric complex Gaussian distributed 
with mean $\boldsymbol{\mu}$ and covariance  $\boldsymbol{\Sigma}$, and denote the corresponding 
density function as $\textswab{g}(\cdot \vert \muv, \matr\Sigma)$. We write $\av \sim_{\text{i.i.d.}} A$ 
to indicate that the entries of a random vector $\av$ are statistically independent and identically distributed as the random variable $A$. When required, we also write $\av \sim_{\text{i.i.d.}} {\rm P}$ where ${\rm P}$ is the marginal distribution the components of $\av$. We define the normalized trace of a matrix $\Am\in\CC^{N\times N}$ by
$\langle\Am \rangle\doteq \frac{1}{N}{\rm tr}(\Am)$ and, for $\av,\bv\in \CC^{N}$, we define the normalized inner product
$\langle \av,\bv \rangle\doteq \frac{1}{N}\av^\dagger\bv$. For a function, $\eta:\mathbb{C}^{N}\to \mathbb{C}^{N}$, we denote its $N \times N$ Jacobian matrix by $\eta'$ with the entries
$
[\eta'(\mathbf{r})]_{ij} \doteq \frac{\partial \eta_i(\mathbf{r})}{\partial r_j}$ for all $i,j$ and the derivatives are understood in the Wirtinger sense\footnote{For a complex variable $r = x + \mathrm{i}y$, the Wirtinger derivative is defined as
$
\frac{\partial}{\partial r}
= \frac{1}{2}\left( \frac{\partial}{\partial x} - \mathrm{i}\frac{\partial}{\partial y} \right)$.}.
Finally, for two jointly distributed random vector $\xv, \yv$ we define $\mathrm{mmse}(\mathbf{x}\mid\mathbf{y})
 \doteq \mathbb{E}\!\left[\|\mathbf{x}-\mathbb{E}[\mathbf{x}\mid\mathbf{y}]\|^2\right]$.

%%%%%%%%%%%%%%%%%%%%%%%%%%%%%%%%%%%%%%%%%%%%%%%%%%%%%%%%%%%%%%%%
\section{Basic Algorithmic Aspects}\label{derivation}

In this section, we provide a compact and self-contained derivation of the proposed algorithm and articulate its connection to Minka's EP \cite{minka2005divergence} and OAMP \cite{ma2017orthogonal,rangan2019vector}. The goal of this section is to convey the rationale of the proposed MU-OAMP algorithm before presenting the main results in the next section.

Recall the observation model and notation in \eqref{joint}. Let $p_0$ denote the prior probability density function (pdf) 
of $\sv$. Then, the posterior pdf of $\sv$ given $\yv, \Am$ factorizes according to 
\begin{equation}
\label{posterior}
p(\sv\vert\yv,\Am)\propto p_0(\sv)\underbrace{\textswab{g}(\yv\vert \Am\sv,\sigma^2\Id_L)}_{\doteq p_1(\sv)}\;
\end{equation}
Since the factor $p_0$ is in general non-Gaussian, the exact computation of the joint a-posteriori pdf $p(\sv\vert\yv,\Am)$ and/or of the posterior marginals $p(s_i \vert\yv,\Am)$, which is the ultimate goal of Bayesian inference, is typically infeasible for large dimension $L$. Then, we follow the (Gaussian) expectation consistency (EC) approximate inference method of \cite{opper2005expectation}: the marginal posterior pdfs $p(s_i\vert\matr y, \Am)$ are approximated by Gaussian pdfs $q(s_i)$ which are the marginals of the joint pdf
\begin{equation}
q(\sv)%\triangleq\frac{1}{Z}
\propto\underbrace{\textswab{g}(\sv\vert\fv,{\rm diag}{(\matr\lambda)})}_{\doteq q_0(\sv)}
\underbrace{\textswab{g}(\sv\vert \rv,{\rm diag}{(\matr\tau)})}_{\doteq q_1(\sv)}\label{apdf}
\end{equation}
where the vectors of means $(\fv,\rv)$ and variances $(\matr\lambda,\matr\tau)$ are determined via a system of equations obtained by imposing suitable consistency conditions. The direct solution of such system of equations is generally infeasible. However, the resulting conditions can be interpreted as a (vector-valued) fixed-point equation, for which a solution can be found in an iterative way starting from a suitable initialization. The resulting iterative algorithm provides the rationale for the actual algorithm proposed in this paper.

%%%%%%%%%%%%%%%%%%%%%%%%%%%%
\subsection{Consistency conditions}

For an unnormalized density function in product form $m_0(\sv)m_1(\sv)$ and a measurable function $g(\cdot)$
we introduce the notation \[\mathbb{E}[g(\sv)]_{m_0\cdot m_1}\doteq \frac{\int g(\sv) m_0(\sv)
m_1(\sv){\rm d}\sv}{\int m_0(\sv)
m_1(\sv){\rm d}\sv} \;.\] 
The consistency conditions are obtained by imposing EC  \cite{opper2005expectation} via the equalities
\begin{equation}
\mathbb{E}[g(\matr \sv)]_{q_0\cdot p_1}=
\mathbb{E}[g(\matr \sv)]_{q_0\cdot q_1}=\mathbb{E}[g(\matr \sv)]_{p_0\cdot q_1}\label{EC}\;.
\end{equation}
for $g(\cdot)$ involving first and second moments of $\sv$, since the underlying Gaussian distributions $q_0$ and $q_1$ are completely specified by their mean and variances. 
Notice that the first equality in \eqref{EC} introduces consistency with the channel transition pdf $p_1$, while the second
equality in \eqref{EC} introduces consistency with the prior pdf $p_0$. 
Different conditions (and therefore different iterative algorithms) are obtained depending on the imposed constraint 
on the form of the variance vectors $(\matr\lambda,\matr\tau)$.  In particular, 
consider the following restricted forms
\begin{align}
\label{unified}
(\matr\lambda,\matr\tau)=\left\{\begin{array}{cc}
(\lambda_{1},\lambda_{2},\cdots,\lambda_{N}), (\tau_{1},\tau_{2},\cdots,\tau_{N}) & \text{Minka's EP} \\
(\lambda_{1}\matr 1_{N_1};\lambda_{2}\matr 1_{N_2};\cdots;\lambda_{U}\matr 1_{N_U}),(\tau_{1}\matr 1_{N_1};\tau_{2}\matr 1_{N_2};\cdots;\tau_{U}\matr 1_{N_U})  
& \text{MU-OAMP}  \\ 
 \lambda\matr 1_{N},\tau\matr 1_N
& \text{OAMP/VAMP}   \\ 
\end{array}\right.\
\end{align}
where $N\doteq\sum_uN_u$. Here, for example, Minka's EP and OAMP refer to the algorithms that iteratively solve the EC conditions \eqref{EC} for unrestricted variance profiles of $\matr\lambda, \matr\tau$ and the fully restricted variance profiles $\matr\lambda=\lambda\matr 1_{N}, \matr\tau=\tau\matr 1_{N}$, respectively.

%\begin{definition}\label{generic} For convenience, we may consider a generic $U$-%source model where $N_u$ can be arbitrary number in $[1,N]$ such that $\sum_{u} N_u  %= N$. Then, we set  
%\begin{equation}
%g(\sv)= (\sv, \norm{\sv_1}^2, \norm{\sv_2}^2,\cdots \norm{\sv_U}^2). %\label{generic_moments}
%\end{equation}
%\end{definition}}

Next, we derive the explicit form of the EC \eqref{EC} under the block-constant restriction (middle line in \eqref{unified}).
To this end, we first note the following identities, which follow from the additivity of the \emph{natural} parameters of Gaussian densities under multiplication. Let the Gaussian density $q$ be given as in \eqref{apdf}. Then,
\begin{subequations}
\label{EPrule}
\begin{align}
\frac{1}{\frac{1}{N_u}\mathbb E[\norm{\sv_u-\mathbb E[\sv_u]_{q}}^2]_{q}}
&=\frac{1}{\tau_u}+\frac{1}{\lambda_u}\\
\frac{\mathbb E[\sv_u]_{q}}{\frac{1}{N_u}\mathbb E[\norm{\sv_u-\mathbb E[\sv_u]_q}^2]_{q}}&= \frac{\rv_u}{\tau_u}+\frac{\fv_u}{\lambda_u}  \;.
\end{align}
\end{subequations}
Throughout this section, we partition any vector vector $\av\in \CC^N$ as $\av =[\av_1^\top, \av_2^\top,\cdots, \av_U^\top]^\top$ where $\av_u \in \CC^{N_u}$. Moreover, we define the selection matrix 
$\mathbf{P}_u \in \{0,1\}^{N_u \times N}$ such that $\mathbf{a}_u = \mathbf{P}_u \mathbf{a}$.

Defining 
\begin{align}
&\matr\mu_u \doteq \mathbb E[\sv_u]_{q_0\cdot p_1} \text{~~and~~}\sigma_u^2 \doteq \frac{1}{N_u}\mathbb E[\Vert\sv_u-\matr\mu_u\Vert^2]_{q_0\cdot p_1}\label{mom1}\;, 
\end{align}
and using \eqref{mom1} and \eqref{EPrule} in the first equality in \eqref{EC}, we obtain
\begin{subequations}
\label{conditionq0p1}
\begin{align}
\frac{1}{\tau_u}&=\frac{1}{\sigma_u^2}-\frac{1}{\lambda_u}\\
\frac{\rv_u}{\tau_u}&=\frac{\matr\mu_u}{\sigma_u^2}-\frac{\fv_u}{\lambda_u}\;.
\end{align}   
\end{subequations}
Furthermore, using the fact that, under $q_0 p_1$, $\sv$ is conditionally Gaussian given $\yv, \Am$, we obtain the explicit expressions
\begin{subequations}
\label{q0p1explicit}
\begin{align}
\matr\mu_u&=\Pm_u\left(\matr \Lambda^{-1}+\frac{1}{\sigma^2}\Am^\dagger\Am\right)^{-1}\left(\matr\Lambda^{-1}\fv+\frac{1}{\sigma^2}\Am^\dagger\yv\right)\nonumber \\
&=\fv_u +\lambda_{u}\Am_{u}^\dagger\left(\sigma^2\Id_L+ \sum_{u'}\lambda_{u'}\Am_{u'}\Am_{u'}^\dagger\right)^{-1} \left(\yv-\Am\fv\right)\\
\sigma_u^2&=\left\langle\Pm_u\left(\matr \Lambda^{-1}+\frac{1}{\sigma^2}\Am^\dagger\Am\right)^{-1}\Pm_u^\dagger\right\rangle \nonumber \\
&=\lambda_{u}-\lambda_u^2 \left\langle\Am_u^\dagger \left(\sigma^2\Id_L+ \sum_{u'}\lambda_{u'}\Am_{u'}\Am_{u'}^\dagger\right)^{-1}\Am_u\right\rangle\;, 
\end{align}  
\end{subequations}
where we define the diagonal matrix
$\matr\Lambda\doteq{\rm diag}(\lambda_{1}\matr 1_{N_1};\lambda_{2}\matr 1_{N_2};\cdots;\lambda_{U}\matr 1_{N_U})$. 

Moving forward, we introduce the posterior mean under additive white Gaussian noise observation $\eta_u(\rv;\tau) \doteq \mathbb E[\matr\sv_u\vert \matr \rv=\matr\sv_u+\sqrt{\tau}\zv]$ with $\zv\sim\mathcal{CN}(\matr  0,\Id_{N_u})$ independent of $\sv_u$. 
Let $\eta_{u}'$ denote the $N_u\times N_u$ Jacobian matrix of $\eta_u$. Then, means and variances under $p_0 q_1$ can be written as
\begin{subequations}
\label{second-EC}
\begin{align}
\mathbb E[\sv_u]_{p_0\cdot q_1} &= \eta_u(\rv_u;\tau_u)  \label{meanp0q1} \\
\frac{1}{N_u}\mathbb E[\Vert\sv_u-\mathbb E[\sv_u]_{p_0\cdot q_1}\Vert^2]_{p_0\cdot q_1}& %= \frac{1}{N_u}\mathbb E[\norm{\sv_u-\mathbb E[\sv_u\vert \rv_u]}^2\mid \rv_u=\matr\sv_u+\sqrt{\tau_u}\zv ]
=\tau_{u}\langle \eta_{u}'(\rv_u;\tau_u) \rangle\;.\label{variancep0q1}
\end{align}
\end{subequations}
 Here, the latter step follows from the standard
``linear-response'' argument~\cite{opper2001adaptive} as
%: define the partition function $Z(\boldsymbol{\gamma})\doteq \int p_0(\mathbf{s})\,{\rm e}^{-\rho\|\mathbf{s}\|^2+2{\rm Re}(\mathbf{s}^\dagger \boldsymbol{\gamma})}\,{\rm d}\mathbf{s}$; then, we have
\begin{align}
%\frac{\partial \ln Z(\boldsymbol{\gamma})}{\partial \boldsymbol{\gamma}^*}
 %   &= \mathbb{E}[\mathbf{s}\mid \boldsymbol{\gamma} = \sqrt{\rho}\,\mathbf{s} + \mathbf{z}]
  %  \label{eq:wirtinger1}\\
\frac{\partial}{\partial \boldsymbol{\rv}}
    \mathbb{E}[\mathbf{s}_u\mid \boldsymbol{\rv} = \,\mathbf{s}_u +\sqrt{\tau} \mathbf{z}]
    &= \tau\mathbb{E}\!\left[
        \bigl(\mathbf{s}_u - \mathbb{E}[\mathbf{s}_u \mid \boldsymbol{\rv}]\bigr)
        \bigl(\mathbf{s}_u - \mathbb{E}[\mathbf{s}_u \mid \boldsymbol{\rv}]\bigr)^{\dagger}
    \,\middle|\,
        \boldsymbol{\rv} = \,\mathbf{s}_u + \sqrt{\tau}\mathbf{z}
    \right].\label{eq:wirtinger2}
\end{align} 

Using \eqref{EPrule} in the second EC equality in \eqref{EC} we obtain
\begin{subequations}
\label{conditionp0q1}
 \begin{align}
\frac{1}{\lambda_u}&=\frac{1}{\tau_{u}\langle \eta_{u}'(\rv_u;\tau_u) \rangle}-\frac{1}{\tau_u}\\
\frac{\fv_u}{\lambda_u}&=\frac{\eta_u(\rv_u;\tau_u)}{\tau_{u}\langle \eta_{u}'(\rv_u;\tau_u) \rangle}-\frac{\rv_u}{\tau_u}\;.
\end{align}
\end{subequations}
The complete set of EC equations is given by \eqref{conditionq0p1}, \eqref{q0p1explicit}, and \eqref{conditionp0q1}. 

\subsection{Iterative solution of the fixed-point equations}

The basic MU-OAMP algorithm is obtained as the iterative solution of the  EC seen as a system of fixed-point equations.  
Starting with initializations $\{\fv_u^{(1)}\}$ and $\{\lambda_u^{(1)}\}$, 
the algorithm consists of the following vector updates for $t=1,2,\cdots$ 
\begin{subequations}
\label{algv}
 \begin{align}
\matr\Sigma^{(t)} &= \left( \sigma^2\Id_{L} + \sum_{u} \lambda_{u}^{(t)} \Am_u \Am_u^\dagger\right)^{-1}\label{filter} \\
\rv_u^{(t)} &=\fv_u^{(t)}+\frac{\Am_u^\dagger\matr\Sigma^{(t)}}{\langle \Am_u^\dagger \matr\Sigma^{(t)} \Am_u \rangle} \left( \yv - \sum_{u'}\Am_{u'}\fv_{u'}^{(t)}\right), \\
\fv_u^{(t+1)}&=\frac{\tau_u^{(t)}\eta_u(\rv_u^{(t)};\tau_u^{(t)})-\nu_u^{(t)}\rv_u^{(t)}}{\tau_u^{(t)}-\nu_u^{(t)}},
\label{ziopirillo}
\end{align}   
\end{subequations}
and scalar updates
\begin{subequations}
\label{algs}
\begin{align}
\tau_{u}^{(t)}&=F_u(\matr \lambda^{(t)})\\
\nu_u^{(t)}&=\tau_{u}^{(t)}\langle \eta_{u}'(\rv_u^{(t)};\tau_u^{(t)}) \rangle\label{empirical}\\
\lambda_{u}^{(t+1)}&=\frac{\tau_{u}^{(t)}\nu_{u}^{(t)}}{\tau_{u}^{(t)}-\nu_{u}^{(t)}} \; ,
\end{align}
\end{subequations}
where $\boldsymbol{\lambda}^{(t)}\doteq(\lambda_1^{(t)},\lambda_2^{(t)},\ldots,\lambda_U^{(t)})$
and where we define, for $\mathbf{x}\in[0,\infty)^U$,
\begin{align}
    F_{u}(\mathbf{x})
    &\doteq \frac{1}{\left\langle\mathbf{A}_u^{\dagger}
        \left(\sigma^2\mathbf{I}
            + \sum_{u'} x_{u'}\mathbf{A}_{u'}\mathbf{A}_{u'}^{\dagger}
        \right)^{-1}
        \mathbf{A}_u\right\rangle} - x_u\;. \label{fNu}
\end{align}

%%%%%%%%%%%%%%%%%%%%%%%%%%%%%%%%%%%%%%%%%%%%%%%%%%%%%%%%%%%%%%%%%%%%%%%%%
\section{The MU-OAMP Framework}\label{main_result}

In the previous section, we derived a fixed-point algorithm from the EC
moment-matching rationale under the block-constant restriction of the parameters (see \eqref{unified}). 
In general, the posterior-mean {\em denoiser}
$\eta_u(\mathbf{r};\tau)=\mathbb{E}\!\left[\mathbf{s}_u \mid \mathbf{r} =
\mathbf{s}_u + \sqrt{\tau}\,\mathbf{z}\right]$ may be difficult to compute. 
Moreover, the matrix inversion involved in $\boldsymbol{\Sigma}^{(t)}$ in~\eqref{filter} may constitute 
a computational bottleneck. To address these issues, we propose a generalized algorithm (whose motivation will follow from the subsequent high-dimensional analysis) in which
\begin{itemize}
    \item[(i)]  the posterior-mean denoiser $\eta_u(\mathbf{r}, \tau)$ is replaced by
                a generic sufficiently well-behaved denoiser function;
    \item[(ii)] the matrix sequence (in $t$) $\boldsymbol{\Sigma}^{(t)}$ defined
                in~\eqref{filter} is replaced by a generic positive-definite matrix sequence
$\matr{\Sigma}^{(t)}\geq \matr 0$  independent of the Haar unitaries $\{\Om_u\}$ (see \eqref{svd}).
\end{itemize}
Specifically, our proposed algorithm is initialized with $\fv_u^{(1)} = \mathbb{E}[\sv_u]$ for each $u =1, \ldots, U$, and proceeds for $t = 1,2,\dots,T$ according to
\begin{subequations}\label{oamp0}
\begin{align}
\rv_u^{(t)} &= \fv_u^{(t)}+\frac{\Am_u^\dagger \matr{\Sigma}^{(t)}}{\langle \Am_u^\dagger \matr{\Sigma}^{(t)} \Am_u\rangle}
\left(\yv - \sum_{u'} \Am_{u'} \fv_{u'}^{(t)}\right) \\
\fv_u^{(t+1)} &= f_{u,t}(\rv_u^{(t)}), \label{sucamillo}
\end{align}
\end{subequations}
where we define
\[
f_{u,t}(\rv) \doteq 
\frac{\tau_u^{(t)}\,\eta_u(\rv;\tau_u^{(t)}) - \nu_u^{(t)}\,\rv}
{\max(\epsilon_u,\tau_u^{(t)} - \nu_u^{(t)})}
\]
and where $\eta_u(\cdot;\cdot): \mathbb{C}^{N_u}\times \RR_{+} \to \mathbb{C}^{N_u}$ is a general denoiser, 
\begin{align*}
\epsilon_u\doteq\frac{\sigma^2}{\langle\Am_u^\dagger\Am_u\rangle} \frac{\sigma^2}{\sigma^2+\lambda_u^{(1)}\langle\Am_u^\dagger\Am_u\rangle}\;.
\end{align*}
In particular, we note that if $\eta_u$ is the posterior-mean denoiser, we have (see \eqref{epsilonbound})\vspace{-0.2cm}
\[\max(\epsilon_u,\tau_u^{(t)} - \nu_u^{(t)})=\tau_u^{(t)} - \nu_u^{(t)}\;\vspace{-0.2cm}\]
and \eqref{sucamillo} reduces to \eqref{ziopirillo}. 

For the scalar updates, we note that \eqref{algs} contains 
the empirical average of the posterior variances along the $u$-th vector (see \eqref{variancep0q1}). To obtain "smoother" (without extra concentration assumptions) finite-sample behavior of the algorithm dynamics, we replace such empirical average with a suitable ensemble expectation. The resulting recursive scalar update 
for arbitrary denoiser functions $\eta_u$ and $\matr{\Sigma}^{(t)}$ is referred to as state evolution (SE):
\begin{definition}[SE]\label{SE_new}
For $u = 1, \ldots, U$, let $\lambda_{u}^{(1)}=\frac{1}{N_u}\mathbb E[\norm{\sv_u-\mathbb E[\sv_u]}^2]$. 
Then, conditioned on $\{\Am_u\Am_u^\dagger\}$, the SE is recursively constructed for $t=1,2,...$ as
\begin{subequations}\label{SE_rec}
 \begin{align}
\Rm^{(t)}&=\sigma^2\Id_L+\sum_u\lambda_u^{(t)}\Am_u\Am_u^\dagger\label{rrm}\\
\tau_u^{(t)}&=\frac{\langle\Am_u^\dagger\matr\Sigma^{(t)}\Rm^{(t)}\matr\Sigma^{(t)}\Am_u\rangle}{\langle \Am_u^\dagger\matr\Sigma^{(t)}\Am_u \rangle^2}
-\lambda_u^{(t)} \label{taut}\\
\nu_{u}^{(t)}&=\tau_{u}^{(t)}\mathbb E[\langle \eta'_{u}(\sv_u+\sqrt{\tau_{u}^{(t)}}\zv;\tau_u^{(t)})\rangle]\\
\lambda_u^{(t+1)}&= \frac{1}{N_u}\mathbb E[\norm{\sv_u-f_{u,t}(\sv_u+\sqrt{\tau_u^{(t)}}\zv)}^2],
\end{align}  
\end{subequations}
where expectations are w.r.t. the mutually independent random vectors $\sv_u \sim p_0$, $\zv \sim \Cc\Nc(0,\Id_{N_u})$ and the expectations are taken conditioned on $\{\Am_u\Am_u^\dagger\}$. 
\end{definition}

%Note that the SE recursion does not depend on the Haar matrices $\{\Om_u\}$~in~\eqref{svd}. 
Next, we verify that by restricting the general dynamics in~\eqref{oamp0} to the
posterior-mean denoiser $\eta_u(\mathbf{r};\tau)$ and choosing
$\boldsymbol{\Sigma}^{(t)}$ as in~\eqref{filter}, we recover
\eqref{algv} and \eqref{algs} with the only difference that 
$\mathbb{E}[\langle \eta'_u(\mathbf{s}_u+\sqrt{\tau_u^{(t)}}\,\mathbf{z};\tau_u^{(t)})\rangle]$
appears instead of its algorithmic counterpart
$\langle \eta'_u(\mathbf{r}_u^{(t)};\tau_u^{(t)})\rangle$.  We refer to this case as the {\em Bayesian-optimal} setting, because in Section~\ref{replica_sect} we shall show, subject to the vaility of the replica ansatz, that for sufficiently large $1 \ll t \ll N_u$, this construction yields the \emph{Bayesian optimal solution} in the so-called algorithmic region \cite{barbier2020mutual}, where the RS ansatz admits a unique fixed point.

\begin{remark}[Bayesian Optimal Setting] \label{Bayesian_optimal}
Let us consider the choice
\begin{equation}\label{filter_lmmse}
    \boldsymbol{\Sigma}^{(t)}=
    \left(\sigma^{2}\mathbf{I}_{L}
        + \sum_u \lambda_{u}^{(t)}\,
          \mathbf{A}_{u}\mathbf{A}_{u}^{\dagger}
    \right)^{-1},
\end{equation}
where $\lambda_{u}^{(t)}$ are defined in the SE recursion \eqref{SE_rec}. Then, the update
rule in~\eqref{taut} simplifies to
\begin{equation}
    \tau_{u}^{(t)} = F_{u}(\boldsymbol{\lambda}^{(t)})\;.
\end{equation}
Moreover, let us set $\eta_{u}(\mathbf{r};\tau)=
\mathbb{E}[\mathbf{s}_u\mid \rv = \mathbf{s}_u+\sqrt{\tau}\,\mathbf{z}]$, so that,
conditioned on $\tau_u^{(t)}$, we have
\begin{align}
    \nu_u^{(t)}
    &= \frac{1}{N_u}\,\mathrm{mmse}
       \!\left(\mathbf{s}_u \mid \mathbf{s}_u + \sqrt{\tau_u^{(t)}}\,\mathbf{z}\right).
\end{align}
We also note that
\begin{align}
    \tau_u^{(t)}-\nu_u^{(t)}
    &\overset{(a)}{\geq}
     \frac{\bigl(\tau_u^{(t)}\bigr)^2}{\tau_u^{(t)}+\lambda_u^{(1)}}\overset{(b)}{\geq}
     \frac{\bigl(F_{u}(\mathbf{0})\bigr)^2}{F_{u}(\mathbf{0})+\lambda_u^{(1)}}
     \equiv \epsilon_u\;,\label{epsilonbound}
\end{align}
where (a) follows directly from~\cite[Lemma~2]{bustin2012mmse}, and
(b) follows from the fact that the mapping $F_u$ is monotone in the sense that $F_{u}(\mathbf{x}) \leq
F_{u}(\mathbf{x}')$ for any $\mathbf{0} \leq \mathbf{x} \leq \mathbf{x}'$
(i.e., $0 \leq x_u \leq x_u'$ for all $u$) (see
Appendix~\ref{proof_rem_contradiction}) and thus
\begin{equation}
\tau_u^{(t)} \geq F_{u}(\mathbf{0})\;.\label{tbound}
\end{equation}
In summary, with these choices of the denoiser $\eta_u$ and the matrix filter $\Sigmam^{(t)}$ the SE in~\eqref{SE_rec} simplifies to
\begin{subequations}\label{SE_optimal}
\begin{align}
    \tau_{u}^{(t)}
    &= F_{u}(\boldsymbol{\lambda}^{(t)})\\
    \nu_u^{(t)}
    &= \frac{1}{N_u}\,\mathrm{mmse}
       \!\left(\mathbf{s}_u \mid \mathbf{s}_u
               + \sqrt{\tau_u^{(t)}}\,\mathbf{z}\right)\\
    \lambda_{u}^{(t+1)}
    &= \frac{\tau_{u}^{(t)}\,\nu_{u}^{(t)}}
            {\tau_{u}^{(t)}-\nu_{u}^{(t)}}\;.
\end{align}
\end{subequations}
\end{remark} 

In passing we note that, in the case of \emph{approximate} statistical symmetry among the users, so that 
$\lambda_{u}^{(t)} \approx \lambda_{u'}^{(t)}$ for all $u\neq u'$, 
we may consider the (conventional) OAMP-type filtering
\begin{equation}
\label{filter_oamp}
\matr{\Sigma}^{(t)} =
   \big(\sigma^2\Id_L+{\lambda}^{(t)}\Am\Am^\dagger\big)^{-1} \quad {\lambda}^{(t)} \doteq \frac{1}{U}\sum_{u'\leq U}\lambda_{u'}^{(t)}\;
\end{equation}
where $\Am =[\Am_1, \Am_2, \ldots, \Am_U]$ and the terms $\lambda_u^{(t)}$ are computed from the SE in Definition~\ref{SE_new}. Thus, the complexity of matrix inversion at every iteration
can be avoided by precomputing, offline, the spectral decomposition of $\Am\Am^\dagger$. 
However, this filtering choice may yield significantly degraded performance as the statistical 
asymmetry among users increases. Nevertheless, the high-dimensional analysis of the algorithm remains valid for any choice of the filter $\boldsymbol{\Sigma}^{(t)}$. 

%%%%%%%%%%%%%%%%%%%%%%%%%%%%%%%%%%%%%%%%%%%%%%%%%%%%%%%%%%%%%%%
\subsection{The High-Dimensional Analysis}
\label{o1_notation}
We analyze the dynamics of the proposed algorithm in terms of the following notion of concentration inequalities with $\mathcal L_p$ norms  $\Vert \cdot\Vert_{\mathcal L^p}\eqdef (\mathbb E\Vert \cdot\Vert^p)^\frac {1}{p}$ in \cite{cakmakit25}:  For sequences (in $L$ or $N_u$) of a RV $\matr\delta\in \CC^{d}$ where $d$ is arbitrary (e.g. $d=1$, $d=N_u$, etc.) we write 
	\begin{equation}
		\matr\delta=\Op{1} ~~\text{if}~~\Vert \matr\delta\Vert_{\mathcal L^p}\leq C_p
        \label{opnotation} 
	\end{equation}
for all $p\in \NN$ and some constants $C_p$.  Furthermore, we write for any deterministic $\kappa>0$ (e.g. $\kappa=L$, $\kappa=1$, $\kappa=1/\sqrt{N_u}$, etc.)
\begin{align}
\matr\delta=\Op{\kappa} \quad \text{if}\quad \frac{1}{\kappa}{\matr\delta}=\Op{1}.
\end{align}

Since the notion of $\Op{1}$ is not very common, we highlight several of its key aspects. While sub-Gaussian or sub-exponential random variables \cite{vershynin2018high}  are commonly used to derive strong non-asymptotic bounds \cite{rush2018,rush24}, the $\Op{1}$ notion allows for a useful relaxation that includes a broader class of distributions exhibiting heavy exponential tails. For instance, if $A$ is a sub-Gaussian random variable, then $A^C = \Op{1}$ for any (large) constant $C$. Distributions of this type fall into the class of sub-Weibull distributions; see \cite{vladimirova2020sub} for a detailed discussion. Furthermore, we note that the notion of $\Op{1}$ is (slightly) stronger than the so-called \emph{stochastic domination} $\mathcal O_\prec(1)$ \cite[Definition 6.4]{erdHos2017dynamical} (see also the equivalence given in \cite[Lemma 10.1]{erdHos2017dynamical}), which is now a widely used notion in random matrix theory.

\begin{assumption}[Model Assumptions]\label{as1}
Consider the observation model in \eqref{eq:system_model}. 
\begin{itemize}
\item [(i)]  The scaling parameters are $L,N_1,N_2,\cdots,N_U$ while $L/N_{u}=\Op{1}$ for all $u$. The number of users $U$, the number of iterations $T$, and the noise variance $\sigma^2$ are independent of $L$. 
\item[(ii)] Each \( \Om_u \) in \eqref{svd} is arbitrary and drawn from a Haar unitary ensemble. 
\item [(iii)] Let $\frac{1}{\sqrt N_u}\norm{\sv_u}=\Op{1}$,  $\Vert\Am_u\Vert_{2}=\Op{1}$, $\|\matr\Sigma^{(t)}\|_2 = \Op{1}$, $(\langle\Am_u^\dagger\matr\Sigma^{(t)}\Am_u \rangle)^{-1}=\Op{1}$, and $(\tau_u^{(t)})^{-1}=\Op{1}$ for all $(u,t)\in[U]\times [T]$  with $\Vert (\cdot) \Vert_2$ denoting the largest singular value of the matrix in the argument.
\end{itemize}
\end{assumption}
We consider a generic denoiser mapping and a general filtering matrix $\matr \Sigma^{(t)}$. Although Assumption (iii) may appear somewhat implicit in this general setting, the following remarks demonstrate that it can easily be satisfied in practical scenarios.

\begin{itemize}
    \item For channel matrices with $\|\Hm_u\|_2 = \Op{1}$ and  
by using random precoding with $\|\Xim_u\|_2 = \Op{1}$, the submultiplicativity 
of the spectral norm, together with the product rule for $\Op{1}$ (i.e., $XY = \Op{1}$ if 
$X = \Op{1}$ and $Y = \Op{1}$), implies $\|\Am_u\|_2 \leq \|\Hm_u\|_2 \, \|\Xim_u\|_2 = \Op{1}$.

    \item Let $\matr\Sigma^{(t)} $ be of the form $\matr\Sigma^{(t)} = (\sigma^2 \Id + \Xm)^{-1}$, where $\Xm \geq \matr 0$. Then, 
    $\|\matr\Sigma^{(t)}\|_2 \leq 1/{\sigma^2}$ Moreover, we have ${\langle \Am_u^\dagger \matr\Sigma^{(t)} \Am_u  \rangle}\geq (\sigma^2+\norm{\Xm}_2)^{-1}\langle\Am_u^\dagger\Am_u\rangle $.
 Thus, given $1/\langle \Am_u^\dagger \Am_u \rangle = \Op{1}$ and $\|\Xm\|_2 = \Op{1}$, we have $1/\langle  \Am_u^\dagger \matr\Sigma^{(t)} \Am_u\rangle=\Op{1}$.

 \item Since $\tau_u^{(t)}$ represents the noise variance in the decoupled input-output
$(\sv_u, \sv_u + \sqrt{\tau_u^{(t)}}\,\zv)$ system, we bound $1/\tau_u^{(t)}$ by restricting 
attention to the Bayesian-optimal setting in Remark~\ref{Bayesian_optimal}. Then, we get from \eqref{tbound} that $\frac{1}{\tau_u^{(t)}} \leq  
\frac{\langle \Am_u^\dagger \Am_u \rangle}{\sigma^2}=\Op{1}$\;.
\end{itemize}
For readability, we defer additional technical assumptions (Assumption~\ref{asm_psd} and \ref{add_asm}) to the next subsection.

We now present our main result: an explicit high-dimensional decomposition of $\rv_u^{(t)}$ in terms of the true signal $\sv_u$
and an additive Gaussian process, whose covariance structure is governed by the two-time generalization of the SE:

\begin{definition}[Two-time SE]\label{defSE}
Conditioned on $\{\Am_u\Am_u^\dagger\}$, we recursively construct $\mathcal C_u^{(t,s)}$ for $1\leq s,t\leq T$ accordingly,
\begin{subequations}
\begin{align}
\Rm^{(t,s)}&=\sigma^2\Id_L+\sum_u\widetilde{\mathcal C}_{u}^{(t,s)}\Am_u\Am_u^\dagger
\\
\mathcal C_{{u}}^{(t,s)}&=\frac{{\langle \Am_u^\dagger\Sigmam^{(s)} \Rm^{(t,s)}\Sigmam^{(t)}\Am_u\rangle}}{\langle \Am_u^\dagger \matr\Sigma^{(t)}\Am_u \rangle\langle \Am_u^\dagger \matr\Sigma^{(s)}\Am_u \rangle}
-\widetilde{\mathcal C}_{u}^{(t,s)}\\
 \widetilde{\mathcal C}_{u}^{(t+1,s+1)}&=\mathbb E\left[\left\langle (f_{u,t}(\sv_u + \matr\phi_u^{(t)})-\sv_u),(f_{u,s}(\sv_u + \matr\phi_u^{(s)})-\sv_u)\right\rangle\right]\label{Cuts}
\end{align}
\end{subequations}
where expectations are taken over $\sv_u$ and the RVs
\[[\phiv_u^{(t)},\phiv_u^{(s)}]=[\zv,\zv']\left[ \begin{array}{cc}
		{\mathcal C}_u^{(t,t)}&{\mathcal C}_u^{(t,s)}\\
		{\mathcal C}_u^{(s,t)}& {\mathcal C}_u^{(s,s)}
	\end{array}\right]^{\frac 1 2}  \]
and $\zv,\zv'\sim\mathcal{CN}(\matr 0,\Id_{N_u})$ are independent and are also independent of $\sv_u$. Moreover, we set $ \widetilde{\mathcal C}_{u}^{(1,1)} =\lambda_u^{(1)}$ and for $t>0$
\[\widetilde{\mathcal C}_{u}^{(1,t+1)} = \mathbb{E}\left[\left\langle (\mathbb{E}[\sv_u] - \sv_u),(f_{u,t}(\sv_u + \matr{\phi}_u^{(t)}) - \sv_u) \right\rangle\right]\;.\]
Note that, by construction we have $\mathcal C_u^{(t,t)}=\tau_{u}^{(t)}$ and $\widetilde{\mathcal C}_u^{(t,t)}=\lambda_u^{(t)}$. For further intuition on how the two-time SE construction arises, we refer to Definition~\ref{def-effective} in Section~\ref{Pth1}.
\end{definition}

In passing, we also note that while the two-time-type general SE formalism is not commonly encountered in the literature, it is particularly useful---and, in fact, necessary---for analyzing contraction properties of the dynamics, such as $\norm{\fv_u^{(t)} - \fv_u^{(s)}}$, as well as related quantities. For further discussion and applications, we refer interested readers to~\cite{Bayati11,Berthier20,loureiro2021learning,gerbelot2022asymptotic,Cakmakisit24}.

\begin{theorem}[Main Result]\label{th1}
Let Assumptions~\ref{as1}--\ref{add_asm} hold. Let the denoisers
$
\eta_{u}(\mathbf{r};\tau):\mathbb{C}^{N_u}\times \mathbb{R}_{+}\to \mathbb{C}^{N_u}$
be differentiable and uniformly Lipschitz (w.r.t. $\mathbf{r}$), with Lipschitz constant $\mathrm{Lip} = O(\tau^{\pm p})$ and $\frac{1}{\sqrt{N_u}}\norm{\eta_u(\matr 0,\tau)}= O(\tau^{\pm p})$ for any constant $p \geq 0$. Then, we prove in Section~\ref{Pth1} that
\begin{align}
\mathbf{r}_{u}^{(t)} = \mathbf{s}_u + \boldsymbol{\phi}_{u}^{(t)} + \Op{1}
\label{fsres}
\quad (u,t)\in[U]\times[T].
\end{align}
Here, each $\{ \boldsymbol{\phi}_{u}^{(t)}\}_{t\leq T}$ is independent of $\{\sv_u\}$. Furthermore, for each $u = 1,\ldots, U$, there exist an independent Gaussian random matrix $\Zm_u\in\CC^{N_u\times T}$ with $(Z_u)_{ij}\sim_{\text{i.i.d.}}\mathcal{CN}(0,1)$ such that
\begin{equation}
[\phiv_u^{(1)},\phiv_u^{(t)},\cdots,\phiv_u^{(T)}]=\Zm_u (\mathcal C_u^{(1:T)})^{\frac{1}{2}}\;.
\end{equation}
Here, $\mathcal{C}_u^{(1:T)}$ denotes the $T \times T$ matrix whose $(t,s)$-indexed entries $\mathcal{C}_u^{(t,s)}$ are given as in \eqref{Cuts}. 
%Hence, conditioned on $\mathcal C_u^{(1:T)}$, the dynamics $\{\matr{\phi}_{u}^{(t)}\}_{t \in [T]}$ is a discrete-time zero-mean Gaussian process with i.i.d. components ${\matr\phi}_{u}^{(t)}\sim_{\text{i.i.d.}}\Phi_{u}^{(t)}$ with $\mathbb E[(\Phi_{u}^{(t)})^*\Phi_{u}^{(s)}]=\mathcal C_{u}^{(t,s)}$. 
\end{theorem}
The proof of Theorem \ref{th1} is deferred to Section \ref{Pth1}. 
We note that the result \eqref{fsres} is non-asymptotic and, in particular, it has the following asymptotic implication: for any \emph{small} $\epsilon > 0$ (independent of $L$) 
\begin{equation}
L^{-\epsilon}\Vert \rv_{u}^{(t)}-(\sv_u+\matr\phi_{u}^{(t)}) \Vert  \overset{\text{a.s.}}{\rightarrow} 0 \quad \text{for}\quad L\to\infty\;.\label{ascon}
\end{equation}  
Indeed, the definition of $\mathcal{O}(L^{-\epsilon})$ and Markov's inequality yield $\mathbb{P}(|\mathcal{O}(L^{-\epsilon})| \geq \epsilon') \leq \frac{C_p^p}{(\epsilon')^p} L^{-\frac{\epsilon p}{2}}$ for any $\epsilon'> 0$ and $p \in \mathbb{N}$ and choosing $p>2/\epsilon$ leads from Borel-Cantelli's lemma to \eqref{ascon}.

\begin{corollary}[Mean Square Error]\label{cor_mse}
Let the premises of Theorem~\ref{th1} hold. Then, for any $(u,t)\in[U]\times[T]$, we have
\begin{equation}
\frac{1}{N_u}\,\lVert \sv_u-\eta_{u}(\rv_u^{(t)};\tau_u^{(t)})
\rVert^2
=
\frac{1}{N_u}\,\mathbb{E}
[\lVert \sv_u-\eta_{u}(\sv_u+\sqrt{\tau_{u}^{(t)}}\,\zv;\tau_u^{(t)})\rVert^{2}]
+ \Op{L^{-\frac{1}{2}}}.
\end{equation}
where the expectation is taken conditioned on $\tau_u^{(t)}$ and $\zv\sim\mathcal{CN}(\matr 0,\Id_{N_u})$ is  independent of $\sv_u$.
\end{corollary}

\subsection{Additional Assumptions}
Next, we present the additional assumptions (i.e., Assumption~\ref{asm_psd} and Assumption~\ref{add_asm}) as the premises of Theorem~\ref{th1}. We will also mention how to bypass and relax these assumptions to have the following \emph{asymptotic} implication of Theorem~\ref{th1}: 
\begin{align}
\rv_{u}^{(t)}&\simeq \sv_u+\matr\phi_{u}^{(t)} \;.\label{asymres}
\end{align}
Here, for the sequences in $L$ of RVs $\av,\widehat\av \in \CC^{d}$, we write 
\begin{equation}
 \widehat \av\simeq \av \quad\text{if}\quad \lim_{N\to\infty}\frac{1}{\sqrt d}\Vert  \widehat\av-\av \Vert\overset{a.s.}{=} 0 \label{simeq}\;.
\end{equation}
Furthermore, we will denote a.s. bounded sequence by
\begin{equation}
A_{L}=\mathcal O_\infty(1)\quad\text{if}\quad\limsup_{N\to\infty}\vert A_L\vert\overset{a.s.}{<}\infty \;.\label{O1}
\end{equation}

\subsubsection{Positive definiteness of the two-time SE matrix}
By construction, we have (see Definition~\ref{def-effective})
\begin{equation}
{\mathcal C}_u^{(1:T)},\widetilde{\mathcal C}_u^{(1:T)}\geq \matr 0.\nonumber 
\end{equation}
\begin{assumption}\label{asm_psd}
There exists $\epsilon>0$ (independent of $L$) such that 
 $\widetilde{\mathcal C}_{u}^{(1:T)}>\epsilon\Id_T$.
\end{assumption}
The strict positive definiteness assumption is common in the finite-sample analysis of AMP/OAMP \cite{rush2018,rush24,reeves2025dimension}. In particular, such a condition is also used as a stopping criterion for the corresponding AMP/OAMP algorithms \cite{rush2018finite,rush24}. Next, we emphasize that, by leveraging the SE-perturbation idea \cite{javanmard2013state,Berthier20}, Assumption \ref{asm_psd} is not required for establishing the asymptotic decoupling principle \eqref{asymres}.
\begin{remark}\label{relax1}
If the notion of $\Op{1}$ in Assumption~\ref{as1}-(iii) is replaced by $\mathcal O_\infty(1)$, then Assumption \ref{add_asm} is no longer needed for the asymptotic decoupling principle in \eqref{asymres}. (See Appendix~\ref{relaxition}). 
\end{remark}

\subsubsection{Concentration Assumptions}
\label{sec:add_assumptions}
\begin{assumption}\label{add_asm}
For each triple $(u,t,s)$ we assume the following concentration properties: For  two jointly Gaussian RVs $\phiv_1,\phiv_2\in \CC^{N_u}$ independent of $\sv_u$ with $[\phiv_1^\top, \phiv_2^\top]^\top \sim \mathcal{CN}(\zerov,\mathcal{C}\otimes \Id_{N_u})$ and some arbitrary covariance matrix $\mathcal{C}\in \CC^{2 \times 2}$ we have
\begin{subequations}
\label{ass2}
\begin{align}
\langle \sv_u,\sv_u \rangle&=\mathbb E[\langle \sv_u,\sv_u\rangle]+ \Op{L^{-\frac 1 2}}\\
\mathbb{E}[\langle\eta_{u,t}(\sv_u+\phiv_1), \eta_{u,s}(\sv_u+\phiv_2)\rangle\vert \sv_u]&= \mathbb{E}[\langle\eta_{u,t}(\sv_u+\phiv_1), \eta_{u,s}(\sv_u+\phiv_2)\rangle]+\Op{L^{-\frac 1 2}}\\
\mathbb{E}[\langle \phiv_1, \eta_{u,t}(\sv_u+\phiv_1)\rangle\vert \sv_u]\
&=\mathbb{E}[\langle\phiv_1, \eta_{u,t}(\sv_u+\phiv_1)\rangle]
+ \Op{L^{-\frac 1 2}}.\\
\mathbb{E}[\langle\sv_u, \eta_{u,t}(\sv_u+\phiv_1)\rangle\vert \sv_u]&=\mathbb{E}[\langle\sv_u, \eta_{u,t}(\sv_u+\phiv_1)\rangle]
+ \Op{L^{-\frac 1 2}}.
\end{align}
where for brevity we set $\eta_{u,t}(\rv) \doteq \eta_{u}(\rv; \tau_u^{(t)})$, and where the expectations are taken implicitly conditioned on $\{\Am_u\Am_u^\dagger\}$.
\end{subequations}
\end{assumption}

\begin{remark}\label{relax2}
We should point out that, for the asymptotic decoupling principle in \eqref{asymres} 
to hold, Assumption~\ref{add_asm} can be relaxed.
In particular, in Appendix~\ref{relaxition} we show that if the notion of $\Op{1}$ in Assumption~\ref{as1}-(iii) is replaced by $\mathcal O_\infty(1)$, then 
concentration statements of the form $\widehat a = a + \Op{L^{-\frac{1}{2}}}$
in Assumption~\ref{add_asm} can be relaxed to almost sure limits as $\widehat a \simeq a$. For example, $
\langle \sv_u,\sv_u \rangle \simeq \mathbb E[\langle \sv_u,\sv_u\rangle] $ and similarly for the other relations in~\eqref{ass2}. 
\end{remark}

In the following, we consider several relevant scenarios in which Assumption~\ref{add_asm} is valid. 
For convenience, we fix an arbitrary triple $(u,t,s)\in[U]\times[T]\times [T]$ and omit the indices $(u,t,s)$, i.e., we write $\eta,\sv$ instead of $\eta_{u,t},\sv_u$, etc. 

\textbf{Uncoded Systems:}
Consider first the uncoded setting in which the entries of $\sv$ are drawn i.i.d.\ from a distribution $S$, i.e., $\sv\sim_{\text{i.i.d.}}S$ with $S=\Op{1}$.
In this case, any meaningful choice of the denoiser is separable:
\[
[\eta(\rv)]_i \equiv \eta(r_i), \qquad \forall i .
\]
The concentration results in \eqref{ass2} then follow from Lemma~\ref{conprod} in Appendix~\ref{preliminariesop}.

\textbf{Coded Systems:}
For coded systems, our analysis holds for \emph{geometrically uniform} codes~\cite{forney2002geometrically}. These are signal sets $\mathcal{C}\subset \CC^N$ that are generated as the orbit of any codeword $\ev \in \Cc$ under the action of the code symmetry group $\Gamma(\Cc)$, i.e., for which there exist a group of unitary matrices $\Gamma(\Cc) = \{\Pm_\sv\}$ such that any codeword $\sv \in \Cc$ is obtained as
$\sv = \Pm_{\sv}\,\ev$.
We require also that the denoiser is chosen to be {\em consistent} with the code geometric uniformity, 
in the sense that
\begin{equation}
\Pm_{\sv}\,\eta(\rv)
\equiv
\eta(\Pm_{\sv}\rv).
\label{per_pro}
\end{equation}
Since $\Pm_\sv$ is unitary we have $(\Pm_\sv\phiv_1,\Pm_\sv\phiv_2)\sim (\phiv_1,\phiv_2)$ and \eqref{per_pro} implies
\begin{subequations}
\label{equivalence}
\begin{align}
\langle \sv,\sv \rangle&=\langle \ev,\ev\rangle\\
\mathbb{E}[\langle\eta(\sv+\phiv_1), \eta(\sv+\phiv_2)\rangle\vert \sv]&= \mathbb{E}[\langle\eta(\ev+\phiv_1), \eta(\ev+\phiv_2)\rangle]\\
\mathbb{E}[\langle \phiv_1, \eta(\sv+\phiv_1)\rangle\vert \sv]
&=\mathbb{E}[\langle\phiv_1, \eta(\ev+\phiv_1)\rangle]\;\\
\mathbb{E}[\langle \sv, \eta(\sv+\phiv_1)\rangle\vert \sv]
&=\mathbb{E}[\langle\ev, \eta(\ev+\phiv_1)\rangle].
\end{align}
\end{subequations}
In other words, the analysis can be restricted to a fixed codeword (typically the codeword corresponding ot the all-zero message). Many codes of interest are geometrically uniform, see~\cite{forney2002geometrically}. 
Next, we address two relevant coding strategies that fulfill \eqref{per_pro} with permutation matrices. 

\textbf{(i) Sparse Regression Codes (SPARCs):}
In SPARCs~\cite{Rush17}, the vector $\sv\in\RR^{N=BQ}$ is divided into $B$ sections of dimension $Q$, and in each section $\sv[b]$ exactly one entry is nonzero, with its position being determined by an input bit sequence, which is chosen uniformly at random. Here,  for any vector $\sv\in\CC^{BQ}$ we write $\sv = [\sv[1]^\top,\sv[2]^\top,\ldots,\sv[B]^\top]^\top $ with each section $\sv[b]\in\CC^{Q}$.
Accordingly, $\{\Pm_{\sv}\}$ is a section-wise Cartesian product of all $Q$-shift matrices (yielding all possible $Q^B$ codewords). The posterior-mean denoiser is section-wise separable as 
\begin{align}
\eta(\rv)[b]&\equiv\mathbb{E}\!\left[\sv[b] \,\middle|\, \rv[b] = \sv[b] + \sqrt{\tau}\,\zv \right] \nonumber \\
&= \sqrt{BP_b} \text{softmax}\!\left(\frac{\sqrt{BP_b}}{\tau}\,\rv[b]\right), \label{eta1}
\end{align}
with $P_b$ standing for power allocation per section, and where 
\begin{equation}
\text{softmax}(x_1,\ldots,x_Q)_i
\doteq
\frac{e^{x_i}}{\sum_j e^{x_j}}.\nonumber 
\label{eq:softmax}
\end{equation}
Hence, the permutation property~\eqref{per_pro} and thereby the equivalence~\eqref{equivalence} hold.

\textbf{(ii) Sparse-Regression LDPC (SR-LDPC):}
A recently proposed concatenated scheme combining SPARCs with non-binary LDPC codes, termed sparse-regression LDPC (SR-LDPC)~\cite{Ebe25}, introduces dependencies among the nonzero positions of $\sv$ across sections.
In this case, the considered denoiser is given by
\[
\eta(\rv) \equiv \sqrt{BP_b}\eta_{\rm BP}(\eta_0(\rv))
\;.
\]
where $\eta_0(\rv)$ produces per-section posterior probabilities as 
\begin{equation}
\eta_0(\rv)[b]=\text{softmax}\!\left(\frac{\sqrt{BP_b}}{\tau}\,\rv[b]\right),
\end{equation}
and these probabilities serve as a priori information for a soft-in soft-out BP decoder denoted by $\eta_{\rm BP}$. For details, we refer to \cite{Ebe25}. While $\eta(\rv)$ is non-separable at all, it has been shown in \cite{Ebe25} and detailed in \cite[Appendix D]{fengler2025sparse} that the permutation property in \eqref{per_pro} also holds for this concatenated denoiser.

%%%%%%%%%%%%%%%%%%%%%%%%%%%%%%%%%%%%%%%%%%%%%%%%%%%%%%%%%%%%%%%%%%%%%%%%
\section{Bayesian Optimality via the RS Ansatz}\label{replica_sect}

In this section, we address the computation of two fundamental measures for the input-output system \eqref{eq:system_model}: the normalized mutual information (in nats) and the per-user MMSE defined as
\begin{align}
\mathcal{I} &\doteq\lim_{N_u,L\to\infty} \frac{1}{L}\mathbb E\Bigl[\ln \frac{p(\yv\mid \sv,\Am)}{p(\yv\mid\Am)} \mid \Am\Bigr]   \label{III}
\\
\nu_u &\doteq \lim_{N_u,L\to\infty}\frac{1}{N_u}\mathbb{E}
\Bigl[
\norm{\mathbf{s}_u - \mathbb{E}[\mathbf{s}_u \mid \mathbf{y},\Am]}^2\mid\Am]  \label{MMSSEE}
\end{align}
where we recall the joint (over users) model elements $\sv$ and $\Am$, and $\yv$ in \eqref{joint}. Our goal is to reveal the consistency between the decoupling principle — emerging in the high-dimensional analysis of the proposed algorithm — and the fundamental quantities defined above.

The rigorous computation of these quantities is an open problem. In the single-user case $U=1$,
rigorous results are available for the separable case $\mathbf{s}_1 \sim_{\mathrm{i.i.d.}}S$ and  $\Xim_1$ having i.i.d. zero-mean Gaussian entries \cite{barbier2016mutual,reeves2019replica}, as well as for the case of a generic right-unitarily invariant $\Xim_1$ in the high-SNR regime $\sigma^2 \to 0$ \cite{10272997}. In what follows, we employ the powerful - albeit mathematically non-rigorous - RS ansatz \cite{tanaka2002statistical,guo2003multiuser,Tulino13,vehkapera2016analysis,bereyhi2019statistical} to derive these asymptotic limits.

\textbf{The RS Ansatz:}
\textit{%\begin{remark}[RS prediction] %\label{RS-remark}
Let each $\alpha_u\doteq N_u/L$ be fixed as $L,N_u\to \infty$. Let each $\norm{\Am_u}_2$ be a.s. bounded as $N_u,L\to \infty$. Also, let $\lim_{N_u,L\to \infty}\frac{1}{N_u}\norm{\Am_u}_{\rm F}^2>0$. Let $\frac{1}{N_u}\,\mathbb{E}[\|\mathbf{s}_u\|^2]$ be bounded. For $\xv\doteq[x_1,x_2,\cdots,x_U]\in [0,\infty)^{U}$ and $x\in[0,\infty)$, let the following limiting functions exist:
\begin{subequations}
\label{limits}
\begin{align}
{\rm G}_{u}(\xv)&\doteq\lim_{N_u,L \to\infty}\frac{1}{N_u} \mathrm{tr}
\left(
\mathbf{A}_u^\dagger
\left(\sigma^2 \mathbf{I}_L + \sum_{u'}x_{u'} \mathbf{A}_{u'} \mathbf{A}_{u'}^\dagger\right)^{-1}
\mathbf{A}_u
\right)\\
{\rm M}_{u}(x)&\doteq \lim_{N_u\to\infty}\frac{1}{N_u}
\mathrm{mmse}
(\mathbf{s}_u \mid \mathbf{s}_u + \sqrt{x}\mathbf{z})\;
\end{align}    
\end{subequations}
where $\mathbf{z} \sim \mathcal{CN}(\mathbf{0}, \mathbf{I}_{N_u})$ is independent of $\mathbf{s}_u$. Then, the 
RS ansatz yields the following predictions:
\begin{align}
\mathcal{I}
&\overset{\rm RS}{=} \inf_{\Theta} \left[\lim_{N_u,L \to\infty}\frac{1}{L}
\ln\big\vert
\mathbf{I}_L+
\frac{1}{\sigma^2}
\sum_{u}
\lambda_u \mathbf{A}_u \mathbf{A}_u^\dagger\big\vert\nonumber\right.\\&\left.\qquad \qquad + \sum_u\alpha_u\left[\lim_{N_u\to\infty}\frac{1}{N_u}\mathcal{I}\bigl(\mathbf{s}_u; \mathbf{s}_u + \sqrt{\tau_u}\mathbf{z}\bigr)-\ln(1+\frac{\lambda_u}{\tau_u})\right]\right]
\label{rs_result}
\end{align}
where the infimum is defined over the set
\begin{align}\label{fixed}
\Theta\doteq\left\{ \matr \lambda \geq \matr 0,\matr\tau>\matr 0:~
\tau_u
=\frac{1}{{\rm G}_u(\matr \lambda)}-\lambda_u ,~\lambda_u
= \frac{\tau_u {\rm M}_u(\tau_u)}{\tau_u-{\rm M}_u(\tau_u)}\;,~ \forall u
\right\}\;.
\end{align}
Moreover, let $(\matr\lambda^\star,\matr\tau^\star)$ denote the global minimum solution in \eqref{rs_result}. Then, we have 
\begin{equation}
\nu_u
\overset{\rm RS}{=}{\rm M}_u(\tau_u^\star)\;\label{mmse}
\end{equation}
%Above, $\mathcal I(\xv;\yv)$ stands for the input-output mutual information (in nats). 
%\end{remark}
}

In \eqref{rs_result}, the existence of the limiting expressions follows by the application of the dominated convergence theorem given the existence of the limiting functions ${\rm G}_{u}$ and ${\rm M}_u$ in \eqref{limits}.
%We state the above result as a remark because its statement is {\em mathematically} exact 
%under the validity of the RS ansatz.%However, since the RS ansatz is a non-rigorous calculation method, we do not claim that the RS predictions of Remark \ref{RS-remark} are necessarily correct. 

%In Appendix~\ref{free_energy}, we derive the RS expression of the mutual information in \eqref{rs_result}, which corresponds to the statistical-physics ``free energy'' up to an additive constant. The MMSE predictions in \eqref{mmse} can be deduced from this free-energy calculation. In particular, one can introduce an auxiliary \emph{external field} in the prior$p_{u}(\sv_u)\ \to p_{u}(\sv_u)\exp\!\left(\hv^\dagger \sv_u\right),$ %so that the MMSE can be expressed as the variational derivative of the mutual information with respect to the parameter~$\hv$. By performing this variation in the (modified) free energy of the same form as \eqref{rs_result}, one recovers the input--output \emph{decoupling principle} $ (\sv_u;\ \sv_u + \sqrt{\tau_u}\,\zv)$ %in the MMSE.

Two aspects distinguish the present analysis from previous RS computations: 1) we consider a \emph{generic non-separable system}, for which the extension of the RS ansatz is  straightforward; 2) the disorder average is taken with respect to Haar unitaries $\{\Om_u\}_{u\leq U}$.
For $U = 1$, the disorder average in the replica calculation can be conveniently computed using the asymptotic Itzykson--Zuber integral \cite{Tulino13,vehkapera2016analysis,bereyhi2019statistical}. However, when $U > 1$, an application of the asymptotic Itzykson--Zuber integral is not clear. We resolve this issue by means of a new approach, and for details, we refer to Appendix~\ref{disorder_average}.

%%%%%%%%%%%%%%%%%%%%%%%%%%%%%%%%%%%%%%%%%%%%%%%%%%%%%%%%%%%%%%%%%%
\subsection{Consistency with the Recent Rigorous RS Results}

Recently, in the context of high-dimensional time series, \cite{tieplova2025information}
studied a related model in which \( \mathbf{H}_u \) is a deterministic circulant matrix and
\( \boldsymbol{\Xi}_u \) is an i.i.d.\ Gaussian random matrix, with the i.i.d. inputs
\( \sv_u \sim_{\text{i.i.d.}} S \) for all $u = 1,\ldots, U$. For this model, a  rigorous derivation of the mutual information \( \mathcal{I} \) was provided. We now show that, when specialized to this setting, our general RS result \eqref{rs_result} recovers the same formula. Furthermore, we note that \cite{tieplova2025information} leaves open the problem of developing a Bayesian-optimal algorithm for this model: an issue 
resolved in the present paper. 

First, we specialize the results of the results of RS-ansatz to the case where the precoding matrix $\Xim_u$ has i.i.d. zero-mean Gaussian entries. This is carried out by applying the results on deterministic equivalents for sums of random matrices in \cite{couillet2011deterministic,speicher2012free}.

\begin{remark}\label{remdarya}
Let the premises of the RS ansatz hold. Let the entries of $\Xim_u$ be i.i.d. zero-mean Gaussian with variance $1/L$. For each $u$, let $\norm{\Hm_u}_2$ be a.s. bounded as $L\to\infty$. Then, the RS-prediction of the limiting mutual information in \eqref{rs_result} reads as (see Appendix~\ref{proof_drem})
\begin{align}
\mathcal{I}
&\overset{\rm RS}{=}\inf_{\Theta}\lim_{L \to\infty}\frac{1}{L}
\ln\vert 
\mathbf{I}_L
+
\frac{1}{\sigma^2}
\sum_{u}\alpha_u\nu_u \Gm_u\vert+\sum_u\alpha_u\left[\lim_{N_u\to\infty}\frac{1}{N_u}\mathcal{I}\bigl(\mathbf{s}_u; \mathbf{s}_u + \sqrt{\tau}_u\mathbf{z}\bigr)-\frac{\nu_u}{\tau_u}\right]\label{rs_darya}
\end{align}
where for short we define the Gramians $\Gm_u\doteq \Hm_u\Hm_u^\dagger$ and the infimum is defined over the set
\begin{align}
\nonumber 
\Theta\doteq\left\{ \matr \nu \geq \matr 0,\matr\tau>\matr 0:
\frac{1}{\tau_u}=\lim_{L\to\infty}\frac{1}{L}{\rm tr}(\Gm_u(\sigma^2\Id_L+
\sum_{u'}\alpha_{u'}\nu_{u'}\Gm_{u'})^{-1}),\nu_u= {\rm M}_u(\tau_u)\;,\forall u
\right\}\;\;.
\end{align}

\end{remark}
Next, we further specialize the RS result \eqref{rs_darya} to the setting considered in~\cite{tieplova2025information}. First, we set
$\mathbf{\Gm}_u\equiv\Tm(\rho_u)$ where \( \Tm(\rho) \) denotes the \( L \times L \) family of Toeplitz matrices
\begin{equation}
\Tm(\rho)
=\frac{1}{1-\rho^2}
{\small 
\left[
\begin{array}{ccccc}
1 & \rho & \rho^2 & \cdots & \rho^{L-1} \\
\rho & 1 & \rho & \cdots & \rho^{L-2} \\
\vdots & \vdots & \ddots & \ddots & \vdots \\
\rho^{L-2} & \rho^{L-3} & \rho^{L-4} & \ddots & \rho \\
\rho^{L-1} & \rho^{L-2} & \rho^{L-3} & \cdots & 1
\end{array}
\right].}
\end{equation}
Then, by the Kac--Murdock--Szegö theorem~\cite{grenander1958toeplitz}, we have for  $\delta_u(\theta) \doteq (1-2\rho_u\cos(\theta)+\rho_u^2)^{-1}$
\begin{eqnarray}
\lim_{L \to \infty}\frac{1}{L}
\ln \big\vert
\mathbf{I}_L
+
\frac{1}{\sigma^2}
\sum_{u }\alpha_u\nu_u \mathbf{G}_u\big\vert
&=&
\frac{1}{\pi}\int_{0}^{\pi}
\ln\big(
1+\frac{1}{\sigma^2}\sum_{u}\alpha_u \nu_u \delta_u(\theta)
\big)
\,\mathrm{d}\theta,
\end{eqnarray}
where we define $\delta_u(\theta) \doteq (1-2\rho_u\cos(\theta)+\rho_u^2)^{-1}$. Second, let \( \mathbf{s}_u \sim_{\mathrm{i.i.d.}} S \), \( \forall u \). Thus, we have
\[\frac{1}{N_u}\mathcal{I}\bigl(\mathbf{s}_u; \mathbf{s}_u + \sqrt{\tau}_u\mathbf{z}\bigr)\equiv \mathcal I(S;  S+\sqrt{\tau_u}Z)\;.\] We also get
\begin{align}
\nonumber 
\Theta\equiv\left\{ \matr \nu \geq \matr 0,\matr\tau>\matr 0:~
\nu_u
=
\mathrm{mmse}\!\left(S \mid S + \sqrt{\tau_u}Z\right),
\frac{1}{\tau_u}=\;\frac{1}{\pi}
\int_{0}^{\pi}
\frac{\delta_u(\theta)\,\mathrm{d}\theta}
{\sigma^2+\sum_{u'}\alpha_{u'}\nu_{u'}\delta_{u'}(\theta)}\;, \forall u\right\} 
\end{align}
Thus, we recover exactly the same result as in~\cite[Theorem~II.1]{tieplova2025information}.
\vspace{-0.5cm}

\subsection{Bayesian Optimality of the Proposed Algorithm}
We next show the SE recursion under the ``Bayesian-optimal'' choice given in \eqref{SE_optimal} algorithmically (i.e., for sufficiently large $1\ll t\ll N$) solves the general RS fixed-point equations:
\begin{align}
\tau_u
=\frac{1}{{\rm G}_u(\matr \lambda)}-\lambda_u, \quad \lambda_u
=\frac{\tau_u{\rm M}_u(\tau_u)}{\tau_u-{\rm M}_u(\tau_u)}\label{rs-fix}
\end{align}

To this end, we first underline the following contraction property of the SE recursion in \eqref{SE_optimal}.
\begin{lemma}\label{contraction}
Consider the SE recursion in \eqref{SE_optimal}. We say the input signal $\sv_u$ has an isotropic Gaussian distribution if $\sv_u\sim \mathcal{CN}(\mv_u,\sigma_u^2\Id_{N_u})$ for some mean $\mv_u$ and variance $\sigma_u^2\geq 0$.  
If $\sv_u$ has an isotropic Gaussian distribution, then $\lambda_u^{(t)}= \sigma_u^2$. Otherwise, 
we have the contraction $0 < \lambda_u^{(t+1)} < \lambda_u^{(t)}$. Moreover, we have the contraction $0 < \tau_u^{(t+1)} < \tau_u^{(t)}$
unless $\sv_u$ for all $u = 1,\ldots, U$ have isotropic Gaussian distributions.  Consequently, $\tau_u^{(t)}$ and $\lambda_u^{(t)}$ converge, as $t \to \infty$, to a fixed point of \eqref{SE_optimal}.
\begin{proof}
See Appendix~\ref{proof_rem_contradiction}\;.
\end{proof}
\end{lemma}
To connect the above iterative solution to the RS fixed-point solution in
\eqref{fixed}, in addition to the premises stated in the RS ansatz,
we impose the following regularity condition on the posterior mean denoiser [as in Theorem~\ref{th1}]: Specifically, let
$\eta_{u}(\rv;\tau)=\mathbb E[\sv_u\vert \rv=\sv_u+\sqrt{\tau}\zv]$ be differentiable and uniformly Lipschitz (w.r.t. $\rv$) with a Lipschitz constant ${\rm Lip}=O(\tau^{\pm p})$ for any constant $p\geq 0$. 
Under this additional regularity assumption, we obtain that (see
Appendix~\ref{proof_rem_contradiction})
\begin{align}\label{convergence_replica}
(\matr\tau,\matr\lambda)=
 \lim_{t\to\infty}\lim_{N_u,L\to\infty} (\matr\tau^{(t)},\matr\lambda^{(t)})\;.
\end{align}

Note that for the Bayesian--optimal setting as in
Remark~\ref{Bayesian_optimal} it follows from
Corollary~\ref{cor_mse} that
\begin{equation}
\frac{1}{N_u}\,\lVert \sv_u-\eta_{u}(\rv_u^{(t)};\tau_u^{(t)})\rVert^2
=
\frac{1}{N_u}\,\mathrm{mmse}(\sv_u \vert\sv_u+\sqrt{\tau_u^{(t)}}\zv)+ \Op{L^{-\frac{1}{2}}}.
\end{equation}
Then, from \eqref{convergence_replica} we have (a.s.)
\begin{equation}
\lim_{t\to \infty}\lim_{N_u,L\to \infty}
\frac{1}{N_u}\,\lVert \sv_u-\eta_{u}(\rv_u^{(t)};\tau_u^{(t)})\rVert^2
=
\lim_{N_u\to \infty}
\frac{1}{N_u}\,\mathrm{mmse}\!\left(
\sv_u \,\middle|\, \sv_u+\sqrt{\tau_u}\,\zv
\right),
\end{equation}
where $\tau_u$ is a solution of the RS fixed--point equation in
\eqref{rs-fix}. 

Thus, assuming that the RS fixed-point equations admit a unique solution—referred to as the \emph{algorithmic phase}~\cite{barbier2019optimal}—then for sufficiently large iterations satisfying $1 \ll t \ll N$, the normalized mean-square error (MSE) of the algorithm approaches arbitrarily closely 
the asymptotic MMSE predicted by the RS ansatz in~\eqref{mmse}.

In contrast, when the RS fixed-point equations admit multiple solutions, the system is said to be in the \emph{hard phase}~\cite{barbier2016mutual}. In this regime, constructing a Bayes-optimal algorithm remains an open problem, even in the classical single-user setting.

\subsection{Universality in the Case of Full Statistical Symmetry Across the Users}
We now restrict our attention to the case of full statistical symmetry across users:
\begin{equation}\label{fullsymmerty}
\Am_u \sim \mathbf{A}_{u'} \quad \text{and} \quad \sv_u \sim \sv_{u'} \quad \forall u \neq u'.
\end{equation}
This assumption will reveal a universality phenomenon: while the concatenated matrix $\mathbf{A} \doteq[\mathbf{A}_1, \mathbf{A}_2, \ldots, \mathbf{A}_U]$ is, in general, not right-unitarily invariant, the replica predictions yield expressions that coincide with those obtained under the assumption that $\mathbf{A}$ is right-unitarily invariant.

Specifically, \eqref{fullsymmerty} implies the functional symmetry properties for all $u\ne u'$
\begin{align}\label{sym}
{\rm G}_u(\xv)={\rm G}_{u'}(\xv)\quad \text{and} \quad
{\rm M}_u(x)={\rm M}_{u'}(x)\;. 
\end{align}
As a consequence, we can define the scalar functions
\begin{align}
{\rm G}(x)&\doteq {\rm G}_u(x\matr 1_U)=\lim_{N,L\to \infty}\frac{1}{N}\bigl(
\mathbf{A}^\dagger
(\sigma^2 \mathbf{I}_L +x\mathbf{A}\mathbf{A}^\dagger)^{-1}
\mathbf{A}
\bigr)\\
{\rm M}(x)&\doteq {\rm M}_u(x)=\lim_{N\to \infty}\frac{1}{N}{\rm mmse}(\sv\vert \sv+\sqrt{x}\zv)\;.
\end{align}

Under the symmetry property \eqref{sym}, the RS prediction of the mutual information simplifies to
\begin{equation}\label{RS_simple}
\mathcal{I}\overset{\rm RS}{=} \inf_{\Theta} \left[\lim_{N,L \to\infty}\frac{1}{L}
\ln\big\vert
\mathbf{I}_L+
\frac{\lambda}{\sigma^2} \mathbf{A} \mathbf{A}^\dagger\big\vert +\alpha\lim_{N\to\infty}\frac{1}{N}\mathcal{I}\bigl(\mathbf{s}; \mathbf{s} + \sqrt{\tau}\mathbf{z}\bigr)-\alpha\ln(1+\frac{\lambda}{\tau})\right]
\end{equation}
where $\alpha=N/L$ and the infimum is defined over the set
\begin{align}
\Theta\equiv\left\{\lambda \geq 0,\tau>0:~~
\tau
=\frac{1}{{\rm G}(\tau)}-\lambda ,~~\lambda
=\frac{\tau {\rm M}(\tau)}{\tau-{\rm  M}(\tau)}
\right\}\;.
\end{align}

Finally, by invoking the following representation of the R-transform for a random variable $X>0$, \cite[Lemma~1]{opper2016theory}
\begin{align}
\mathbb E[\ln X]=-(1+\ln\mathbb E[X^{-1}])+\int_{0}^{\mathbb E[X^{-1}]}{\rm R}_{X}(-s){\rm ds}
\end{align}
where ${\rm R}_X$ denotes the R-transform of the distribution of $X$, we note that for separable systems $\mathbf{s} \sim_{\mathrm{i.i.d.}} S$, the expression in \eqref{RS_simple} exactly coincides with the RS formula in \cite[Claim~1]{Tulino13} which has been rigorously proved in the regime $\sigma^2 \to 0$ in \cite{10272997}.

%%%%%%%%%%%%%%%%%%%%%%%%%%%%%%%%%%%%%%%%%%%%%%%%%%%%%%%%%
\section{Numerical Results}\label{sec_sim}

%%%%%%%%%%%%%%%%%%%%%%%%%%%%%%%%%
\subsection{Channel Model}

The linear Gaussian multiuser model considered in this paper is very general and encompasses several 
settings of practical relevant interest. In this section we discuss the case  of 
channel matrices resulting from a general linear time-varying dispersive multipath fading channel
coupled with precoding and post-processing corresponding to waveforms used in wireless 
communications. We develop the model for the single-antenna case, 
while the extension to multiple antennas at both the users' and the receiver sides is straightforward. 

A typical time-varying underspread multipath fading channel, including the transmitter square-root 
Nyquist pulse shaping filter 
$g(t)$, has complex baseband impulse response given by \cite{gaudio2021otfs}
\begin{equation}
    h_u(t,\tau) = \sum_{p=1}^{P_u} \alpha_{u,p} \, e^{j 2\pi \nu_{u,p} t} \, 
    g(\tau - \tau_{u,p}),  \label{phy-channel}
\end{equation}
where $(\alpha_{u,p}, \nu_{u,p}, \tau_{u,p})$ denote the complex gain, Doppler shift, 
and delay of the $p$-th path of user $u$. 
Each user $u$ sends a time-domain sequence of discrete-time ``chips'' obtained from the precoded information symbols $\xv_u$ via a waveform-dependent precoder, at a rate of $W$ chips/s, where $W$ is the (unilateral) 
channel bandwidth around the carrier frequency $f_0$. At the receiver, after carrier demodulation, the signal is passed through a chip matched filter $g(-t)^*$ and it is sampled at rate $W$. 
The resulting discrete-time baseband channel can be represented as a 
tall banded {\em discrete-time convolution} matrix $\Hm_{0,u}$ where the bandwidth is related to the delay spread (in chip intervals) and is generally much smaller than transmission blocklength $L$. 
The $L \times L$ effective channel matrix is then given by  $\mathbf{H}_u = \mathbf{U}_{\rm rx} \mathbf{H}_{0,u} \mathbf{U}_{\rm tx}^\dagger$, where $\mathbf{U}_{\rm rx}$ and $\mathbf{U}_{\rm tx}$ are modulation-dependent postprocessing and precoding matrices. 
A comprehensive treatment of typical waveforms with their corresponding 
postprocessing and precoding matrices can be found in \cite{zhang2026unified}.

For example, in the case of OFDM, letting $L = K M$ and applying cyclic prefix insertion and removal to each block and IDFT/DFT per block, $\Hm_u$ takes on a block-diagonal form. If the channel is time-invariant, each block is also diagonal. Otherwise, in the presence of significant Doppler spread (generating {\em inter-carrier interference}), the blocks are non-diagonal. Other modulation schemes, such as OTFS (see \cite{zhang2026unified} above), differ in the nature of the time-frequency precoding/postprocessing, and in block partitioning, but conceptually can be handled in the same say. In general, by partitioning the $L$-length information sequence into smaller $M$-blocks and introduce some redundancy (e.g., cyclic-prefix, zero-padding) to avoid inter-block interference, $\Hm_u$ 
takes on a block diagonal form. This approach has a significant advantage in terms of complexity, since the matrix inverses (see \eqref{filter_lmmse}) in the MU-OAMP can be performed on a per-block basis, but incurs in some redundancy. 
In contrast, schemes that do not operate such partitioning are more spectrally efficient but generally incur a much larger complexity. 

In this work, we consider an OFDM system with $M = 256$ subcarriers and total bandwidth 
$W = 20~\mathrm{MHz}$, using $K = 16$ OFDM blocks with the multi-block cyclic-prefix
configuration, so that $\mathbf{H}_u$ is in block diagonal form with $K$ blocks of dimension $M \times M$. 
A maximum delay spread of $\max_p \tau_p = 1~\mu\mathrm{s}$ (corresponding to a 
path-length difference of approximately $300~\mathrm{m}$) gives a discrete channel 
length of $\ell = W\max_p\tau_p = 20 \ll M$. 
With carrier frequency $f_c = 4~\mathrm{GHz}$ 
and user velocity $v = 80~\mathrm{km/h}$, the maximum Doppler shift is 
$\max_p \nu_p = f_c v/c \approx 300~\mathrm{Hz} \ll W$, with 
comfortably satisfies the underspread channel assumption\cite{tse2005fundamentals}.

To introduce inter-user asymmetry, the channel matrices are rescaled as 
\[\mathbf{H}_u = \sqrt{p_u}\,\mathbf{H}\;,\]where $\mathbf{H}$ is generated 
independently for each user and $p_u$ is a user-specific power gain.
The $P_u = 23$ channel path parameters $(\alpha_p, \nu_p, \tau_p)$ are drawn from 
the 5G-3GPP channel model \cite{5G3GPP}, which can be generated via MATLAB's 
5G Toolbox.

%%%%%%%%%%%%%%%%%%%%%%%%%%%%%%%%%%%%%%%%%%%
\subsection{Signal Model}

We restrict the simulations to the case of two users, i.e., \(U = 2\). For convenience, statistical symmetry is assumed across the transmitted signals as $\sv_1\sim\sv_{2}$.
We address two signal models:
\begin{itemize}
    \item Uncoded (i.i.d.) inputs: We cosider i.i.d. QPSK constellations as   
\begin{align}
    \sv_u \sim_{\text{i.i.d.}} \frac{1}{\sqrt{2}}(R+\mathrm{j}I) \quad\text {with}\quad (R,I) \sim \mathrm{Unif}\!\big(\{\pm 1\}\big)\label{qpsk}
\end{align}
\item SR-LDPC Inputs: 
The constellations $\sv_u \in \CC^{BQ}$ are created by representing a codeword from a non-binary LDPC code with alphabet size $Q$ and length $B$ as a sequence of $B$ orthogonal basis vectors through the canonical mapping $T:\FF_Q \to \{0,1\}^Q$, $T(q) = \ev_q$. The vector $\sv_u$ is then multiplied by a coding matrix $\Xim_u \in \CC^{L \times BQ}$. 
{ To reduce the computational complexity of constructing the random semi-unitary matrices
\( \mathbf{\Xi}_u \), motivated by universality results for OAMP~\cite{CakmakOpper19,wang2024universality,dudeja2023universality},
we choose \( \mathbf{\Xi}_u \) to be a randomly signed sub-sampled Fourier matrix. This also reduces the complexity of computing matrix--vector products of the form
\( \mathbf{\Xi}_u \mathbf{x} \) and \( \mathbf{\Xi}_u^\dagger \mathbf{y} \) in the algorithm, and hence, allows for efficient encoding and decoding with $\mathcal{O}(BQ\log BQ+KM^3)$ complexity.}
The sign vectors and sub-sampling patterns are different for each user and known at the receiver. For decoding, we use the proposed algorithm with a soft-in soft-out BP denoiser, as described in Sec. \ref{sec:add_assumptions}. For the simulations, we use a rate $3/4$ outer LDPC code with alphabet size $Q = 256$, and length $B=600$. This give a payload length of $3600$ bits per user. Following \cite{Ebe25}, the non-binary LDPC code is constructed by fixing the variable node degree to $d_v = 3$, and constructing the factor graph by progressive edge growth.     
\end{itemize}
In passing, we note that our MU-OAMP framework has also been applied, as a follow-up study, to SR codes without an outer LDPC code, together with optimized power allocation across sections; see~\cite{Hao26} for details.
Our motivation for considering SR--LDPC codes here is that the outer LDPC code induces full dependence across sections, providing a compelling testbed for examining theory--simulation discrepancies in non-separable systems.

\subsection{Toy Static Channels: Random Precoding versus OFDMA}\label{motivation}

To illustrate the impact of random precoding, we first consider the static scenario with \(v=0~\mathrm{km/h}\); so that \(\nu_p=0\) for all paths.
While OFDMA is known to be highly effective in this regime, we demonstrate that the MU-OAMP approach can significantly outperform OFDMA.

In OFDMA, different users occupy disjoint subcarrier bands to eliminate inter-user interference.
For two users, we define the projection matrices
$
\Pm_1 \doteq {\small
\begin{bmatrix}
     \Id_{L/2}  \\
     \mathbf{0}  
\end{bmatrix}}
\text{~and~}
\Pm_2 \doteq 
{\small \begin{bmatrix}
     \mathbf{0}  \\
     \Id_{L/2} 
\end{bmatrix}}$
leading to the received signal
\begin{equation}\label{equ:OFDMA}
\yv = \Hm_1 \Pm_1 \sv_1 + \Hm_2 \Pm_2 \sv_2 + \nv .
\end{equation}
For static channels, with the OFDM waveform $\Hm_1$ and $\Hm_2$ are diagonal, rendering the observation model fully decoupled and allowing exact MMSE detection.

For comparison, we consider random precoding:
\begin{equation}
\yv = \Hm_1 \Xim_1 \sv_1 + \Hm_2 \Xim_2 \sv_2 + \nv,
\end{equation}
where $\Xim_1$ and $\Xim_2$ are independent semi-unitary of dimensions $L \times L/2$. In static channels, the MU-OAMP algorithm with the optimal Bayesian filter~\eqref{filter_lmmse} does not require matrix inversion, resulting in per-iteration complexity \(\mathcal{O}(L^2)\), comparable to OFDMA detection.

Restricting attention to uncoded i.i.d.\ QPSK inputs~\eqref{qpsk}, Figure.~\ref{fig:RM vs OFDMA} shows that MU-OAMP significantly outperforms OFDMA with MMSE detection.
\begin{figure}[!htbp] 
\centering \includegraphics[width=0.45\linewidth]{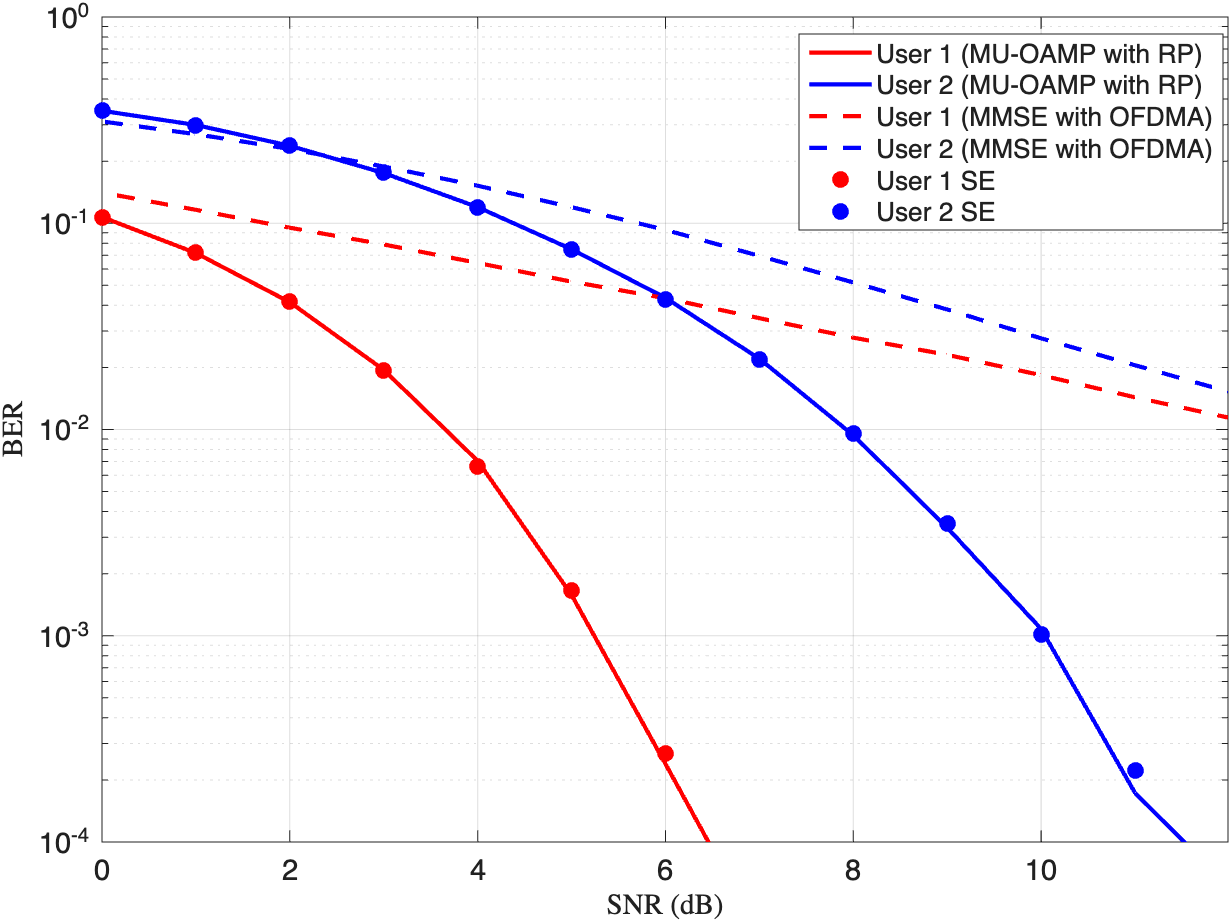} \caption{Time-invariant channels: MU-OAMP vs. OFDMA (QPSK signals, $(p_1,~p_2)=(6~\mathrm{dB},~0~\mathrm{dB})$).} \label{fig:RM vs OFDMA} 
\end{figure} 
\vspace{-.5cm}
\begin{figure}[!htbp] \centering \includegraphics[width=0.6\linewidth]{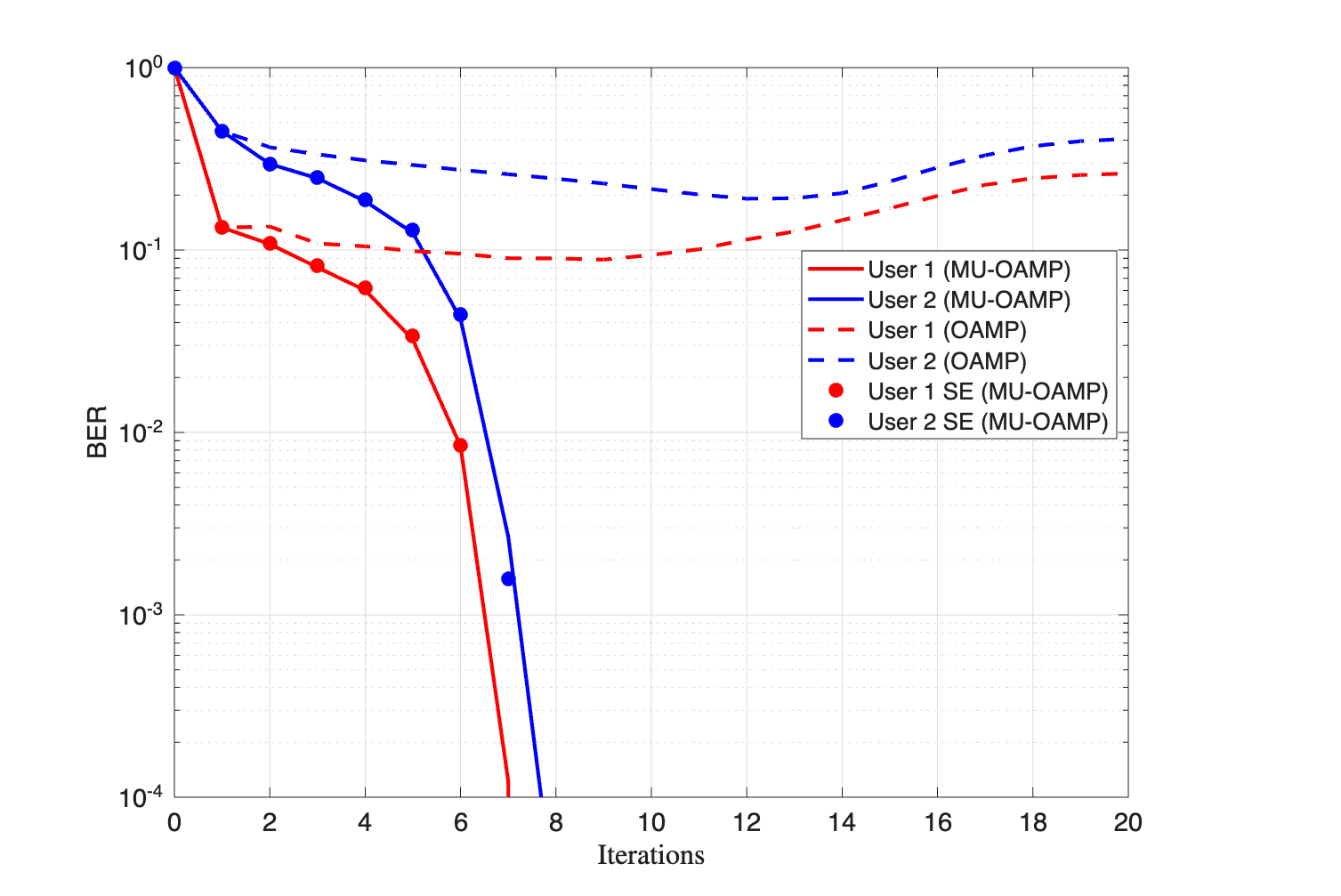} \caption{Time-varying Channes: MU-OAMP vs. OAMP  (QPSK signals, $M=256,~ \mathrm{snr}=13~\mathrm{dB},~(p_1,~p_2)=(6~\mathrm{dB},~0~\mathrm{dB})$).} \label{fig:Exact vs approximate vs matched} 
\end{figure}
\subsection{Time-Varying Channels}
\begin{figure}[!htbp] \centering \includegraphics[width=0.48\linewidth]{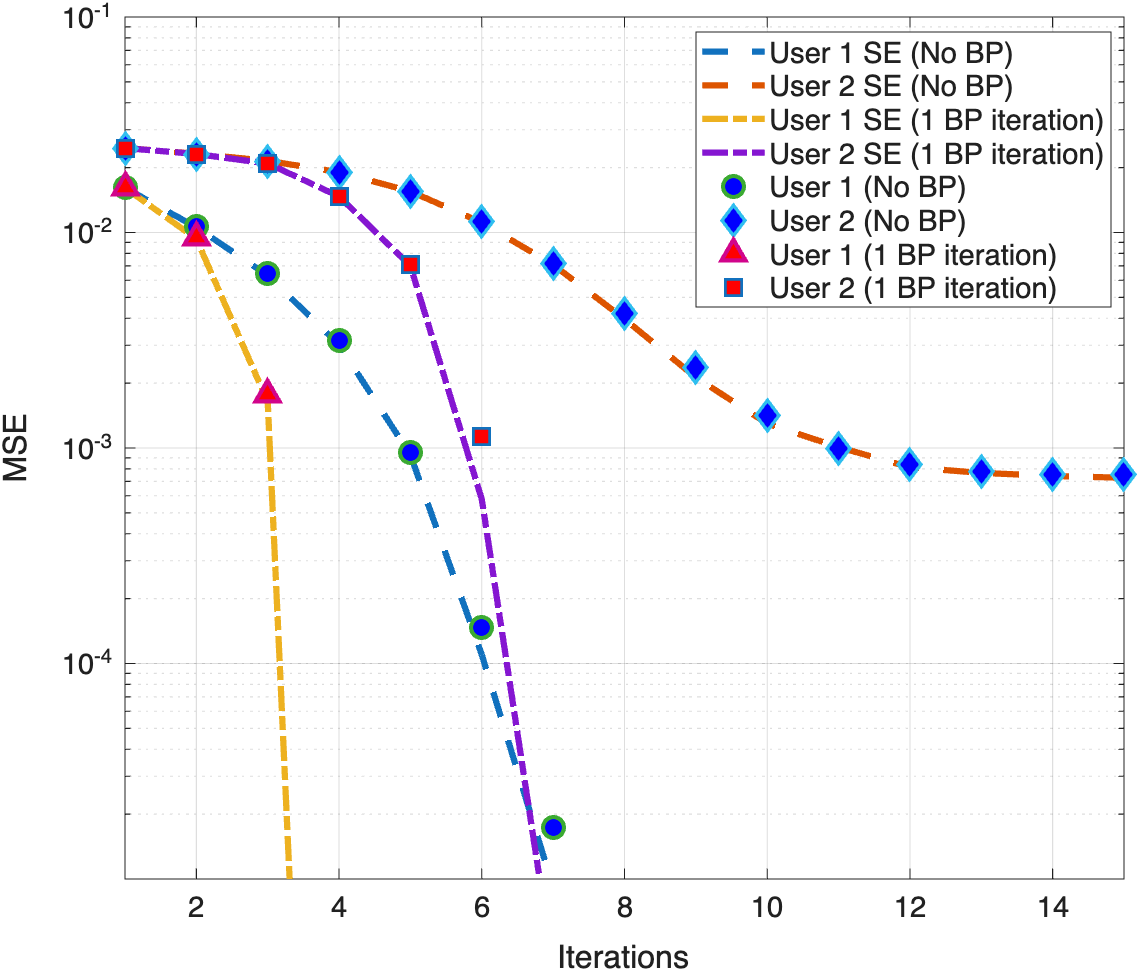} \caption{Time-varying Channels: MSE (SR-LDPC signals, $Q=256,~ E_b/N_0=10.5~\mathrm{dB},~(p_1,~p_2)=(6~\mathrm{dB},~0~\mathrm{dB})$).} \label{fig:sparc_sim1} \end{figure} \begin{figure} \centering \includegraphics[width=0.48\linewidth]{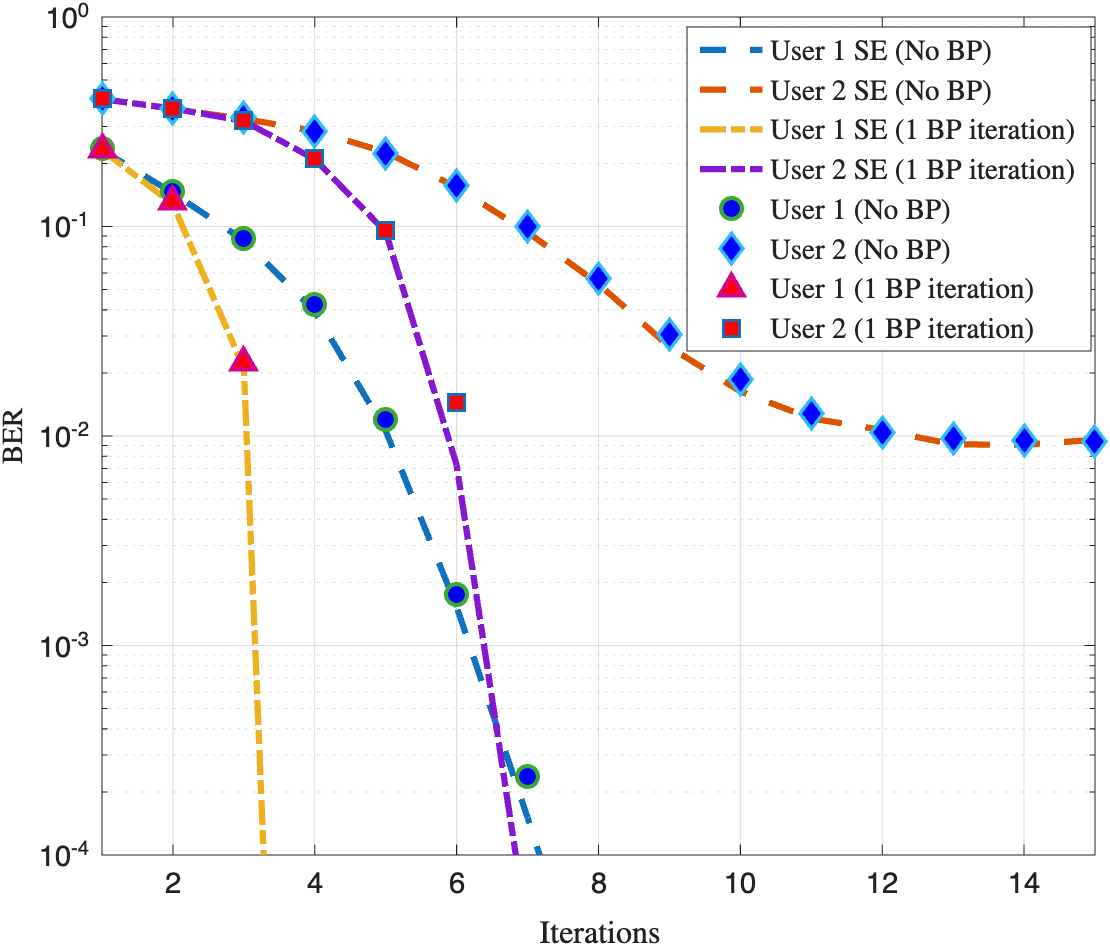} \caption{ Time-varying Channels: BER (SR-LDPC signals, $Q=256,~ E_b/N_0=10.5~\mathrm{dB},~(p_1,~p_2)=(6~\mathrm{dB},~0~\mathrm{dB})$).} \label{fig:sparc_sim2} 
\vspace{-.6cm}
\end{figure} 
We now consider time-varying channels by setting the user velocity to \(v=80~\mathrm{km/h}\). %and using a no-waveform model $\Hm \equiv \Hm_0 $. 

%{\RED [I DON'T UNDERSTAND ... what does it mean ``no waveform''? you should say that for the case $\Um_{\rm tx} = \Um_{\rm rx} = \Id_L$ the ``waveform'' correspond to single-carrier transmission, and in this case $\Hm_u = \Um_{0,u}$ is a convolution matrix ... it may be tall or square depending on truncation of the ``tail'' of the block ... with truncation, we can make it square. QUESTION: would it be possible to have both results for single carrier, which is $\Um_{\rm tx} = \Um_{\rm rx} = \Id_L$, and for the same channel results for OFDM, which of course can be applied also to the time-varying case yielding block diagonal but not diagonal blocks, and verify that effectively the performances are essentially the same ??? are they (I think so) ...]}

Figure~\ref{fig:Exact vs approximate vs matched} compares MU-OAMP using the optimal Bayesian construction~Remark~\ref{Bayesian_optimal} with conventional OAMP under user power asymmetry \((p_1,p_2)=(6~\mathrm{dB},0~\mathrm{dB})\).
Under full statistical symmetry, i.e., \(\Hm_1\Xim_1\sv_1 \sim \Hm_2\Xim_2\sv_2\), both methods perform similarly as expected from the general SE formalism in~Definition~\ref{SE_new}.

Figures~\ref{fig:sparc_sim1} and~\ref{fig:sparc_sim2} consider the same channel model but employ SR--LDPC coded signals with \(L=KM\) and \(K=16\).
We compare a non-separable denoiser using one BP iteration per OAMP iteration with a separable softmax denoiser.
The choice of a single BP iteration (per denoising step) follows~\cite{fengler2025sparse}.
The MSE and BER required for SE are precomputed empirically as functions of the effective noise variance, since they are independent of the channel and coding matrices.

The results clearly demonstrate that the non-separable denoiser yields substantial performance gains, and that the proposed SE framework accurately predicts empirical performance even in challenging settings involving non-separable denoising and time-varying channels.

%%%%%%%%%%%%%%%%%%%%%%%%%%%%%%%%%%%%%%%%%%%%%%%%%%%%%%%%%
\section{Proof of Theorem~\ref{th1}}\label{Pth1}

In this section, we prove Theorem~\ref{th1}; yet to maintain the flow of exposition, the proofs of the lemmas invoked in the argument are collected in Appendix~\ref{proof_auxilary}.

We recall the singular value decomposition of the precoding matrix as $\matr\Xi_u = \Um_u \matr{\Lambda}_u \Om_u$ in \eqref{svd} and introduce for convenience the following matrix elements:
\begin{align}
\Qm_u&\doteq\Hm_{u}\Um_u\matr\Lambda_u\label{widehatQ}\\
\Mm_u^{(t)}&\doteq \frac{\Qm_u^\dagger \matr \Sigma^{(t)}}{\langle \Am_u^\dagger\matr \Sigma^{(t)}\Am_u\rangle }\label{Mut}\;.
\end{align}
We will analyze the residual dynamics for each $(u,t)\in[U]\times [T]$
\begin{align}
\widehat\phiv_u^{(t)}\doteq \rv_u^{(t)}-\sv_u\text{~~and~~}\widetilde\phiv_u^{(t)}\doteq \fv_{u}^{(t)}-\sv_u
\end{align}
In particular, from the original dynamics \eqref{oamp0} we have the recursion 
\begin{subequations}\label{resoamp}
\begin{align}
\widehat{\qv}_u^{(t)}&=\Om_u \widetilde\phiv_u^{(t)}\label{alg1}
\\  \widehat\psiv_u^{(t)}&=\Mm_u^{(t)}\nv+(\Id-\Mm_u^{(t)}{\Qm}_u)\widehat\qv_u^{(t)} -\Mm_u^{(t)}\sum_{u'\neq u} \Qm_{u'}\widehat{\qv}_{u'}^{(t)}\label{alg2}\\
\widehat\phiv_{u}^{(t)}&=\Om^\dagger_u\widehat\psiv_u^{(t)}\label{alg3}\\
\widetilde\phiv_{u}^{(t+1)}&=f_{u,t}(\sv_u+\widehat\phiv_u^{(t)})-\sv_u\;. \label{alg4}
\end{align}
\end{subequations}

The dynamics \eqref{resoamp} are coupled by the Haar matrices $\Om_u$ through the products $ \{\Om_u \widetilde\phiv_u^{(t)},\Om_u^\dagger\widehat\psiv_u^{(t)}\}\;.$  As a first step,  we will leverage the rotational invariance of the Haar distribution to construct a statistically equivalent representation of the AMP-type dynamics that avoids explicit dependencies on the Haar matrices. This approach resembles the ``Householder Dice'' approach of ~\cite[Section~III-C]{lu2021householder} but we use a Gram-Schmidt representation instead of explicit Householder reflections. While the latter is computationally more stable and efficient in simulating AMP-type dynamics, the Gram-Schmidt representation is convenient for the high-dimensional analysis employed here. 
\subsection{The Haar-Free Equivalent}\label{step1}
We introduce the classical Gram-Schmidt notation: Let $\vv^{(1:t-1)}= \{\vv^{(1)},\vv^{(2)},\ldots, \vv^{(t-1)}\}$ 
be a collection of  vectors in $\CC^{N}$ with $\langle \vv^{(i)},\vv^{(j)} \rangle = \delta_{ij}$ for all $i,j$. Also, let us denote the projection matrix to the orthogonal complement of $\text{span}(\vv^{(1:t-1)})$ by
\begin{equation}
	\Pm^\perp_{\vv^{(1:t-1)}}\doteq \Id_N- \frac 1 N\sum_{1\leq s<t}\vv^{(s)}(\vv^{(s)})^\dagger \;.
\end{equation}
Then, for any $\bv\in \CC^{N}$, we construct a new orthogonal vector $\vv^{(t)}=\GS{\bv}{\vv^{(1:t-1)}}$ where
\begin{equation}
	\GS{\bv}{\vv^{(1:t-1)}} \doteq 
    \sqrt{N}\frac{\Pm^\perp_{\vv^{(1:t-1)}}\bv}{
\Vert \Pm^\perp_{\vv^{(1:t-1)}}\bv \Vert }
\end{equation}
unless $\bv\in{\rm span}(\vv^{(1:t-1)})$. In case $\bv\in{\rm span}(\vv^{(1:t-1)})$ we generate an arbitrary vector $\vv^{(t)}$ such that $\langle \vv^{(t)},\vv^{(i)}\rangle=\delta_{ti}$. %\footnote{E.g., let  $ \vv^{(t)}=\GS{\zv}{\vv^{(1:t-1)}}$ with arbitrary $\zv\sim\mathcal{CN}(\matr 0,\Id_N)$ such that $\norm{\Pm^\perp_{\vv^{(1:t')}}\zv}=\norm{\zv'}\overset{a.s.}{>}0$  for $N>t'$ where $\zv'\sim\mathcal{CN}(\matr 0,\Id_{N-t'})$\;.}. 
Furthermore, by convention, we set $\mathcal {GS}(\bv)\equiv\mathcal {GS}(\bv\vert \emptyset)$. Note also that
	\begin{equation}
		\bv=\sum_{1\leq s\leq t}\vv^{(s)}\langle\vv^{(s)},\bv \rangle\;. \label{gsd}
	\end{equation}
\begin{lemma}\label{lemmaseq} Fix $T$ and let $N_u>T$ for all $u = 1,\ldots, U$. Let $\Om_u$ be Haar unitary matrix where $\Om_u$ and $\Om_{u'}$ are mutually independent for any $u\neq u'$.  Then, as we will verify in Appendix~\ref{proof_lemmaseq}, the joint probability distribution of the sequence of vectors $\{{\widehat\phiv}_u^{(1)},{\widehat\phiv}_u^{(2)},\cdots,{\widehat\phiv}_u^{(T)}\}_{u\in [U]}$ generated by dynamics \eqref{resoamp} is equal to that of the same sequence generated by the following dynamics: We begin with the same initialization ${
\fv}_u^{(1)}$ of the dynamics \eqref{resoamp} (for each $u\in[U]$) and for the iteration steps $t=1,2,\cdots, T$ we construct
\begin{subequations}
\label{hauseholderep}
\begin{align}
\vv_{u}^{(2t-1)}&=\mathcal{GS}(\gv_{u}^{(t)}\vert \vv_{u}^{(1:2(t-1))})\\
\widetilde\vv_{u}^{(2t-1)}&=\mathcal{GS}({\widetilde\phiv}_{u}^{(t)}\vert \widetilde\vv_{u}^{(1:2(t-1))})\\
\widehat\qv_{u}^{(t)}&=\sum_{1\leq s< 2t}\vv_{u}^{(s)}\langle\widetilde\vv_{u}^{(s)},{\widetilde\phiv}_{u}^{(t)} \rangle \\
\widehat\psiv_u^{(t)}&=\Mm_u^{(t)}\nv+(\Id-\Mm_u^{(t)}{\Qm}_u)\widehat\qv_u^{(t)}-\Mm_u^{(t)}\sum_{u'\neq u}{\Qm}_{u'}\widehat{\qv}_{u'}^{(t)}\\
\vv_{u}^{(2t)}&=\mathcal{GS}(\widehat{\psiv}_{u}^{(t)}\vert \vv_{u}^{(1:2t-1)})\\
\widetilde\vv_{u}^{(2t)}&=\mathcal{GS}(\widetilde{\gv}_{u}^{(t)}\vert \widetilde\vv_{u}^{(1:2t-1)})\\
\widehat\phiv_u^{(t)}&=\sum_{1\leq s\leq 2t}\widetilde\vv_{u}^{(s)}\langle\vv_{u}^{(s)},{\widehat\psiv}_{u}^{(t)} \rangle\\
{\widetilde\phiv}_{u}^{(t+1)}&=f_{u,t}(\sv_u+\widehat\phiv_u^{(t)})-\sv_u\;.
\end{align}
\end{subequations}
Here, for each pair $(u,t)\in[U]\times[T]$, the random vectors 
$\mathbf{g}_u^{(t)},\widetilde{\mathbf{g}}_u^{(t)}\sim\mathcal{CN}(\mathbf{0},\mathbf{I}_{N_u})$ are mutually independent and independent across all indices $(u,t)$.
\end{lemma} 
 \subsection{The High-Dimensional Equivalent}\label{step2}
We next use the preliminary results on high-dimensional probability given in Appendix~\ref{preliminariesop} to derive a high-dimensional representation of the Haar-free equivalent dynamics \eqref{hauseholderep}, and thereby a high-dimensional representation of the original residual dynamics \eqref{resoamp}. 

We will begin by introducing the effective high-dimensional representations of the dynamics in~\eqref{hauseholderep} (will be proved to be exact up
to $\Op{1}$ deviations). To this end, it is useful to introduce the following notation for \emph{Cholesky decompositions}: Let $\mathcal{A}^{(1:t')}$ denote a $t' \times t'$ matrix whose $(t,s)$-indexed entries are denoted by $\mathcal{A}^{(t,s)}$. For any positive semidefinite matrix $\mathcal{A}^{(1:t')} \geq \mathbf{0}$, let $\mathcal{B}^{(1:t')}$ be the corresponding $t' \times t'$ upper-triangular matrix such that $\mathcal{A}^{(1:t')} 
= \big(\mathcal{B}^{(1:t')}\big)^{\dagger} \mathcal{B}^{(1:t')}$. We denote this factorization by $\mathcal B^{(1:t')}={\rm chol}(\mathcal A^{(1:t')})$\;.
\begin{definition}\label{def-effective}
We define the high-dimensional effective fields of $\widehat\qv_u^{(t)}$, $\widehat\psiv_u^{(t)}$ and $\widehat\phiv_u^{(t)}$ in \eqref{hauseholderep} as
 \begin{subequations}
\label{qvs}
\begin{align}
\qv_{u}^{(t)}&=\sum_{1\leq s\leq t}{\widetilde{{\mathcal B}}}_u^{(s,t)}
\gv_u^{(s)}\label{qve}\\
\psiv_{u}^{(t)}&=\Mm_u^{(t)}\nv+(\Id-\Mm_u^{(t)}{\Qm}_u)\qv_{u}^{(t)}-\Mm_u^{(t)}\sum_{u'\neq u}{\Qm}_{u'}{\qv}_{u'}^{(t)}\;.\\
\phiv_{u}^{(t)}&= \sum_{1\leq s\leq t}{{{\mathcal B}}}_u^{(s,t)}
\widetilde\gv_u^{(s)}
\end{align}
\end{subequations}
where we introduce the coefficients ${\mathcal B}_u^{(s,t)}$ and $\widetilde{\mathcal B}_u^{(s,t)}$ through the Cholesky decompositions of the two time SE cross-correlation matrices (see Definition~\ref{SE_new}) as
 \begin{align}\label{Bdef}
\widetilde{\mathcal B}_u^{(1:T)}={\rm chol}(\widetilde{\mathcal C}_u^{(1:T)})~~\text{and}~~
\mathcal B_u^{(1:T)}={\rm chol}(\mathcal C_u^{(1:T)})\;.
\end{align}
Here, we note that while by construction it is clear that $\widetilde{\mathcal C}_u^{(1:T)} \geq \mathbf{0}$ it is somewhat implicit that $\mathcal C^{(1:T)} \geq \mathbf{0}$. On the other hand, by given that $\widetilde{\mathcal C}_u^{(t,s)} \equiv \mathbb{E}\big[\langle \mathbf{q}_u^{(t)}, \mathbf{q}_u^{(s)} \rangle\big]$ it follows that 
\[
\mathbb{E}\big[\langle \boldsymbol{\psi}_u^{(t)}, \boldsymbol{\psi}_u^{(s)} \rangle\big] \equiv\mathcal C_u^{(t,s)}.
\]
This clarifies the positive \emph{semi}-definiteness of $\widetilde{\mathcal C}_u^{(1:T)}$. 
Furthermore, by the positive-definiteness assumptions on $\widetilde{\mathcal C}_u^{(1:T)}$ [i.e. Assumption~\ref{asm_psd}], it follows that each ${\mathcal C}_u^{(1:T)}$ is positive-definite. Thereby, the Cholesky decompositions in \eqref{Bdef} are all unique. 
\end{definition}

We begin with the following auxiliary results.
\begin{lemma}\label{laux1}
For any $(u,t)\in[U]\times [T]$ we have (see Appendix~\ref{paux1})
\begin{align}
	\widetilde\phiv_u^{(t)}=\Op{\sqrt{N}}\quad \text{and}\quad
\widehat{\phiv}_u^{(t)}=\Op{\sqrt{N}}\;.
	\end{align}
\end{lemma}
In particular, these results imply that for any $(s,t,u)$
\begin{subequations}
 \label{aux1}
\begin{align}
\vert\langle \widetilde\vv_u^{(s)},\widetilde\phiv_u^{(t)}\rangle\vert&{\leq}  \frac{\norm{\widetilde\vv_u^{(s)}}}{\sqrt {N_u}}\frac{\norm{\widetilde\phiv_u^{(t)}}}{\sqrt {N_u}}=\frac{\norm{\widetilde\phiv^{(t)}}}{\sqrt{N_u}}=\Op{1}\\
\vert\langle \vv_u^{(s)},\widehat\psiv_u^{(t)}\rangle\vert&{\leq}  \frac{\norm{\vv_u^{(s)}}}{\sqrt {N_u}}\frac{\norm{\widehat\psiv_u^{(t)}}}{\sqrt {N_u}}=\frac{\norm{\widehat\phiv_u^{(t)}}}{\sqrt{N_u}}=\Op{1}
\end{align}    
\end{subequations}
where the inequalities follow from the Cauchy-Schwarz inequality.

Second, we have the following auxiliary result.
\begin{lemma}\label{laux2}
For any $(u,t)\in[U]\times [T]$ we have (see Appendix~\ref{paux2})
		\label{aux2}
		\begin{align}
			\vv_u^{(2t-1)}= \gv_u^{(t)}+\Op{1}\quad\text{and}\quad
			\widetilde\vv_u^{(2t)}= \widetilde\gv_u^{(t)}+\Op{1}\;.
		\end{align}   
\end{lemma}

Then, using these results together with the arithmetic properties of $\Op{\kappa}$ in Lemma~\ref{lemma:op_properties}, we obtain from \eqref{hauseholderep} that :
\begin{align}
\widehat{\qv}_u^{(t)}&=
\sum_{1\leq s\leq t}\widehat{{\widetilde{\mathcal B}}}_u^{(s,t)}
\gv_u^{(s)} +\matr\delta_{q_u}^{(t)}+\Op{1}\label{qv}\\
\widehat\phiv_u^{(t)}&= \sum_{1\leq s\leq t}\widehat{{\mathcal B}}^{(s,t)}_u\widetilde\gv_u^{(s)}+\matr \delta_{\phi_u}^{(t)}+\Op{1}\label{phiv}
\end{align}
where we have defined, for all $s\in[t]$,
\begin{align}
		\widehat{\widetilde{{\mathcal B}}}_u^{(s,t)}&\doteq \langle \widetilde\vv_u^{(2s-1)},\widetilde\phiv_u^{(t)} \rangle\\
		\widehat{{{\mathcal B}}}_u^{(s,t)}&\doteq \langle \vv_u^{(2s)},\widehat\psiv_u^{(t)}\rangle\\
		\matr \delta_{q_u}^{(t)}&\doteq\sum_{1\leq s\leq t}\vv_u^{(2s)}\langle \widetilde\gv_u^{(s)},{\widetilde\phiv}_u^{(t)} \rangle \\
		\matr \delta_{\phi_u}^{(t)}&\doteq \sum_{1\leq s\leq t}\widetilde\vv_u^{(2s-1)}\langle\gv_u^{(s)},\widehat\psiv_u^{(t)} \rangle.
	\end{align}	

In the proof, we essentially need to show the concentrations of the empirical dynamical order parameters $\widehat{\widetilde{{\mathcal B}}}_u^{(s,t)}$ and $\widehat{{{\mathcal B}}}_u^{(s,t)}$. To this end, it is convenient to express them by means of the Cholesky decomposition operation. Specifically, let us introduce the projected dynamics as
\begin{align}
\widetilde\phiv_{u,\perp}^{(t)}&\doteq (\Id_{N_u}-\frac{1}{N_u}\sum_{1\leq s<t}\widetilde\vv_u^{(2s)}(\widetilde\vv_u^{(2s)})^\top)\widetilde\phiv_u^{(t)}\\
\widehat\psiv_{u,\perp}^{(t)}&\doteq  (\Id_{N_u}-\frac{1}{N_u}\sum_{1\leq s\leq t}\vv_u^{(2s-1)}(\vv_u^{(2s-1)})^\top)\widehat\psiv_u^{(t)}\;.
\end{align}
Then, we define the empirical cross-correlations for $t,s\in[T]$:
\begin{align}
\widehat{\widetilde{\mathcal C}}_u^{(t,s)}\doteq 
 \langle {\widetilde {\phiv}}_{u,\perp}^{(t)}, {\widetilde {\phiv}}_{u,\perp}^{(s)} \rangle
\text{~~and~~} \widehat{\mathcal C}_u^{(t,s)}\doteq\langle\widehat{\psiv}_{u,\perp}^{(t)},\widehat{\psiv}_{u,\perp}^{(s)} \rangle \;.
\end{align}
Notice from \eqref{gsd} that 
 \begin{align}
\widehat{\widetilde{\mathcal B}}_u^{(1:T)}={\rm chol}(\widehat{\widetilde{\mathcal C}}_u^{(1:T)})~~\text{and}~~
\widehat{\mathcal B}_u^{(1:T)}={\rm chol}(\widehat{\mathcal C}^{(1:T)}_u)\;.\label{ebg2}
\end{align}

Let us introduce the hypothesis  ${\mathcal H}_{t'}$ for any $1\leq t\leq t'$
\begin{subequations}
 \begin{align}
 \widehat {{\mathcal B}}^{(1:t')}_u&= {\mathcal B}^{(1:t')}_u+\Op{L^{-\frac 1 2}}, ~~ \matr \delta_{\phi_u}^{(t)}=\Op{1}  \label{H22}\\
\widehat{\widetilde{\mathcal B}}_u^{(1:t')}&=
\widetilde{\mathcal B}^{(1:t')}_u+\Op{L^{-\frac 1 2}},~~
\matr \delta_{q_u}^{(t)}=\Op{1}\;\label{H21}.
	\end{align}   
\end{subequations}

We complete the proof of Theorem~\ref{th1} by verifying the induction step, i.e., ${\mathcal H}_{t'-1}\implies {\mathcal H}_{t'}$, together with the base case ${\mathcal H}_1$. These follow easily by using the following induction-step results:
\begin{lemma}\label{lemcon2}
		Let Assumption~\ref{as1} hold.  
		\begin{itemize}
			\item [(i)] If \eqref{H22} holds for $t\in[t'-1]$, then for $1\leq s\leq t\leq t'$ 
			\begin{subequations}
				\label{concen2}
				\begin{align}
					\langle \widetilde{\gv}_u^{(s)}, \widetilde{\phiv}_u^{(t)}\rangle&=
					\Op{L^{-\frac 1 2 }}\label{res21}\\
						\widehat{\widetilde{\mathcal C}}_u^{(t,s)}&= \widetilde{\mathcal C}_u^{(t,s)}+\Op{L^{-\frac 1 2 }}\label{res22}\;.
				\end{align}	
			\end{subequations}
			\item [(ii)] If \eqref{H21} holds for $t\in[t']$, then for $1\leq s\leq t\leq t'$ 
			\begin{subequations}
				\begin{align} 
					\langle{\gv}_u^{(s)},\widehat\psiv_u^{(t)} \rangle&=\Op{L^{-\frac 1 2}}\label{res24}\\
                \widehat{{\mathcal C}}_u^{(t,s)}&=\mathcal C^{(t,s)}_u+\Op{L^{-\frac 1 2}}\label{res23}\;.
				\end{align}
			\end{subequations}			
		\end{itemize}
	\end{lemma}

We begin by verifying the induction step, that is, ${\mathcal H}_{t'-1}\implies {\mathcal H}_{t'}$. From \eqref{res21}, for all $t\in[t']$,
\begin{equation}
    \matr\delta_{q_u}^{(t)}= \Op{1}\;.
\end{equation}
Moreover, applying Lemma~\ref{lemchol1} in Appendix~\ref{preliminariesop} to \eqref{res22} we obtain that
\begin{equation}
   \widehat{\widetilde{\mathcal B}}_u^{(1:t')}=
	\widetilde{\mathcal B}_u^{(1:t')}\; +\Op{L^{-\frac 1 2}}.
\end{equation}
In other words, ${\mathcal H}_{t'-1}$ implies that \eqref{H21} also holds for $t'$. 

Similarly, from \eqref{res24}, we have for all $t\in[t']$
\begin{equation}
    \matr\delta_{\phi_u}^{(t)}=\Op{1}\;.
\end{equation}
Moreover, applying Lemma~\ref{lemchol1} in Appendix~\ref{preliminariesop} to \eqref{res23}, we obtain that
\begin{equation}
   \widehat{{\mathcal B}}_u^{(1:t')}=
	{\mathcal B}_u^{(1:t')}\; +\Op{L^{-\frac 1 2}}.
\end{equation}
This completes the proof of the induction step. 

As for the proof of the base case $\mathcal{H}_1$, we note that $\matr{\delta}^{(t)}_{q_u} = \matr{0}$. Moreover,  we have
\begin{align}
 \widehat{\widetilde{\mathcal{B}}}_u^{(1,1)} &= \frac{\|\sv_u-\mathbb E[\sv_u]\|}{\sqrt{N_u}} \\ 
    &=\sqrt{\widetilde{\mathcal{C}}_u^{(1,1)}+\Op{N^{-\frac 1 2}}}=
    \widetilde{\mathcal B}_u^{(1,1)}+\Op{L^{-\frac 1 2}}.
\end{align}
 Hence, \eqref{H21} holds for $t=1$. Then, from Lemma~\ref{lemcon2}-(ii), we have 
\begin{align}
    \langle \gv_u^{(1)}, \widehat \psiv_u^{(1)} \rangle=\Op{L^{-\frac 1 2}}  &\implies \matr \delta_{q_u}^{(1)}= \Op{1}, \\  
    \widehat{\mathcal{C}}_u^{(1,1)}= \mathcal{C}_u^{(1,1)}+\Op{L^{-\frac 1 2}} &\implies \widehat{\mathcal{B}}_u^{(1,1)}= \mathcal{B}_u^{(1,1)}+\Op{L^{-\frac 1 2}}.
\end{align}
This completes the proof of Theorem~\ref{th1}.

\section{Conclusion and Outlook}\label{sec_conc}

In this paper we have introduced a signal-recovery framework for a multiuser linear-Gaussian 
communication system, dubbed MU-OAMP, in which $U$ users simultaneously transmit to a common receiver (e.g., a base station) using random precoding at the transmitters. The proposed framework accommodates generic channel matrices, right-unitarily invariant precoding matrices, and non-separable input signals.

We provide an explicit high-dimensional analysis of the proposed MU-OAMP algorithm in the case of generic non-separable signal models.
Furthermore, we analyze the associated inference problem using the replica symmetry (RS) ansatz and derive asymptotic expressions for both the mutual information and the per-user MMSE. The derivation under the RS ansatz is based on a novel disorder-averaging technique.
The finite-sample analysis is shown to coincide with the RS predictions, implying that---under the validity of the RS ansatz---the proposed algorithm is asymptotically Bayes-optimal.

A recent related work in the context of high-dimensional time series~\cite{tieplova2025information}, where the authors rigorously derived replica predictions for the mutual information under similar models involving structured channel matrices and random precoding (e.g., circulant and i.i.d.\ Gaussian matrices). However, no algorithm was proposed in that work to attain the Bayes-optimal performance predicted by theory. The present paper resolves this issue.
From a technical perspective, an important future direction is to establish the universality of the SE analysis beyond Haar-distributed unitary matrices, for instance by considering pseudo-random unitary constructions such as randomly signed Fourier matrices\cite{anderson2014asymptotically,hu2022universality,wang2022universality}.  

From a practical standpoint, an important limitation of the current work is the assumption of 
perfect channel knowledge at the receiver. Investigating the impact of imperfect channel knowledge, 
particularly under realistic pilot-based channel estimation schemes, remains an important direction for future research.

%%%%%%%%%%%%%%%%%%%%%%%%%%%%%%%%%%%

\appendices

%%%%%%%%%%%%%%%
\section{Preliminaries from High Dimensional Probability}\label{preliminariesop}
In this section, we present the useful properties of concentration inequalities in terms of the notion of $\Op{\kappa}$ introduced in Section \ref{o1_notation}. 
\begin{lemma}\cite{cakmakit25}\label{lemma:op_properties}
Consider the random variables ${A} = \Op{\kappa}$ and ${B} = \Op{\widetilde\kappa}$. Then, 
		\begin{align}
			{A + B} &= \Op{\max(\kappa, \widetilde \kappa)}\label{eq:op_sum}\\
			{AB}&= \Op{\kappa \widetilde \kappa}\label{eq:op_prod}		%\sqrt{A}&=\Op{\sqrt{\kappa}}.\label{eq:op_sqrt_pos}
		\end{align}
%\begin{proof}
%{\brk Note that for any small $\epsilon>0$ we have $\Vert A\Vert_{\mathcal L^p}\leq N^{\frac{\epsilon}{2}} C_p\kappa$ and $\Vert A\Vert_{\mathcal L^p}\leq N^{\frac \epsilon 2}C_p\widetilde\kappa$ for some constants $C_p$. Then, the results \eqref{eq:op_sum} and \eqref{eq:op_prod} follow from the Minkowski inequality  (i.e., $\Vert A+B \Vert_{\mathcal L^p}\leq\Vert A\Vert_{\mathcal L^p}+\Vert B\Vert_{\mathcal L^p}$) and H\"older’s inequality (i.e.,  $\Vert A B \Vert_{\mathcal L^p}\leq\Vert A\Vert_{\mathcal L^{2p}}\Vert B \Vert_{\mathcal L^{2p}}$), respectively. The result \eqref{eq:op_sqrt_pos} follows from the inequality $\Vert {A}^{1/2} \Vert_{\mathcal L^{p}}\leq \Vert {\rm A}^{1/2} \Vert_{\mathcal L^{2p}}=\Vert {\rm A} \Vert_{\mathcal L^{p}}^{\frac 1 2}$\;.}
%\end{proof}
	\end{lemma}

\begin{lemma}\cite{cakmakit25}
\label{lemchol1}
Consider $T\times T$ matrices $\widehat{\mathcal C}\geq \matr 0$ and $\mathcal C>\epsilon\mathcal I$ where $\norm{\widehat{\mathcal C}}_{\rm F}=\Op{1}$ and $\epsilon>0$ is independent of $N$.  Let $\norm{\widehat{\mathcal C}-\mathcal C}_{\rm F}=\Op{N^{-c}}$ for some constant $c>0$. Then, 
\begin{equation}
\norm{{\rm chol}(\widehat{\mathcal C})-{\rm chol}({\mathcal C})}_{\rm F}=\Op{N^{-c}}
\end{equation}
Here, for $\mathcal C\geq \matr 0$, let ${\mathcal B}={\rm chol}(\mathcal C)$ for $\mathcal C\geq \matr 0$ being the upper-triangular matrix such that $\mathcal C={\mathcal B}^\dagger{\mathcal B}$. 
\begin{proof}
Note that the entries of $\mathcal B = \operatorname{chol}(\mathcal C)$ satisfy
\begin{equation}
\mathcal B^{(s,t)} \mathcal B^{(s,s)}
=
\mathcal C^{(t,s)}
-
\sum_{s'=1}^{s-1}
\mathcal B^{(s',t)} \mathcal B^{(s',s)}.
\end{equation}
Note also that given $\mathcal C>\epsilon\mathcal I$ we have $\mathcal B^{(t,t)}>\sqrt{\epsilon}$ for each $t$. 
Then, the proof easily follows from the arithmetic properties of $\Op{\kappa}$ together with the fact $\sqrt{1+\Op{N^{-c}}}=1+\Op{N^{-c}}$ which follows from the bound $\vert 1- \sqrt{1+x}\vert=\frac{\vert x\vert}{1+\sqrt{1+x}}\leq \vert x\vert$\;. 
\end{proof}
 \end{lemma}
\begin{lemma}\label{lemmakey}
Consider a mapping \( h(\av;\theta):\mathbb{C}^{3N}\times \RR\to \mathbb{C} \) such that for every
\( \av,\bv \in \mathbb{C}^{3N} \) there exists a random variable \({\rm Lip}(\theta)>0\) such that
\[
\bigl| h(\av;\theta)-h(\bv;\theta) \bigr|
\leq
\frac{{{\rm Lip}}(\theta)}{\sqrt N}\left(1+\frac{\|\av\|}{\sqrt{N}}+\frac{\|\bv\|}{\sqrt{N}}\right)\|\av-\bv\|.
\]
Let the RVs \( \sv\in\mathbb{C}^{N} \) and \( \zv\in\mathbb{C}^{2N} \) and $\theta$ be mutually indepdent independent, with $\frac{1}{\sqrt N}\|\sv\| = \Op{1}$, ${\rm Lip}(\theta)=\Op{1}$ and $\zv \sim \mathcal{CN}(\mathbf{0},\Id_{2N})$.
Then,
\[
h(\sv,\zv;\theta)
=
\mathbb{E}\!\left[h(\sv,\zv;\theta)\mid \sv\right]
+
\Op{N^{-\frac 1 2}}\;.
\]
Here, the expectation is taken conditionally on $\mathbf{s}$ and implicitly on $\theta$. 
\end{lemma}
\begin{lemma}\label{conprod}
Consider the RVs $\av,\bv\in \CC^{N}$ where  $ \av\sim_\text{\rm i.i.d.} A$ and  $ \bv\sim_\text{\rm i.i.d.} B$  with ${A}=\Op{1}$, $B=\Op{1}$.
Let the random matrix $\Xm\in \CC^{N\times N}$ be independent of $\av$ and $\bv$ with $\Vert\Xm\Vert_{2}=\Op{1}$. 
Then, 
\begin{equation}      \langle\av,\Xm\bv\rangle=\langle\Xm\rangle \mathbb E[A^\star B]+ (\langle \matr 1_N, \Xm\matr 1_N \rangle-\langle \Xm\rangle) \mathbb E[A^\star]\mathbb E [B]
     +\Op{N^{-\frac 1 2}}\;. \label{eq_conprod}
\end{equation}
\end{lemma}

We defer the proofs of  Lemma~\ref{lemmakey} and Lemma~\ref{conprod}  to the following subsections. 

\subsection{Proof of Lemma~\ref{lemmakey}}
For convenience, we introduce
\[
L_{\sv,\theta} \doteq \frac{\sqrt{2}{\rm Lip}(\theta)}{\sqrt N}\!\left(1+2\frac{\|\sv\|}{\sqrt{N}}\right)=\Op{N^{-\frac 1 2}}.
\]
where the latter follows by the product rule of $\Op{\kappa}$ in Lemma~\ref{lemma:op_properties}.
For all \( \zv,\zv'\in\mathbb{C}^{2N} \), we have
\[
\bigl| h(\sv,\zv;\theta)-h(\sv,\zv';\theta) \bigr|
\leq
L_{\sv,\theta}
\left(1+\frac{\|\zv\|}{\sqrt{2N}}+\frac{\|\zv'\|}{\sqrt{2N}}\right)
\|\zv-\zv'\|.
\]

Conditioned on \( \sv,\theta \), we will verify that 
\begin{equation}
\delta(\zv)
\doteq 
\frac{1}{L_{\sv,\theta}}(h(\sv,\zv;\theta)
-
\mathbb{E}\!\left[h(\sv,\zv;\theta)\mid \sv\right]) \label{deltav}
\end{equation}
is a sub-exponential random variable \cite{vershynin2018high} such that 
\begin{equation}
\norm{\delta(\zv)}_{\mathcal L^p}\leq pC \;. \label{bound}
\end{equation}
with $C$ is a constant independent of \(\sv,\theta\). Then, the thesis evidently follows from $L_{\sv,\theta}=\Op{N^{-\frac 1 2}}$.

To this end, we first note that $\zv \equiv \frac{1}{\sqrt{2}}(\zv_1 + \mathrm{i}\zv_2)$ where the RVs \( \zv_1, \zv_2 \sim \mathcal{N}(\mathbf{0}, \Id_{2N}) \) are independent. Consequently, any function
\( \delta:\mathbb{C}^{2N} \to \mathbb{C} \) admits the decomposition
\[
\delta(\zv)  \equiv \delta_{\mathrm{R}}(\zv_1,\zv_2) + \mathrm{i} \delta_{\mathrm{I}}(\zv_1,\zv_2),
\]
where \( \delta_{\mathrm{R}}, \delta_{\mathrm{I}} : \mathbb{R}^{4N} \to \mathbb{R} \). Since \( \delta \) is pseudo-Lipschitz with constant \( 1\), both
\( \delta_{\mathrm{R}} \) and \( \delta_{\mathrm{I}} \) are pseudo-Lipschitz functions
on \( \mathbb{R}^{4N} \) with constant \( 1/\sqrt{2} \). 
Thus, without loss of generality, we can assume \( \zv \) is a real-valued standard Gaussian vector.

The rest is a straightforward adaptation of the concentration bound for Lipschitz functions; here we extend the argument to pseudo-Lipschitz functions. We begin from the final step in the proof of \cite[Theorem~2.1.12]{tao2023topics}\footnote{We note an typo in the corresponding step of \cite{tao2023topics}: the exponent $\exp\left(\frac{2}{\pi}(\cdot)\right)$ should instead read $\exp\left(\frac{\pi}{2}(\cdot)\right)$.} which gives that for every \(t \in \mathbb{R}\)
\begin{equation} \mathbb E[{\rm e}^{t\delta(\zv)}]\leq\mathbb E[{\rm e}^{\frac{\pi^2 t^2}{8}\norm{\nabla \delta(\zv)}^2}]\; \end{equation} where $\nabla$ is the Gradient operator. We have 
\[
\|\nabla \delta(\zv)\|^2
\le
\!\left(1+\sqrt{\frac 2 N}\|\zv\|\right)^2
\le
2\!\left(1+\frac{2}{N}{\|\zv\|^2}\right).
\]
Then, we get
\begin{align}
\mathbb{E}\!\left[{\rm e}^{t\delta(\zv)}\right]
&\le
{\rm e}^{\frac{\pi^2t^2}{4}}\mathbb{E}\!\left[
{\rm e}^{\frac{\pi^2t^2}{2N}{\|\zv\|^2}}\right] \\
&=
{\rm e}^{\frac{\pi^2t^2}{4}}
\left(1-\frac{\pi^2t^2}{N}\right)^{-2N},
\quad  \vert t\vert<\frac{\sqrt{N}}{\pi}.
\end{align}
Furthermore, using \( (1-x)^{-1/2}\le {\rm e}^{x} \) for \(0\le x\le 1/2\), we then get \[\mathbb{E}\!\left[e^{t\delta(\zv)}\right]
\le
{\rm e}^{\frac{17\pi^2}{4}t^2},\quad 
\vert t\vert\le \frac{\sqrt{N}}{\sqrt{2}\pi}\;.\] 
We then get the desired bound 
\begin{equation}
\mathbb{E}\!\left[e^{t\delta(\zv)}\right]
\le
{\rm e}^{\frac{17\pi^2}{4}t^2},
\qquad
\vert t\vert \le \frac{2}{\pi\sqrt{17}}. \label{bbb}
\end{equation}
 Equivalently, from  
 \cite[Proposition 2.8.1]{vershynin2018high} we have~\eqref{bound}.

\subsection{Proof of Lemma~\ref{conprod}}
The proof is based on the large-deviation bounds given in \cite[Lemma 7.10]{erdHos2017dynamical}. We first underline the following result: Let the  RVs $\xv,\av\in\CC^{N}$ be independent. Let $\xv=\Op{\sqrt N}$ and  $\av\sim_{\text{i.i.d}}A$ with $A=\Op{1}$. Then,  conditioned on $\xv$, we have from \cite[Eq.7.57]{erdHos2017dynamical} that
\begin{equation}
\langle\av,\xv\rangle=\mathbb E[A^*]\langle\matr 1_N,\xv \rangle+\Op{\norm{\xv}N^{-1}}\\
\;
\end{equation}
Since $\xv=\Op{\sqrt{N}}$, this implies that e.g. 
\begin{equation}
\langle\av,\xv\rangle=\mathbb E[A^*]\langle\matr 1_N,\xv \rangle+\Op{N^{-\frac 1 2}}\;. \label{meancleaner}
\end{equation}

From $\Vert\Xm\Vert_2=\Op{1}$, we have $\Vert \Xm\matr 1_N\Vert=\Op{\sqrt N}$. This then implies from \eqref{meancleaner} that 
\begin{align}
\quad  \langle\av,\matr X\matr 1_N \rangle&=\langle\matr 1_N,\matr X\matr 1_N\rangle\mathbb E[A^*]+\Op{N^{-\frac 1 2}}\;. 
\end{align}
Hence, to verify \eqref{eq_conprod}, without loss of generality we can assume $\mathbb E[A]=\mathbb E[B]=0$ and verify 
\begin{equation}
\langle\av,\Xm\bv\rangle=\langle\Xm\rangle \mathbb E[A^\star B]+\Op{N^{-\frac 1 2}}. 
\end{equation}
We then write 
    \begin{align}
     \langle\av,\Xm\bv\rangle-\langle\Xm\rangle \mathbb E[A^* B]
     &=\underbrace{\frac{1}{N}\sum_{i}X_{ii}(a_{i}^\star b_{i}-\mathbb E[A^*B])}_{\doteq S_1}+\underbrace{\frac{1}{N}\sum_{i\neq j}X_{ij}a_{i}^* b_{j}}_{\doteq S_2}\;. 
    \end{align}
For short let $\xv\doteq{\rm diag}(\Xm)$. From $\Vert\Xm\Vert_{\rm F}=\Op{\sqrt N}$, we have $\Vert \xv\Vert=\Op{\sqrt N}$ and $\Vert \Xm-{\rm diag}(\xv)\Vert_{\rm F}=\Op{\sqrt{N}}$. 
Note also that from \eqref{eq:op_prod} we have $A^\star B=\Op{1}$. 

Firstly, from the result~\eqref{meancleaner} we have
$   S_1=\Op{N^{-\frac 1 2 }}\label{S1}$. Furthermore, following the steps in the proof \cite[Lemma 7.10]{erdHos2017dynamical}, we obtain (conditioned on $\Xm$) 
\begin{equation}
   S_2=\Op{\Vert \Xm-{\rm diag}(\xv)\Vert_{\rm F}N^{-1}}\;.\label{S2}
 \end{equation}
This then implies from $\Vert \Xm-{\rm diag}(\xv)\Vert_{\rm F}=\Op{\sqrt{N}}$ that $S_2=\Op{N^{-\frac 1 2}}$.
\section{The Proofs of Auxiliary Lemmas in the proof of Theorem~\ref{th1}}\label{proof_auxilary}
\subsection{Proof of Lemma~\ref{lemmaseq}}\label{proof_lemmaseq}
For convenience, we fix \( u\in [U] \) and suppress the subscript \( u \), e.g., \( \Om \equiv \Om_u \), \( \widehat{\psiv}^{(t)} \equiv \widehat{\psiv}_u^{(t)} \), etc. Also, within this proof, we set \( N = N_u \) (rather than $N=\sum_u N_u$).

In the proof, we adaptively use the following representation of the Haar unitary matrices. 
	\begin{lemma}\cite{meckes2019random}\cite[Lemma~1]{lu2021householder}\label{conditioning}
Let $\vv^{(1)}\doteq \mathcal{GS}(\gv^{(1)})$ with  $\gv^{(1)}\sim\mathcal {CN}(\matr 0,\Id_N)$ and a RV $\widetilde\vv^{(1)}\in \CC^{N}$ with $\langle \widetilde\vv^{(1)},\widetilde\vv^{(1)}\rangle=1$ 
and a Haar unitary matrix ${\Om^{(N-1)}}\in \CC^{(N-1)\times (N-1)}$ be mutually independent. Then,  
		\begin{align}
			\Om&\doteq\frac{1}{N}\vv^{(1)}(\widetilde\vv^{(1)})^\dagger+\matr\Pi_{\vv^{(1)}}^\perp\Om^{(N-1)}(\matr\Pi_{\widetilde\vv^{(1)}}^\perp)^\dagger\label{haarrep}
		\end{align}
		is a Haar unitary and independent of $\widetilde\vv^{(1)}$. Here, we introduce the semi-unitary matrices $\matr\Pi_{\widetilde\vv^{(1)}}^\perp, \matr\Pi_{\vv^{(1)}}^\perp\in \CC^{N\times(N-1)}$ whose columns span the orthogonal complements of the column spans of $\widetilde\vv^{(1)}$ and $\vv^{(1)}$, respectively. E.g., we have
		\[\matr\Pi_{\vv^{(1)}}^\perp(\matr\Pi_{\vv^{(1)}}^\perp)^\dagger=\Id-\frac{1}{N}\vv^{(1)} (\vv^{(1)})^\dagger=\Pm_{\vv^{(1)}}^\perp.\]
	\end{lemma}

We begin with the iteration step $t=1$. We recall Lemma~\ref{conditioning} and set
\[\widetilde\vv^{(1)}\equiv \mathcal{GS}({\widetilde\phiv^{(1)}})\;.\] 
Hence, we have 
\begin{equation}
\Om \widetilde\phiv^{(1)}=\langle\widetilde\vv^{(1)},\widetilde\phiv^{(1)} \rangle\vv^{(1)}. 
\end{equation}
For the multiplications from left as $\Om^\dagger \widehat\psiv^{(1)}$, we introduce a new Gaussian RV $\widetilde\gv^{(1)}\sim\mathcal {CN}(\matr 0,\Id_N)$.  Notice that 
	\[(\matr\Pi_{\widetilde\vv^{(1)}}^\perp)^\dagger\widetilde\gv^{(1)} \sim \mathcal {CN}(\matr 0,\Id_{N-1})\] 
which is independent of $\widetilde\vv^{(1)}$. We then set $\vv^{(2)}\equiv \GS{\widehat \psiv^{(1)}}{\vv^{(1)}}$ and use Lemma~\ref{conditioning} to represent $\Om^{(N-1)}$ in \eqref{haarrep} as 
	\begin{align}
		\Om^{(N-1)}&=\frac{1}{N}\mathcal{GS}((\matr\Pi_{\vv^{(1)}}^\perp)^\dagger \vv^{(2)})\mathcal{GS}((\matr\Pi_{\widetilde\vv^{(1)}}^\perp)^\dagger\widetilde\gv^{(1)})^\dagger +\matr\Pi_{(\matr\Pi_{\vv^{(1)}}^\perp)^\dagger\vv^{(2)}}^\perp 
		\Om^{(N-2)}\left(\matr\Pi_{(\matr\Pi_{\widetilde\vv^{(1)}}^\perp)^\dagger\widetilde\gv^{(1)}}^\perp \right)^\dagger\;
	\end{align}
	where $\Om^{(N-2)}\in \CC^{(N-2)\times (N-2)}$ is Haar unitary.
	Notice also that
	\begin{align}
		\matr\Pi_{\vv^{(1)}}^\perp\mathcal{GS}((\matr\Pi_{\vv^{(1)}}^\perp)^\dagger\vv^{(2)})&\equiv \vv^{(2)}\\
		\matr\Pi_{\widetilde\vv^{(1)}}^\perp\mathcal{GS}\label{gset}((\matr\Pi_{\widetilde\vv^{(1)}}^\perp)^\dagger\widetilde\gv^{(1)})
		&=\underbrace{\GS{\widetilde\gv^{(1)}}{\widetilde \vv^{(1)}}}_{\doteq \widetilde \vv^{(2)} }\;
	\end{align}
	Hence, we get the representation
	\begin{align}
		&\Om=\frac{1}{N} \vv^{(1)}(\widetilde\vv^{(1)})^\dagger  +\frac{1}{N} \vv^{(2)}(\widetilde \vv^{(2)})^{\dagger}+\underbrace{\left(\matr\Pi_{\vv^{(1)}}^\perp\matr\Pi_{(\matr\Pi_{\vv^{(1)}}^\perp)^\dagger\vv^{(2)}}^\perp \right)}_{\doteq \matr\Pi^\perp_{\vv^{(1:2)}}}
		\Om^{(N-2)}{\underbrace{\left(\matr\Pi_{\widetilde\vv^{(1)}}^\perp\matr\Pi_{(\matr\Pi_{\widetilde\vv^{(1)}}^\perp)^\dagger\widetilde\gv^{(1)}}^\perp \right)}_{\doteq \matr\Pi^\perp_{\widetilde\vv^{(1:2)}}}}^\dagger\;.\label{usethis2}
	\end{align}
	Here, e.g. one can show from \eqref{gset} that $\matr\Pi^\perp_{\vv^{(1:2)}}\in \CC^{N\times (N-2)}$ semi unitary matrix  such that 
	\begin{equation}
		\matr\Pi^\perp_{\vv^{(1:2)}}(\matr\Pi^\perp_{\vv^{(1:2)}})^\dagger=\Pm^\perp_{\vv^{(1:2)}}.
	\end{equation}
	So that we  get from \eqref{usethis2} 
	\begin{equation}
	\Om^\dagger\widehat\psiv^{(1)}=\langle\vv^{(1)},\widehat\psiv^{(1)}\rangle\widetilde\vv^{(1)} +\langle\vv^{(2)},\widehat\psiv^{(1)}\rangle\widetilde\vv^{(2)}\;.
	\end{equation}
	Hence, we complete the $\Om$-free representation for the first iteration step $t=1$. 
	
	Moving on to the second iteration step, we fix $t=2$ but mimic the notations in a manner that the arguments can be recalled for any $t>1$. Firstly, we generate arbitrary Gaussian elements, 
	\[\widetilde\gv^{(t)}\sim\gv^{(t)}~\sim\mathcal {CN}(\matr 0,\Id_N) \;.\]
	Second, we construct the semi unitary matrices
	\begin{align}
    \vv^{(2t-1)}&=\mathcal{GS}(\gv^{(t)}\vert \vv^{(1:2(t-1))})\;.\\
		\widetilde\vv^{(2t-1)}&=\mathcal{GS}(\widetilde\phiv^{(t)}\vert \widetilde\vv^{(1:2(t-1))}) \label{constphiv}
	\end{align}
	Notice that 
	\begin{equation}
		(\matr\Pi_{\widetilde\vv^{(1:2(t-1)}}^\perp)^\dagger\widetilde\gv^{(t)} \sim_{\text{i.i.d.}}\mathcal {CN}(\matr 0,\Id_{N-t})
	\end{equation}
	and the product is independent of $\widetilde\vv^{(1:2(t-1)}$. Moreover, 
	\begin{align}
		\matr\Pi_{\vv^{(1:2(t-1))}}^\perp\mathcal{GS}((\matr\Pi_{\vv^{(1:2(t-1))}}^\perp)^\dagger\vv^{(2t-1)})&\equiv \vv^{(2t-1)}\\
		\matr\Pi_{\widetilde\vv^{(2(t-1))}}^\perp\mathcal{GS}((\matr\Pi_{\widetilde\vv^{(2(t-1))}}^\perp)^\dagger\widetilde\gv^{(t)})
		&\equiv \widetilde\vv^{(2t-1)}\;.
	\end{align}
	Hence, similar to \eqref{gset} we have the representation   
	\begin{align}
		\Om&=\frac{1}{N} \sum_{1\leq s\leq 2t-1}\vv^{(s)}(\widetilde\vv^{(s)})^\dagger+{\matr\Pi^\perp_{\vv^{(1:2t-1)}}}
		\Om^{(N-2t+1)}({\matr\Pi^\perp_{\widetilde\vv^{(1:2t-1)}}})^\dagger\;.\label{usethisend1}
	\end{align}
	Note that by the construction \eqref{constphiv} , we have 
	\begin{align}
		\widetilde\phiv^{(t)}&=\sum_{1\leq s\leq 2t-1}\langle \widetilde\vv^{(s)},\widetilde\phiv^{(t)} \rangle\widetilde\vv^{(s)}
	\end{align}
	Hence, we get from \eqref{usethisend1} that
	\begin{align}
		\Om\widetilde\phiv^{(t)}&=\sum_{1\leq s\leq 2t-1}\langle \widetilde\vv^{(s)},\widetilde\phiv^{(t)} \rangle\vv^{(t)}\;. 
	\end{align}
	
	Finally, to have the $\Om$-free representation of $\Om^\dagger\widehat\psiv^{(t)}$, we construct 
	\begin{align}
		\vv^{(2t)}&=\mathcal{GS}(\widehat\psiv^{(t)}\vert \vv^{(1:2t-1)})\label{constpsi}\\
		\widetilde\vv^{(2t)}&=\mathcal{GS}(\widetilde\gv^{(t)}\vert \widetilde \vv^{(1:2t-1)})\;.
	\end{align}
	Similar to \eqref{usethisend1} we obtain the representation 
	\begin{align}
		\Om&=\frac{1}{N} \sum_{1\leq s\leq 2t}\vv^{(s)}(\widetilde\vv^{(s)})^\dagger+{\matr\Pi^\perp_{\vv^{(1:2t)}}}
		\Om^{(N-2t)}({\matr\Pi^\perp_{\widetilde\vv^{(1:2t)}}})^\dagger\;.\label{usethisend2}
	\end{align}Note that by the construction \eqref{constpsi}, we have 
	\begin{align}
		\widehat\psiv^{(t)}&=\sum_{1\leq s\leq 2t}\langle \vv^{(s)},\psiv^{(t)} \rangle\vv^{(t)}\;.
	\end{align}
	Hence, we get from \eqref{usethisend1} that
	\begin{align}
		\Om^\dagger\widehat\psiv^{(t)}&=\sum_{1\leq s\leq 2t}\langle\vv^{(s)},\psiv^{(t)} \rangle \widetilde\vv^{(t)}
	\end{align}

Clearly, for the iteration steps $t=3,4,\cdots$ we can repeat the same arguments as for $t=2$ which completes the proof of Lemma~\ref{lemmaseq}.

\subsection{Proof of Lemma~\ref{laux1}}\label{paux1}

For convenience, we work on the normalized dynamics: for all $(u,t)\in[U]\times [T]$ let
\begin{align}
	\widetilde \rv_u^{(t)}\doteq \frac{\rv_u^{(t)}}{\sqrt L} ~~\text{and }~~\widetilde \fv_u^{(t)}\doteq \frac{\fv_u^{(t)}}{\sqrt L}\;.
\end{align}
Similarly, we define $\widetilde \yv\doteq\yv/\sqrt{L}$, etc. We note  from the original iterations in \eqref{oamp0} that
\begin{subequations}\label{tildeAMP}
	\begin{align}
		\Vert \widetilde {\zv}^{(t)}\Vert&\leq\Vert \matr\Sigma^{(t)}\Vert_2(\Vert \widetilde\yv\Vert +\sum_u\Vert\Am_u\Vert_{2}\Vert \widetilde\fv_u^{(t)}\Vert)\\
		\Vert \widetilde {\rv}_u^{(t)}\Vert&\leq \Vert\Am_u\Vert_{2}\Vert \widetilde\zv^{(t)}\Vert+\Vert \widetilde\fv_u^{(t)} \Vert\\
		\Vert \widetilde \fv_u^{(t+1)}\Vert&\leq {\rm Lip}(\tau_u^{(t)})\Vert \widetilde {\rv}_u^{(t)}\Vert
	\end{align} 
\end{subequations}
Here, the latter inequality follows Lemma~\ref{lemmaft} below and  ${\rm Lip}(\tau_u^{(t)})=\Op{1}$.

Moreover, we recall the assumptions that $\Vert\Xm_u\Vert_{2}=\Op{1}$ for $\Xm\in\{\Am_u,\matr\Sigma^{(t)}\}$.
For example, this implies from the arithmetic properties of the notion of $\Op{\kappa}$ Lemma~\ref{lemma:op_properties} in Appendix~\ref{preliminariesop} that 
\begin{equation}
\Vert \widetilde\yv\Vert\leq \sum_{u\in U}\Vert \Am_u\Vert_2   \Vert \widetilde \sv_u\Vert+\Vert \widetilde\nv\Vert=\Op{1}\;.
\end{equation}
Here, we use a simple application of Lemma~\ref{conprod} in Appendix~\ref{preliminariesop} that $\Vert\nv\Vert=\Op{\sqrt{L}}$; so that $\Vert\widetilde\nv\Vert=\Op{1}$. Then, from \eqref{tildeAMP} it follows inductively (over iteration steps) that all the dynamics in \eqref{tildeAMP} belong to the family of $\Op{1}$, i.e., $\widetilde {\matr\rv}_u^{(t)}=\Op{1}$, $\widetilde\fv_u^{(t)}=\Op{1}$, etc. Note the by the additive property of $\Op{\kappa}$ in Lemma~\ref{lemma:op_properties}  this implies that e.g. $\widehat\phiv_u^{(t)}=\rv_u^{(t)}-\sv_u=\Op{\sqrt L}$. Similarly, we have $\widetilde\phiv_u^{(t)}=\Op{\sqrt L}$.

\begin{lemma}\label{lemmaft}
Under the premises of Theorem~\ref{th1}, for any $(u,t)\in[U]\times[T]$, we have
\begin{equation}
\|f_{u,t}(\mathbf{r}) - f_{u,t}(\mathbf{r}')\|
\leq \mathrm{Lip}\!\left(\tau_u^{(t)}\right)\|\mathbf{r} - \mathbf{r}'\|,
\quad \text{with} \quad
\mathrm{Lip}\!\left(\tau_u^{(t)}\right) = \Op{1},
\label{lipf}
\end{equation}
for all $\mathbf{r}, \mathbf{r}' \in \mathbb{C}^{N_u}$. Furthermore, it holds that $
\frac{\|f_{u,t}(\mathbf{0})\|}{\sqrt{N_u}} = \Op{1}$. 

\begin{proof}
The proof proceeds by induction over the state evolution recursion~\eqref{SE_new}: Suppose that $\tau_u^{(t)} = \Op{1}$ and $\nu_u^{(t)} = \Op{1}$. Since
$\eta_u(\mathbf{r};\tau)$ is uniformly Lipschitz with respect to $\mathbf{r}$,
with Lipschitz constant $\mathrm{Lip} = O(\tau^{\pm p})$, it follows that
\eqref{lipf} holds at iteration $t$. This then implies $\lambda_u^{(t+1)} = \Op{1}$. Given
$\lambda_u^{(t+1)} = \Op{1}$ and Assumption~\ref{as1}-(iii), we obtain
$\tau_u^{(t+1)} = \Op{1}$ and consequently $\nu_u^{(t+1)} = \Op{1}$. Thus, by induction, we have for all $(u,t)\in[U]\times[T]$
\[
\lambda_u^{(t)} = \Op{1}, \quad
\tau_u^{(t)} = \Op{1}, \quad
\nu_u^{(t)} = \Op{1}
\]
which gives \eqref{lipf}. Finally, since $\frac{1}{\sqrt{N_u}}\|\eta_u(\mathbf{0}, \tau)\| = O(\tau^{\pm p})$, we have that $\frac{\|f_{u,t}(\mathbf{0})\|}{\sqrt{N_u}} = \Op{1}$.
\end{proof}
\end{lemma}
\subsection{Proof of Lemma \ref{laux2}}\label{paux2}
Let $\Pm\in \CC^{N\times N}$ be a projection matrix to some fixed $t$ (w.r.t. $N$)  dimensional subspace in $\CC^{N}$. Let $\Pm$ and $\gv\sim\mathcal {CN}(\matr 0,\Id_N)$ be mutually independent. Then, the thesis follows once we verify
\begin{align}
\vv\doteq \frac{(\Id_N-\Pm)\gv}{\frac{1}{\sqrt N}\norm{(\Id_N-\Pm)\gv}} \overset{?}{=} \gv+\Op{1}\label{vv}
\end{align}
Consider the eigenvalue decomposition $\Pm=\Um^\dagger {\rm diag}(\matr 1_{t},\matr 0_{N-t})\Um$. By the invariance $\Um\gv\sim \gv$ we get
\begin{align}
\norm{\Pm\gv} \equiv \norm{\gv'}~~\text{and}~~
\norm{(\Id-\Pm)\gv}\equiv\norm{\gv''}
\end{align}
where  $\gv'\sim\mathcal{CN}(\matr 0,\Id_{t})$ and
$\gv''\sim\mathcal{CN}(\matr 0,\Id_{N-t})$.  Now using definition \eqref{vv} and together with $\Pm\gv=\Op{1}$ we get 
\begin{align}
\frac{1}{\sqrt N}\norm{(\Id_N-\Pm)\gv}\vv=(\Id_N-\Pm)\gv= \gv+\Op{1}\;.\label{vv0}
\end{align}
Furthermore, we have from Lemma~\ref{lemma:op_properties} that 
	\begin{align}	
    \frac{1}{\sqrt N}\norm{(\Id_N-\Pm)\gv}
		&=\sqrt{1+\Op{N^{-\frac 1 2}}}\\&
		=1+\Op{N^{-\frac 1 2}}\;, \label{qm}
	\end{align}
where the later step uses the inequality $\vert 1- \sqrt{1+x}\vert\leq \vert x\vert$. Using \eqref{qm} in \eqref{vv0} gives the thesis.

\subsection{Proof of Lemma~\ref{lemcon2}}\label{paux3}
\textbf{Proof of (i).}

We use the concentration inequality for pseudo-Lipschitz (random) mappings stated in Lemma~\ref{lemmakey}. 
To this end, we suppress the index $u$ for notational convenience and define a mapping $f:\CC^{N}\times \RR_{+}\to \CC^N$ such that $f(\mathbf{r};\tau^{(t)})\equiv f_t(\mathbf{r})$ for all $t$.

We first establish a concentration result for a class of generic mappings. For $\mathbf{z}\in\mathbb{C}^{2N}$, define linear mappings $\Phi_i:\mathbb{C}^{2N}\to\mathbb{C}^N$ via
\[
\begin{bmatrix}
\Phi_1(\mathbf{z}) \\
\Phi_2(\mathbf{z})
\end{bmatrix}
\doteq (\mathcal{C}^{\frac{1}{2}} \otimes \mathbf{I}_N)\mathbf{z},
\]
where $\mathcal{C}\geq \matr 0$ is a $2\times 2$ random matrix with entries satisfying $\mathcal{C}_{ij} = \Op{1}$. Then, we define 
\begin{align}
h_1(\mathbf{s},\mathbf{z})
&\doteq \langle f(\mathbf{s}+\Phi_1(\mathbf{z});\mathcal{C}_{11}), 
f(\mathbf{s}+\Phi_2(\mathbf{z});\mathcal{C}_{22}) \rangle, \\
h_2(\mathbf{s},\mathbf{z})
&\doteq \langle \Phi_1(\mathbf{z}), 
f(\mathbf{s}+\Phi_1(\mathbf{z});\mathcal{C}_{11}) \rangle, \\
h_3(\mathbf{s},\mathbf{z})
&\doteq \langle \mathbf{s}, f(\mathbf{s}+\Phi_1(\mathbf{z});\mathcal{C}_{11}) \rangle.
\label{st}
\end{align}
The function $f:\mathbb{C}^N\times\mathbb{R}_+\to\mathbb{C}^N$ satisfies, for all $\mathbf{r},\mathbf{r}'\in\mathbb{C}^N$ (see Lemma~\ref{lemmaft})
\begin{equation}
\|f(\mathbf{r};\tau)-f(\mathbf{r}';\tau)\|
\leq \mathrm{Lip}(\tau)\|\mathbf{r}-\mathbf{r}'\|,
\quad \text{with} \quad
\mathrm{Lip}(\tau)=\Op{1},
\end{equation}
and $\frac{1}{\sqrt{N}}\|f(\mathbf{0};\tau)\|=\Op{1}$. In particular, for any $\mathbf{z},\mathbf{z}'\in\mathbb{C}^{2N}$, one can verify the bound
\begin{align}
\big|h_i(\mathbf{s},\mathbf{z}) - h_i(\mathbf{s},\mathbf{z}')\big|
&\leq
\frac{\mathrm{Lip}(\mathcal{C})}{\sqrt{N}}
\left(1 + \frac{\|\mathbf{s}\|}{\sqrt{N}}
+ \frac{\|\mathbf{z}\|}{\sqrt{N}}
+ \frac{\|\mathbf{z}'\|}{\sqrt{N}}\right)
\|\mathbf{z}-\mathbf{z}'\|,
\label{usefulbound}
\end{align}
where $\mathrm{Lip}(\mathcal{C})=\Op{1}$. In particular, for $\mathbf{z}\sim\mathcal{CN}(\mathbf{0},\mathbf{I}_{2N})$ we have the concentrations from Lemma~\ref{lemmakey} as
\begin{align}
h_i(\mathbf{s},\mathbf{z})
&=
\mathbb{E}\!\left[h_i(\mathbf{s},\zv) \mid \mathbf{s}\right]
+\mathcal{O}(N^{-\frac 1 2})\\
&=\mathbb{E}\!\left[h_i(\mathbf{s},\zv)\right]
+\mathcal{O}(N^{-\frac 1 2})
\label{usethisbound}
\end{align}
where the latter step follows from Assumption~\ref{add_asm}. 

We first verify \eqref{res21}: By the premise of (i) and the Lipschitz continuity of $f_{u,t}$, we have
\[
\widetilde{\boldsymbol{\phi}}_u^{(t)}
=
\widetilde{\boldsymbol{\phi}}_{e,u}^{(t)}
+\mathcal{O}(1),
\qquad t\le t',
\]
where the effective random field is defined as
\begin{equation}
\widetilde{\boldsymbol{\phi}}_{e,u}^{(t)}
=
f_{u,t-1}(
\mathbf{s}_u
+
\underbrace{\sum_{s<t}\mathcal{B}_{u}^{(s,t-1)}\widetilde\gv_u^{(s)}}_{= \boldsymbol{\phi}_u^{(t-1)}})
-
\mathbf{s}_u.
\end{equation}
Consequently,
\begin{align}
\langle \widetilde\gv_u^{(s)}, \widetilde{\boldsymbol{\phi}}_u^{(t)} \rangle
&=
\langle \widetilde\gv_u^{(s)}, \widetilde{\boldsymbol{\phi}}_{e,u}^{(t)} \rangle
+\mathcal{O}(L^{-1/2}) \\
&=
\langle \widetilde{\gv}_u^{(s)},
f_{u,t-1}(\mathbf{s}_u+\boldsymbol{\phi}_u^{(t-1)}) \rangle
+\mathcal{O}(L^{-1/2}) \\
&\overset{(a)}{=}
\mathbb{E}[
\langle \widetilde{\gv}_u^{(s)},
f_{u,t-1}(\mathbf{s}_u+\boldsymbol{\phi}_u^{(t-1)}) \rangle
\mid \mathbf{s}_u
]
+\mathcal{O}(L^{-1/2}) \\
&=
\mathcal{B}_u^{(s,t-1)}
\mathbb{E}[
\langle f'_{u,t-1}(\mathbf{s}_u+\boldsymbol{\phi}_u^{(t-1)}) \rangle
\mid \mathbf{s}_u
]
+\mathcal{O}(L^{-1/2}) \\
&\overset{(b)}{=}
\mathcal{B}_u^{(s,t-1)}
\mathbb{E}[
\langle f'_{u,t-1}(\mathbf{s}_u+\boldsymbol{\phi}_u^{(t-1)}) \rangle
]
+\mathcal{O}(L^{-1/2}) \\
&\overset{(c)}{=}
\mathcal{O}(L^{-1/2}),\label{mc1}
\end{align}
where (a) follows from \eqref{usethisbound}, (b) follows from Assumption~\ref{add_asm}, and (c) follows from 
\[
\mathbb{E}[
\langle f'_{u,t}(\mathbf{s}_u+\boldsymbol{\phi}_u^{(t)}) \rangle
]=0,\quad \forall (u,t).
\]

Now, from \eqref{mc1} and together with Lemma~\ref{aux2} we have $\langle \widetilde\vv_u^{(2s)}, \widetilde{\boldsymbol{\phi}}_u^{(t)} \rangle=\Op{N^{-1/2}}$. This then gives for any $t\leq t'$
\begin{align}
\widetilde{\boldsymbol{\phi}}_{u,\perp}^{(t)}=\widetilde{\boldsymbol{\phi}}_u^{(t)}+\Op{1}=   \widetilde{\boldsymbol{\phi}}_{e,u}^{(t)}+\Op{1}
\end{align}
Then, \eqref{res22} follows from the concentration result \eqref{usethisbound}.

\textbf{Proof of (ii)}
 By the premise of (ii), we have
 \[
\widehat{\psiv}_u^{(t)}
=
\psiv_{u}^{(t)}
+\mathcal{O}(1),
\qquad t\leq t',
\]
where the effective random field $\psiv_{u}^{(t)}$ is as defined in Definition~\ref{def-effective}. 

Firstly, we have
\begin{align}
\langle \gv_u^{(t)},\widehat{\psiv}_u^{(s)}\rangle&\overset{(a)}{=}\mathbb E[\langle \gv_u^{(s)},{\psiv}_{u}^{(t)}\rangle]+\Op{L^{-\frac 1 2}}\\
&=\mathbb E[\langle \gv_u^{(s)},(\Id-\Mm_u^{(t)}\Qm_u)\qv_{u}^{(t)}\rangle]+\Op{N^{-\frac 1 2}}\\
&=\widetilde{\mathcal B}_u^{(s,t)}\langle\Id-\Mm_u^{(t)}\Qm_u\rangle+\Op{L^{-\frac 1 2}}\\
&= \Op{L^{-\frac 1 2}}\label{mc2}
\end{align}
where the step $(a)$  follows from Lemma~\ref{conprod}.
The latter step follows from the fact that \[\langle\Id-\Mm_u^{(t)}\Qm_u\rangle=0\].

From \eqref{mc2} and together with  Lemma~\ref{aux2} we have $\langle \vv_u^{(2s-1)}, \widehat{\boldsymbol{\psi}}_u^{(t)} \rangle=\Op{L^{-1/2}}$. This then gives for any $t\leq t'$
\begin{align}
\widehat{\boldsymbol{\psi}}_{u,\perp}^{(t)}=\widehat{\boldsymbol{\psi}}_u^{(t)}+\Op{1}=  {\boldsymbol{\psi}}_{u}^{(t)}+\Op{1}
\end{align}
Then, we have
\begin{align}
\langle \widehat{\psiv}_{u,\perp}^{(t)},\widehat{\psiv}_{u,\perp}^{(s)}\rangle &=\langle {\psiv}_{u}^{(t)},{\psiv}_{u}^{(s)}\rangle+\Op{L^{-1/2}}\\
&\overset{(a)}{=}\underbrace{\mathbb E[\langle {\psiv}_{u}^{(t)},{\psiv}_{u}^{(s)}\rangle]}_{=\mathcal C_u^{(t,s)}}+\Op{L^{-1/2}}
\end{align}
where the concentration result in the step $(a)$ follows from Lemma~\ref{conprod}.

\section{Proof of Remark~\ref{relax1} and Remark~\ref{relax2}}\label{relaxition}

\subsection{Proof of Remark~\ref{relax2}}
We recall the arithmetic properties of the finite-sample notion in Lemma~\ref{lemma:op_properties}. Similarly, we highlight the arithmetic properties of the asymptotic concentration notion as in \eqref{O1}. Specifically, for any random variables $A_N = \mathcal O_\infty{(1)}$ and $B_N = \mathcal O_\infty{(1)}$, we evidently have
\begin{subequations}
\label{arit}
\begin{align}
A_N + B_N &= \mathcal O_\infty{(1)}, \\
A_N B_N &= \mathcal O_\infty{(1)}, 
\end{align}
\end{subequations}
Moreover, as noted below \eqref{ascon}, if $A_N = \Op{N^{-\epsilon}}$ for any small constant $\epsilon > 0$, then we have $A_N \simeq 0$. Also, if $\widehat{\av} = \av + \Op{1}$ for $\av \in \CC^{N_u}$, then $\widehat{\av} \simeq \av.$

Clearly, by replacing the notion of $\Op{1}$ in Assumption~\ref{as1}-(iii) with $\mathcal{O}_\infty(1)$, and the concentration assumptions in Assumption~\eqref{add_asm} of the form $\widehat A = A + \Op{N^{-1/2}}$ with almost sure limiting concentrations $\widehat A \simeq A$, we can follow exactly the same steps as in Section~\ref{step2} and verify $\rv_{u}^{(t)} \simeq \sv_u + \phiv_u^{(t)}$.
\subsection{Proof of Remark~\ref{relax1}}
To bypass the singularity issues of the cross-correlation matrices,
we consider the perturbation idea of \cite[Section~5.4]{Berthier20}.  
Specifically, for each $(u,t)\in[U]\times[T]$, we introduce arbitrary the Gaussian RVs  $\matr \xi_u^{(t)}\sim \mathcal{CN}(\matr 0,\Id_{N_u})$ and 
perturb the OAMP dynamics \eqref{resoamp} as
\begin{subequations}\label{oamp}
		\begin{align}
        \rv_{u,\epsilon}^{(t)}&=\frac{\Am_u^\dagger \matr{\Sigma}^{(t)}}{\langle \Am_u^\dagger \matr{\Sigma}^{(t)} \Am_{u}\rangle}\left(\yv-\sum_{u'}\Am_{u'}{\fv}_{u',\epsilon}^{(t)}\right)+{\fv}_{u,\epsilon}^{(t)} \\
\fv_{u,\epsilon}^{(t+1)}&=f_{u,t,\epsilon}(\rv_{u,\epsilon}^{(t)})\; +\epsilon\matr\xi_u^{(t+1)}
		\end{align}    
	\end{subequations}
with $\fv_{u,\epsilon}^{(1)}=\fv_u^{(1)}+\epsilon \matr\xi_u^{(1)}$. Moreover, we define 
\begin{equation}
f_{u,t,\epsilon}(\xv)\doteq f_{u,t}(\xv)-\beta_{u,\epsilon}^{(t)}\xv
\end{equation}
where  
$\beta_{u,\epsilon}^{(t)}\in \CC$ is a deterministic sequence defined using the state-evolution of the perpetuated dynamics (defined in the following) to ensure the diverge-free property 
\begin{equation}
\mathbb E[\langle f'_{u,t,\epsilon}(\sv_u+\sqrt{\mathcal C_{u,\epsilon}^{(t,t)}}\zv_u)\rangle]=0\;.\label{dv2}
\end{equation}
In particular, we will subsequently prove for all $t\in[T]$  
\begin{equation}
\beta_{u,\epsilon}^{(t)}=o_N(\epsilon)\;.
\end{equation}
Here, in general we write $\matr\delta= o_N(\epsilon)$ for $\deltav\in \CC^{d}$ (we use $d=1$ and $d=N_u$) if
\begin{equation}
\lim_{\epsilon\to0}\lim_{N\to \infty}\frac{1}{\sqrt{d}}\Vert \deltav\Vert=0\;.
\end{equation}

In the steps of Section~\ref{paux1},
by substituting e.g., $\widetilde {\rv}_u^{(t)}$ with the difference $\rv_{u,\epsilon}^{(t)}-\rv_u^{(t)}$ and  $ \widetilde \fv_u^{(t)}$ with the difference $ \fv_{u,\epsilon}^{(t)}-{\fv}_{u}^{(t)}$, etc. one can verify inductively (over iteration steps) 
\begin{align}
\rv_{u,\epsilon}^{(t)}-\rv_u^{(t)}&=o_N(\epsilon)\label{usethis}\;.
\end{align}

Next, we introduce the perturbated version of the two-time state evolution in Definition~\ref{defSE}. Specifically, for each $u\in[U]$, let \(\{{\matr\phi}_{u,\epsilon}^{(t)} \in \mathbb C^{N_u}\}_{t \in [T]}\) be zero-mean Gaussian process independent of $\{\sv_u\}_{u\in[U]}$ such that ${\matr\phi}_{u,\epsilon}^{(t)}\sim_{\text{i.i.d.}}\Phi_{u,\epsilon}^{(t)}$
and the two-time covariances  $\mathcal C_{u,\epsilon}^{(t,s)} \doteq\mathbb E[(\Phi_{u,\epsilon}^{(t)})^*\Phi_{u,\epsilon}^{(s)}]$
are recursively constructed as
% \matr\Delta_\epsilon^{(t,s)} should be replaced by \matr\Delta_{u,\epsilon}^{(t,s)}, because this matrix is related to u
\begin{align}
\Rm_\epsilon^{(t,s)}&=\sigma^2\Id_L+\sum_u\widetilde{\mathcal C}_{u,\epsilon}^{(t,s)}\Am_u\Am_u^\dagger\nonumber 
\\
\mathcal C_{{u},\epsilon}^{(t,s)}&={\left\langle \Mm_u^{(s)} (\Rm_\epsilon^{(t,s)}) (\Mm_u^{(t)})^\dagger\right\rangle}
-\widetilde{\mathcal C}_{u,\epsilon }^{(t,s)}
\end{align}
Here, $\widetilde{\mathcal C}_{u,\epsilon}^{(t,s)}$ are deterministic such that we have as $N_u\to \infty$ 
\begin{align}
 \widetilde{\mathcal C}_{u,\epsilon}^{(t,s)}&= 
\mathbb E[\langle\widetilde{\matr\phi'}_{u,\epsilon}^{(t)}, \widetilde{\matr\phi'}_{u,\epsilon}^{(s)}\rangle]
\end{align}
where  $\widetilde{\matr\phi'}_{\epsilon}^{(1)}\doteq\fv_{u}^{(1)}-\sv_u+\epsilon\matr\xi_u^{(1)}$ and for $t\geq1$
\begin{align}
\widetilde{\phiv'}_{u,\epsilon}^{(t+1)}\doteq f_{u,t,\epsilon}(\sv_u+\matr\phi_{u,\epsilon}^{(t)})-\sv_u +\epsilon\matr\xi_u^{(t+1)}\;.
\end{align}   
Finally,  we define 
\begin{equation}
\beta_{u,\epsilon}^{(t)}\doteq\frac{1}{\mathcal C ^{(t,t)}_{u,\epsilon}}\mathbb E\langle\matr\phi_{u,\epsilon}^{(t)},
f_{u,t}(\sv_u+\matr\phi_{u,\epsilon}^{(t)})\rangle.\nonumber 
\end{equation}
Notice that by this definition, the diverge-free property \eqref{dv2} is fulfilled.

We now have the desired properties
\begin{align}
\widetilde{\mathcal C}^{(1:T)}_{u,\epsilon}>\epsilon^2  \Id_T\;
\end{align}
We then follow the same steps as in the proof of Theorem~\ref{th1} (simply adding the subscript $\epsilon$ to the notations) and verify that under the premises of Theorem~\ref{th1}
\begin{align}
\rv_{u,\epsilon}^{(t)}&\simeq \sv_u+{\matr\phi}_{u,\epsilon}^{(t)}\;.
\end{align}
We then have from \eqref{usethis},
\begin{align}
\rv_{u}^{(t)}&=\sv_u+{\matr\phi}_{u,\epsilon}^{(t)}+o_{N}(\epsilon) \;.
\end{align}

Let $\mathcal H_{t'}$ stand for the hypothesis that for all $t,s\in [t']$
\begin{align}
\widetilde{\mathcal C}_{u,\epsilon}^{(t,s)}&=\widetilde{\mathcal C}_{u}^{(t,s)}+o_N(\epsilon)\label{perturb1}
\\
\beta_{u,\epsilon}^{(t)}&=o_N(\epsilon)\label{perturb2}
\end{align}
We also note that \eqref{perturb1} implies that \[{\mathcal C}_{u,\epsilon}^{(t,s)}={\mathcal C}_{u}^{(t,s)}+o_N(\epsilon)\;.\] We will complete the proof by verifying that $\mathcal H_{t'}$ implies $\mathcal H_{t'+1}$ in the sequel. The proof of the base case $\mathcal H_1$ is analogous (and rather simpler). 

For convenience, we fix $u\in[U]$ and suppres the subscript $u$ from the notation, e.g. ${\mathcal C}_\epsilon^{(t,s)}\equiv {\mathcal C}_{u,\epsilon}^{(t,s)}$, $\matr\phi_{\epsilon}^{(t)}\equiv \matr\phi_{u,\epsilon}^{(t)}$, $\matr\phi^{(t)}\equiv \matr\phi_{u}^{(t)}$, etc. By the inductive hypothesis $\mathcal H_{t'}$ we have for $t,s\in[t']$
\begin{align}
&\left[\begin{array}{c}
     \matr\phi_{\epsilon}^{(t)} \\ 
     \matr\phi_{\epsilon}^{(s)} 
\end{array}\right]\doteq \left(\left[ \begin{array}{cc}
		{\mathcal C}_\epsilon^{(t,t)}&{\mathcal C}_\epsilon^{(t,s)}\\
		{\mathcal C}_\epsilon^{(s,t)}& {\mathcal C}_\epsilon^{(s,s)}
	\end{array}\right]\otimes \Id_N\right)^{\frac 1 2} \left[\begin{array}{c}
     \gv_1 \\
     \gv_2 
\end{array}\right]\nonumber \\
    &= \underbrace{\left(\left[ \begin{array}{cc}
		{\mathcal C}^{(t,t)}&{\mathcal C}^{(t,s)}\\
		{\mathcal C}^{(s,t)}& {\mathcal C}^{(s,s)}
	\end{array}\right]\otimes \Id_N\right)^{\frac 1 2}\left[\begin{array}{c}
     \gv_1 \\
     \gv_2 
\end{array}\right]}_{\doteq \left[\begin{array}{c}
     \matr\phi^{(t)} \\ 
     \matr\phi^{(s)} 
\end{array}\right]}+ o_N(\epsilon)\left[\begin{array}{c}
     \gv_1 \\
     \gv_2 
\end{array}\right]\label{good1}
\end{align}
where $\gv_1,\gv_2\sim \mathcal{CN}(\matr 0,\Id_N)$ are independent.  Moreover, we note that for any $t,s>1$
\begin{align}
\mathbb E[\langle\widetilde{\matr\phi'}_{\epsilon}^{(t+1)},\widetilde{\matr\phi'}_{\epsilon}^{(s+1)} \rangle]= \mathbb E[\langle\widetilde\phiv_{\epsilon}^{(t+1)},\widetilde{\matr\phi}_{\epsilon}^{(s+1)} \rangle]-\beta_{\epsilon}^{(t)}\beta_{\epsilon}^{(s)}\mathcal C_{\epsilon}^{(t,s)}
\end{align}
where we have defined for any $t>1$ $\widetilde{\matr\phi}_{\epsilon}^{(t)}\doteq f_{t-1}(\sv+\matr\phi_{\epsilon}^{(t)})-\sv\;.$
We note that 
\begin{align}
&\vert\langle\widetilde{\matr\phi}_{\epsilon}^{(t+1)},\widetilde{\matr\phi}_{\epsilon}^{(s+1)} \rangle-\langle\widetilde{\matr\phi}^{(t+1)},\widetilde{\matr\phi}^{(s+1)} \rangle\vert\nonumber\\
&\leq C(1+\frac{\Vert \sv \Vert}{\sqrt N}+\frac{\Vert \matr\phi_{\epsilon}^{(t,s)} \Vert}{\sqrt N}+ \frac{\Vert \matr\phi^{(t,s)} \Vert}{\sqrt N})\frac{ \Vert \matr\phi_{\epsilon}^{(t,s)}-\matr\phi^{(t,s)}\Vert}{\sqrt N}\label{endd} 
\end{align}
where for short, we write e.g., $\phiv^{(t,s)}\doteq (\phiv^{(t)},\phiv^{(s)})$ and 
From \eqref{good1} we note that 
\begin{equation}
\Vert\phiv_{\epsilon}^{(t,s)}-\phiv^{(t,s)}  \Vert\leq (\Vert\gv_1\Vert+\Vert \gv_2\Vert)o_N(\epsilon)\;.
\end{equation}
Hence, it evidently follows from \eqref{endd} that
$\mathcal H_{t'}$ implies that \eqref{perturb1} holds for any $t,s\in[t+1].$ 

Similarly, to verify the implication that  $\mathcal H_{t'}$ implies \eqref{perturb2} holds for $t'+1$ we note the bound 
\begin{align}
&\vert\langle\matr\phi_{\epsilon}^{(t)},
f_{t}(\sv+\matr\phi_{\epsilon}^{(t)})\rangle-\langle\matr\phi^{(t)},
f_{t}(\sv+\matr\phi^{(t)})\rangle\vert \nonumber \\
&\leq C(1+\frac{\Vert \sv \Vert}{\sqrt N}+\frac{\Vert \matr\phi_{\epsilon}^{(t)} \Vert}{\sqrt N}+ \frac{\Vert \matr\phi^{(t)} \Vert}{\sqrt N})\frac{ \Vert \matr\phi_{\epsilon}^{(t)}-\matr\phi^{(t)}\Vert}{\sqrt N} 
\end{align}
Note that $\mathbb E[\langle\matr\phi^{(t)},
f_{t}(\sv_u+\matr\phi^{(t)})\rangle]\simeq 0$. This implies $\beta_{\epsilon}^{(t+1)}=o_N(\epsilon)$ which completes the proof.

\section{The RS Ansatz}\label{free_energy}
We only derive the replica-symmetric (RS) prediction of the mutual information in \eqref{rs_result}, which coincides with the “free energy” of statistical physics up to an additive constant. Similarly, the MMSE prediction in \eqref{mmse} can be obtained within the same steps. Specifically, following the arguments in \cite[Appendix~E]{Tulino13}, we can introduce an auxiliary “external field” of the form ${\rm e}^{h g(\sv^{(1)}, \sv^{(2)}, \cdots, \sv^{(n)})}$ to couple replicated independent RVs $\sv^{(i)}$. By computing the resulting modified free energy, say $\mathcal I(h)$ and evaluating its derivative $I'(0)$, one obtains the desired input–output decoupling principle, of the form $(\sv_u; \sv_u + \sqrt{\tau_u}\zv)$ in the MMSE.

We note that $\mathcal I(\sv;\yv\vert \Am)=\mathcal I(\sv;\frac{1}{\sigma}\yv\vert \Am)$. Hence, we can substitute $\sigma^2\to 1$ and $\Am_u\to \frac{1}{\sigma}\Am_u$. Second, for convenience, we introduce $\widetilde \Am_u\doteq \Am_u\Om_u^\dagger$, i.e., 
\begin{equation}
\Am_u\equiv \widetilde\Am_u\Om_u
\end{equation}
where $\Om_u$ are independent Haar unitary random matrices.

We define the partition function, i.e., the probability of data matrix $\yv$, as 
\begin{equation}
	Z(\yv)\doteq \int {\rm dP}(\sv)\; {\rm e}^{-{L}\ln\pi+ -\norm{\yv-\sum_u\widetilde\Am_u\Om_u\sv_u}^2\;}\;.
\end{equation}
We make use of the typical assumption that the normalized mutual information is asymptotically (as $L, N_u\to \infty$) self-averaging (w.r.t. disorder matrices $\{\Om_u\}$) and study the limit
\begin{align}
	&\lim_{N_u,L\to \infty}\frac{1}{L}\mathbb E_{\{\Om_u\}}[\mathcal I(\sv;\yv)]=-1-\ln\pi-\lim_{N_u,L\to \infty}\frac{1}{L}\mathbb E_{\{\Om_u\}}\int {\rm d}\yv\; Z(\yv)\ln Z(\yv) \label{fenergy}\;.
\end{align}
Remark that
\begin{align}
\mathbb E_{\{\Om_u\}}\int {\rm d}\yv\; Z(\yv)\ln Z(\yv)&=\lim_{n\to 1}\frac{\partial}{\partial n}\ln\mathbb E_{\{\Om_u\}}\int {\rm d}\yv\; Z^{n}(\yv)\;.
\end{align}
Then, assuming we can exchange the order of the limits $n\to 1$ and $L\to \infty$, we write
\begin{equation}
	\mathcal I=-1- \ln\pi-\lim_{n\to 1}\frac{\partial}{\partial n}\underbrace{\lim_{N_u,L\to \infty}\frac 1 L\ln\mathbb E_{\{\Om_u\}}\int {\rm d}\yv\; Z^{n}(\yv)}_{\eqdef \mathcal M^{(n)}}\;.
\end{equation}
In the RS ansatz, we treat $n$ as an integer. In particular, using a standard Gaussian integration method one can verify that 
\begin{equation}
    \int {\rm d}\yv Z^{n}(\yv)={n^{L}}\int  \prod_{a\leq n}{\rm dP}(\sv^{(a)})\; {\rm e}^{- {\rm tr}[(\sum_{u}\widetilde\Am_u\Om_u\Sm_u)\Pm^\perp_{\matr 1}(\sum_{u}\widetilde\Am_u\Om_u\Sm_u)^\dagger]}
\end{equation}  
where for convenience we the $N_u\times n$ matrices $\Sm_u\eqdef[\sv_u^{(1)},\sv_u^{(2)},\ldots,\sv_u^{(n)}]$ and define the $n\times n$ projection matrix the orthogonal complements of $n$-dim. all one vector as $\Pm^\perp_{\matr 1}\doteq \Id_n-\frac{1}{n}\matr 1_n\matr 1_n^\top$.
\subsection{Disorder Average}\label{disorder_average}
Computing the disorder average $\mathbb E_{\{\Om_u\}} {\rm e}^{- {\rm tr}\left[\left(\sum_{u}\widetilde\Am_u\Om_u\Sm_u\right)\Pm^\perp_{\matr 1}\left(\sum_{u}\widetilde\Am_u\Om_u\Sm_u\right)^\dagger\right]}$ can be (asymptotically) evaluated using the asymptotic Itzykson–Zuber integral for $U=1$. However, when $U > 1$, the application of the Itzykson--Zuber approach is not straightforward. A similar challenge is treated for the simplified $T$-orthogonal setup \cite[Appendix~A]{Kabashima}, i.e., where $\widetilde\Am_u = \Id_N$ for all $u$, and the disorder average is handled by transforming the ''spin" variables into a suitable space in which the modified spins can be treated as uncorrelated. Instead, we introduce a new approach that is both generally applicable and straightforward.

Let us introduce the $n\times n$ empirical cross-correlation matrices as
\begin{equation}
 \mathcal Q_{u}=\frac{1}{N_u}\Sm_u^\dagger \Sm_u\;.
\end{equation}
We then write from ${\rm QR}$ decomposition of $\Sm_u$ as
\begin{equation}
\Sm_u=\sqrt{N_u}\widetilde\Vm_{u}{\rm chol}(\mathcal Q_{u})
\end{equation}
such that $\frac{1}{N_u}
\widetilde\Vm_u^\dagger\widetilde\Vm_u=\Id_n$. Above ${\rm chol}(\cdot)$ stands for the cholesky decomposition of the matrix in the argument.  For each $u$, let $\Vm_u\in\CC^{N_u\times n}$ be an arbitrary \emph{bi-unitarily invariant} matrix such that $\frac{1}{N_u}\Vm_{u}^\dagger\Vm_{u}=\Id_n$ %i.e., $\frac{1}{\sqrt N}\Vm_u$ is uniformly distributed on the Stiefel manifold $\mathcal V_{n}(N_u):\doteq\{\Xm: \Xm_u^\dagger\Xm_u=\Id_n\}$.
Furthermore,  let  $\Om_u^{(N_u-n)}$ is a $(N_u-n)\times (N_u-n)$ be an arbitrary Haar unitary matrix.  Then, from \cite[Lemma~1]{ccakmak2024multi} we have that
\begin{equation}
\Om_u\equiv\frac{1}{{N_u}}\Vm_u \widetilde\Vm_u^\dagger + \matr\Pi_{\Vm_u}^\perp\Om_u^{(N-n)}(\matr\Pi_{\Sm_u}^\perp)^\dagger\label{haarrepnew}
\end{equation}
is Haar unitary and independent of $\widetilde\Vm_u$. Here, $ \matr\Pi_{\Vm_u^{(1)}}^\perp$ and $ \matr\Pi_{\Sm_u}^\perp$ are semi-unitary matrices whose columns span the orthogonal complements of the column spans of $\Vm_u$ and $\widetilde\Vm_u$, respectively. E.g., we have $\matr\Pi_{\Vm_u}^\perp(\matr\Pi_{\Vm_u}^\perp)^\dagger=\Id_{N_u}-\frac{1}{N_u}\Vm_u \Vm_u^\dagger$.
		
From the representation \eqref{haarrepnew} we have 
\begin{align}
\matr\Gamma_u\doteq \Om_u \Sm_u=\sqrt{N_u}\Vm_u{\rm chol}(\mathcal Q_{u})
\end{align}
Hence, each $\matr\Gamma_u$ is independent with the density function conditioned on $\mathcal Q_u$
\begin{equation}
p(\matr\Gamma_u)=\frac{1}{Z(\mathcal Q_{u})}\delta \left( \matr\Gamma_u^\dagger\matr\Gamma_u- N_u\mathcal Q_{u}\right)\; 
\end{equation}
where the normalization constant reads as (see \cite[Theorem 2.1.14]{muirhead2009aspects})
\begin{equation}
\ln Z(\mathcal Q_u)= (N_u-n-1)\ln\vert N_u\mathcal Q_u\vert +N_un\ln \pi-\sum_{i\leq n}\ln \Gamma(N_u-i+1)\;. \label{exact}
\end{equation}
with $\Gamma(\cdot)$ denoting the Gamma-function.  Thus, we have
\begin{align}
&\mathbb E_{\{\Om_u\}} {\rm e}^{- {\rm tr}\left[\left(\sum_{u}\widetilde\Am_u\Om_u\Sm_u\right)\Pm^\perp_{\matr 1}\left(\sum_{u}\widetilde\Am_u\Om_u\Sm_u\right)^\dagger\right]}={\rm e}^{-\sum_u\ln Z(\mathcal Q_u)}\times\nonumber \\
&\times\int\prod_{u}\delta( \matr\Gamma_u^\dagger\matr\Gamma_u- N_u\mathcal Q_{u}){\rm d\matr\Gamma_u}\;{\rm e}^{- {\rm tr}\left[\left(\sum_{u}\widetilde\Am_u\matr\Gamma_u\right)\Pm^\perp_{\matr 1}\left(\sum_{u}\widetilde\Am_u\matr\Gamma_u\right)^\dagger\right]} \;.\label{disorder-simple}
\end{align}
We note that the Fourier representation of the Dirac-delta function for $\mathcal A,\mathcal B\in \CC^{n\times n}$ as
\begin{equation}
\delta(\mathcal A-\mathcal B)=\frac{1}{(2\pi{\rm i})^{n^2}}\int {\rm d}\widetilde {\mathcal Q} \;{\rm e}^{-{\rm tr}(\widetilde {\mathcal Q} (\mathcal A-\mathcal B))} 
\end{equation}
where integration is taken over $(t-{\rm i}\infty,t+{\rm i}\infty)^{n\times n}$ for some $t\in \RR$. Hence, the integral in \eqref{disorder-simple} can be easily computed asymptotically by the saddle-point method. 
\subsection{Saddle-Point Method}
Next, we will invoke \eqref{disorder-simple} together with the standard saddle-point method to compute $\mathcal M^{(n)}$. 

Firstly, we define the ``rate'' function 
\begin{align}
{\mathcal I}_u(\mathcal Q)\eqdef\lim_{N_u\to\infty}\frac{1}{N_u}\ln \int \prod_{a\leq n} {\rm dP}_u(\sv_u^{(a)})\;\delta (\Sm_u^\dagger\Sm_u-N_u\mathcal Q)\;\label{latdef}
\end{align}
Then, by the standard saddle-point method, we can express \eqref{latdef} as
\begin{align}
	\mathcal I_{u}(\mathcal Q)=\sup_{\widetilde{\mathcal Q}}\left({\rm tr}(\widetilde{\mathcal Q}\mathcal Q)+\mathcal G_u(\widetilde{\mathcal Q})\right)\;. \label{saddle1}
\end{align}
where we have introduced the moment-generating function
\begin{equation}
	\mathcal G_u(\widetilde{\mathcal Q})\eqdef\lim_{N_u\to \infty}\frac{1}{N_u}\ln\int \prod_{a\leq n} {\rm dP}_u(\sv_u^{(a)})\; {\rm e}^{-{\rm tr}(\Sm_u\widetilde{\mathcal Q}\Sm_u^\dagger)}\;.\label{eqc}
\end{equation}

Second, conditioned on $\mathcal Q_{u}$, we define
\begin{align}
\mathcal G(\{\mathcal Q_u\})=\lim_{N_u,L\to\infty} \frac{1}{L}\ln \mathbb E_{\{\Om_u\}} {\rm e}^{- {\rm tr}\left[\left(\sum_{u}\widetilde\Am_u\Om_u\Sm_u\right)\Pm^\perp_{\matr 1}\left(\sum_{u}\widetilde\Am_u\Om_u\Sm_u\right)^\dagger\right]}
\end{align}
Then, from the representation \eqref{disorder-simple}, it again follows by the saddle point method that 
\begin{align}
&\mathcal G(\{\mathcal Q_u\})= -n\ln\pi- n -\sum_{u}\alpha_u\ln\mathcal Q_u
\nonumber \\
&+\sup_{\widehat{\mathcal Q}}\left[\sum_u\alpha_u{\rm tr}( \widehat{\mathcal Q}_u\mathcal Q_u)-{\alpha_u}\ln\vert\widehat{\mathcal Q}_u\vert- \lim_{N_u,L\to \infty}\frac{1}{L}\ln \left\vert \Id_n\otimes \Id_L +\sum_u \widehat{\mathcal Q}_u^{-1}\Pm^\perp_{\matr 1}\otimes \Am_u\Am_u^\dagger\right\vert
\right]\;.\label{saddle2}
\end{align}
Above, as to the first line, we note from \eqref{exact} that
\begin{equation}
\lim_{N_u\to \infty}\frac{1}{N_u}\ln Z(\mathcal Q_u)=n(1+\ln \pi)+\ln\vert \mathcal Q_u \vert.
\end{equation}

Combining the two saddle-point results \eqref{saddle1} and \eqref{saddle2} we then have 
\begin{align}
\mathcal M(n)=n+n\ln\pi+&\sup_{\{\mathcal Q_u,\mathcal{\widetilde Q}_u,\widehat{\mathcal Q}_u\}}f(\{\mathcal Q_u,\widetilde{\mathcal Q}_u,\widehat{\mathcal Q}_u\}).\label{extrm}
\end{align}
where for convenience, we define 
\begin{align}
f(\{\mathcal Q_u,\widetilde{\mathcal Q}_u,\widehat{\mathcal Q}_u\})\doteq &\sum_{u}{\alpha_u}[\mathcal G_u(\widetilde{\mathcal Q}_u)+{\rm tr}(\widetilde{\mathcal Q}_u\mathcal Q_u) +{\rm tr}(\widehat{\mathcal Q}_u\mathcal Q_u)-\ln\vert\widehat{\mathcal Q}_u\mathcal Q_u\vert]\nonumber \\
&-\lim_{N_u,L\to \infty}\frac{1}{L}\ln \big\vert \Id_n\otimes \Id_L +\sum_u \widehat{\mathcal Q}_u^{-1}\Pm^\perp_{\matr 1}\otimes \Am_u\Am_u^\dagger\big\vert\;.
\end{align}

\subsection{Replica Symmetry Assumption}
We essentially need to simplify the optimization problem in \eqref{extrm}, which we resolve by invoking the replica-symmetry assumption. To this end, we first examine the properties of the stationary points of $\{\mathcal Q_u, \widetilde{\mathcal Q}_u, \widehat{\mathcal Q}_u\}$.

\begin{remark}\label{remrs1}
	Let $\mathcal Q_u^\star$, $\widetilde{\mathcal Q}_u^\star$, and $\widehat{\mathcal Q}_u^\star$ denote the optimizers that achieve the supremum in~\eqref{extrm}. By taking the gradients of the objective function with respect to the optimization variables and setting them equal to zero, we obtain the following stationary-point conditions\footnote{We make use the technical assumption by replacing the order of the gradients with the limiting expression $\lim_{N\to \infty}$ as $\frac{\partial}{\partial \mathcal Q}\lim_{N\to \infty}f_{N}(\mathcal Q)=\lim_{N\to \infty}\frac{\partial}{\partial \mathcal Q}f_{N}(\mathcal Q)$.}:
    \begin{subequations}
    \label{general_stationary}
     \begin{align}
	&\mathcal Q_u^{\star(a,b)}=\lim_{N_u\to\infty}\mathbb E[\langle \sv_u^{(a)},\sv_u^{(b)}\rangle]_{\mathcal G_u}\quad \label{rem11}\\
	&(\mathcal Q_u^\star)^{-1}=\mathcal{\widetilde Q}^\star_u+\mathcal{\widehat Q}^\star_u \label{rem12}\\
    &(\mathcal {\widehat Q}_u^\star)^{-1}= \mathcal Q_u^\star+(\mathcal {\widehat Q}_u^\star)^{-1}\lim_{N\to \infty}\frac{1}{N_u}{\rm tr}_{L}\left[(\Id_n\otimes \Id_L +\sum_u \widehat{\mathcal Q}_u^{-1}\Pm^\perp_{\matr 1}\otimes \Am_u\Am_u^\dagger)^{-1}\Pm^\perp_{\matr 1}\otimes \Am_u\Am_u^\dagger\right](\mathcal {\widehat Q}_u^\star)^{-1} \;.\label{rem13}
    \end{align}   
    \end{subequations}
	Here, $\mathbb E[(\cdot)]_{\mathcal G_u}$ denotes the expectation w.r.t. the moment generating function $\mathcal G_u(\widetilde{\mathcal Q}_u^{\star})$ and ${\rm tr}_L$ stands for the partial trace, which can be compactly defined as for an $\Xm\in\CC^{nL\times nL}$
    \begin{equation}
    {\rm tr}_{L}(\matr X)\doteq \mathbb E[(\Id_n\otimes \zv)^\herm\Xm(\Id_n\otimes \zv)]\in \CC^{n\times n}\;.
    \end{equation}
\end{remark}

We now postulate that, in the solution $\mathcal Q_u^{\star(a,b)}$, there are no preferred replica indices, i.e., replica-symmetry assumption: For all replica indices $a,b,c \in [n]$,
\begin{align}
\mathcal Q_u^{\star(a,a)} &= \mathcal Q_u^{\star(b,b)}, \quad a\neq b \\
\mathcal Q_u^{\star(a,b)} &= \mathcal Q_u^{\star(a,c)}, \quad a \neq b \neq c .
\end{align}
Equivalently, we introduce the $n\times n$ orthogonal matrix $\Om_n$ through the eigenvalue decomposition
\begin{equation}
\Pm^\perp_{\matr 1}=\Om_n {\rm diag}(0,1,\cdots,1)\Om_n^\top 
\end{equation}
and then we have the eigenvalue decomposition $\mathcal Q_u^\star$ of the form
\begin{subequations}
\label{symmetry}
\begin{align}
\mathcal Q_u^\star\equiv\Om_n{\rm diag}(q_u, \nu_u,\cdots ,\nu_u)\Om_n^\top\;.
\end{align}
Furthermore, from \eqref{general_stationary} we have the replica-symmetry solutions for $\widetilde{\mathcal Q}_u$ and $\widetilde{\mathcal Q}_u$ of the form
\begin{align}
 \widetilde{\mathcal Q}_u^\star&=\Om_n{\rm diag}(\widetilde q_u, 1/\tau_u,\cdots ,1/\tau_u)\Om_n^\top  \\
  \widehat{\mathcal Q}_u^\star&=\Om_n{\rm diag}(\widehat q_u, 1/\lambda_u,\cdots ,1/\lambda_u)\Om_n^\top\;. 
\end{align}
\end{subequations}
Here, for the sake of upcoming notational flow, we choose the notational setup of $1/ \tau_u$ and $1/ \lambda_u$. In particular,  \eqref{symmetry} simplifies the the  stationary-point conditions in \eqref{rem12} and \eqref{rem13} as
\begin{subequations}
\label{sol_new}
\begin{align}
\left[ \begin{array}{c}
\frac{1}{q_u}\\
\frac{1}{\nu_u}
\end{array}\right]&=\left[ \begin{array}{c}
\widetilde q_u\\
\frac{1}{\tau_u}
\end{array}\right]  +\left[ \begin{array}{c}
\widehat q_u\\
\frac{1}{\lambda_u}
\end{array}\right] \\
\left[ \begin{array}{c}
\frac{1}{\widehat q_u}\\
{\lambda_u}
\end{array}\right]&=\left[ \begin{array}{c}
q_u\\
\nu_u
\end{array}\right] +\left[ \begin{array}{c}
0\\
\lambda_u^2{\rm G}_u(\matr\lambda)
\end{array}\right] \;.
\end{align}
where $g_u$ is defined as in \eqref{limits}. 
Note that this solution already implies that 
\begin{equation}
\widehat q_u=\frac{1}{q_u} \quad \text{and}\quad  \widetilde q_u=0. \label{simple}
\end{equation}
\end{subequations}

From the replica-symmetry solution in~\eqref{sol_new}, we first note that 
\begin{subequations}
\begin{align}
{\rm tr}(\widetilde{\mathcal Q}_u\mathcal Q_u) +{\rm tr}(\widehat{\mathcal Q}_u\mathcal Q_u)-\ln\vert\widehat{\mathcal Q}_u\mathcal Q_u\vert&=n+(n-1)\ln{(1+\tau_u/\lambda_u)}\\
\lim_{N_u,L\to \infty}\frac{1}{L}\ln \vert \Id_n\otimes \Id_L +\sum_u \widehat{\mathcal Q}_u^{-1}\Pm^\perp_{\matr 1}\otimes \Am_u\Am_u^\dagger\vert&=(n-1)\lim_{N_u,L\to\infty}\frac{1}{L}\ln\vert
\mathbf{I}_L
+
\sum_{u \leq U}
\lambda_u \mathbf{A}_u \mathbf{A}_u^\dagger\vert\;.
\end{align}
\end{subequations}
From the replica-symmetry solution in~\eqref{sol_new} we have
\begin{align}
\widetilde{\mathcal Q}_u^\star=\frac{1}{\tau_u}\Pm^\perp_{\matr 1}\;.
\end{align}
Hence, by using the standard Gaussian integral identity, we get
\begin{align}
	\mathcal G_u(\widetilde{\mathcal Q}_u^\star)&\equiv\lim_{N\to\infty}\frac{1}{N_u} \ln\int \prod_{a=1}^{n}{\rm dP}_u(\sv^{(a)})\; {\rm e}^{-\frac{1}{\tau_u}{\rm tr}(\sv^{(1:n)}\Pm^\perp_{\matr 1}(\sv^{(1:n)})^\dagger)}\\
&=\ln\pi+\lim_{N_u\to\infty}\frac{1}{N_u}\ln \int {\rm d}\rv (\mathcal Z_u(\rv))^{n}\label{qhatsym}
\end{align}
where we have defined $\mathcal Z_u(\rv)\doteq 
	\int {\rm dP}_u(\sv)\;{\rm e}^{-\norm{\rv-\frac{1}{\sqrt{\tau_u}}\sv}^2}$.
Then, have 
\begin{align}
	\lim_{n\to 1}\frac{\mathcal G_u(\widetilde{\mathcal Q}_u^\star)}{\partial n}&=\lim_{N_u\to \infty}\frac{1}{N_u} \int {\rm d}\rv \mathcal Z_u(\rv)\ln \mathcal Z_u(\rv)\;
    \\&=-1- \lim_{N_u\to \infty}\frac{1}{N_u}\mathcal I(\sv_u;\sv_u+\sqrt{\tau_u}\zv). 
\end{align}
Hence, putting everything together, we complete the derivation as 
\begin{equation*}
\mathcal I=\inf_{\matr\tau,\matr\lambda}\lim_{N_u,L\to\infty}\frac{1}{L}\ln\big\vert
\mathbf{I}_L
+
\sum_{u \leq U}
\lambda_u \mathbf{A}_u \mathbf{A}_u^\dagger\big\vert+\sum_u\alpha_u\big[\lim_{N_u\to \infty}\frac{1}{N_u}\mathcal I(\sv_u;\sv_u+\sqrt{\tau_u}\zv)-\ln(1+\frac{\tau_u}{\lambda_u})\big]\;.
\end{equation*}
\section{Proof of Remark~\ref{remdarya}}\label{proof_drem}
As a first step, we use the deterministic-equivalent result \cite[Theorem~2]{couillet2011deterministic} to simplify the the corresponding $\ln\det$ expression as
\begin{align}
&\lim_{N_u,L \to\infty}\frac{1}{L}
\ln\big\vert
\mathbf{I}_L
+
\frac{1}{\sigma^2}
\sum_{u}
\lambda_u \mathbf{A}_u \mathbf{A}_u^\dagger\big\vert\\
&= \lim_{L\to \infty} \ln\big\vert
\mathbf{I}_L
+
\frac{1}{\sigma^2}
\sum_{u}\frac{\alpha_u\lambda_u}{1+e_u}\Hm_u \Hm_u^\dagger\big\vert +\sum_u\alpha_u[\ln (1+e_u)- \frac{e_u}{1+e_u}]
\end{align}
where $\{e_u\}$ are the unique solutions (given $\matr\lambda$) of 
\begin{align}
e_u=\lambda_u\lim_{L\to\infty}\frac{1}{L}{\rm tr}(\Hm_u\Hm_u^\dagger(\sigma^2\Id_L+\sum_{u'}\alpha_{u'}\frac{\lambda_{u'}}{1+e_{u'}}\Hm_{u'}\Hm_{u'}^\dagger)^{-1})\;. \label{eu}
\end{align}

By using the matrix inversion Lemma,  we represent the RS-fixed point equations \eqref{fixed}  as
\begin{subequations}
\label{Rss}
\begin{align}
\nu_u&\equiv\lim_{N,L_u\to\infty}\frac{1}{N_u}{\rm tr}\left(\frac{1}{\lambda_u}\Id_{N_u} + \Am_u^\dagger\left(\sigma^2\Id_L+\sum_{u'\neq u}\lambda_{u'}\Am_{u'}\Am_{u'}^\dagger\right)^{-1} \Am_u
\right)^{-1}\\
\frac{1}{\tau_u}&\equiv\frac{1}{\nu_u}-\frac{1}{\lambda_u}\\
\nu_u &\equiv{\rm mmse}(\sv_u\vert \sv_u+\sqrt{\tau_u}\zv)\;.
\end{align}    
\end{subequations}
In the sequel, we will verify that
\begin{equation}
\label{RS_new}
\begin{aligned}
\frac{1}{\tau_u}&=\lim_{L\to\infty}\frac{1}{L}{\rm tr}(\Hm_u\Hm_u^\dagger(\sigma^2\Id_N+\sum_{u'}\alpha_u\nu_{u'}\Hm_{u'}\Hm_{u'}^\dagger)^{-1})\;.\\
\nu_u &={\rm mmse}(\sv_u\vert \sv_u+\sqrt{\tau_u}\zv)\;.
\end{aligned}   
\end{equation}
In particular, for any given solution $\matr \lambda$ and $\matr\tau$ of \eqref{RS_new} we get the solution of $e_u$ in \eqref{eu} as 
\begin{equation}
e_u=\frac{\lambda_u}{\tau_u}~~\text{and}
~~\frac{\lambda_u}{1+e_u}=\nu_u\;.
\end{equation}
Hence, we will complete the proof  by deriving the simplified representation of \eqref{Rss} as \eqref{RS_new}. 

To this end, for convenience, we introduce
\begin{equation}
\Dm_u\doteq \matr\Hm_u^{-1}\left(\sigma^2\Id_L+\sum_{u'\neq u}\lambda_{u'}\Am_{u'}\Am_{u'}^\dagger\right)^{-1}(\Hm_u^{-1})^\dagger\;.
\end{equation}
Here, $\Hm^{-1}$ in general stands for a pseudo inverse of $\Hm_u$. 
Then, by construction, we have
\begin{align}
\frac{1}{\tau_u}&={{\rm R}}_{\Xim_u^\dagger \Dm_u^{-1}\Xim_u}(-\nu_u)\\
&=\int \frac{x {\rm d F}_{\Dm_u^{-1}}(x)}{1+\alpha_u\nu_ux}\label{MP}\\
&=\underbrace{\int \frac{{\rm d F}_{\Dm_u}(x)}{x+\alpha_u\nu_{u}}}_{={\rm G}_{\Dm_u}(-\alpha_u \nu_u)}\;.
\end{align}
Here, ${\rm R}_{(\cdot)}$ and ${\rm G}_{(\cdot)}$ denote the \emph{R-transform} and \emph{Stieltjes transform}, respectively, of the limiting spectral distribution of the matrix in the argument \cite{mingo2017free}, while ${\rm F}_{(\cdot)}$ denotes the limiting distribution itself. Note that the step~\eqref{MP} follows from a simple application of \emph{Marčenko--Pastur Theorem}.

We then simplify ${\rm G}_{\Dm_u}(-\alpha_u\nu_u)$ by invoking the deterministic-equivalent result  \cite[Theorem~1]{couillet2011deterministic} (see also \cite{speicher2012free}): Specifically, consider the random-matrix-free [i.e., $\Xim_u$-free] representation
\begin{align}
\widetilde\Dm_u\doteq \Hm_u^{-1}(\sigma^2\Id_N+\sum_{u'\neq u}\alpha_{u'}\frac{\lambda_{u'}}{1+\tilde e_{u'}}\Hm_{u'}\Hm_{u'}^\dagger)(\Hm_u^{-1})^\dagger
\end{align}
where each $\tilde e_u$ has a unique solution through the $U$ fixed-point equations 
\begin{equation}
\tilde e_u=\lambda_u\lim_{L\to\infty}\frac{1}{L}{\rm tr}(\Hm_u\Hm_u^\dagger(\sigma^2\Id_L+\sum_{u'\neq u}\alpha_{u'}\frac{\lambda_{u'}}{1+e_{u'}}\Hm_{u'}\Hm_{u'}^\dagger+\alpha_u \nu_u\Hm_u\Hm_u^\dagger)^{-1})\;.\label{e_u}
\end{equation}
From  \cite[Theorem~1]{couillet2011deterministic}, we have 
\begin{equation}
\frac{1}{\tau_{u}}={\rm G}_{\widetilde \Dm_u}(-\alpha_u\nu_u)\;. 
\end{equation}
Now, let us check the consistency of the solutions:
\begin{equation}
\tilde e_u=\frac{\nu_u}{\tau_u-\nu_u}\quad \forall u\;. \label{postulate}
\end{equation}
Recall that $\lambda_u=\tau_u\nu_u/(\tau_u-\nu_u)$; so that from \eqref{postulate} we get
\begin{equation}
\frac{\lambda_u}{1+\tilde e_{u}}=\nu_u.
\end{equation}
Indeed, from \eqref{e_u} we have $\tilde e_{u}=\lambda_u{\rm G}_{\widetilde \Dm_u}(-\alpha_u\nu_u)\;.$ This completes the derivation.

\section{Proof of Lemma~\ref{contraction}}\label{proof_rem_contradiction}
We introduce the mappings for $\xv\in[0,\infty)^U$ and $x\in (0,\infty)$
\begin{align}
{F}_{N_u}(\xv)&=\frac{1}{\langle\Am_u^\dagger \matr\Sigma(\xv)\Am_u\rangle}-{x_u}\\
{G}_{N_u}(x)&\doteq  \left(\frac{1}{N_u}{\rm mmse}(\sv_u\vert \frac{1}{\sqrt x}\sv_u+\zv)\right)^{-1}-\frac{1}{x}
\end{align}
where for convenience, we define  
\begin{align}
\matr\Sigma(\xv)&\doteq\left(\sigma^2 \mathbf{I}_L + \sum_u x_{u}\mathbf{A}_{u} \mathbf{A}_{u}^\dagger\right)^{-1}\;.
\end{align}
We note that the SE recursions simplify to
\begin{align}
\tau_u^{(t)}&={ F}_{N_u}(\matr\lambda^{(t)})\\
\lambda_u^{(t+1)}&=\frac{1}{{G}_{N_u}(\tau_u^{(t)})}\;.     
\end{align}

 We have the derivatives 
\begin{align}
\frac{\partial{{ F}_{N_u}(\xv)}}{\partial x_{u'}}&=\left(\frac{\langle \Am_u^\dagger \Am_{u'}\matr\Sigma(\xv)^{2}\Am_{u'}^\dagger\Am_{u}\rangle}{\langle \Am_u^\dagger\matr\Sigma(\xv)\Am_u\rangle^2}-\delta_{u,u'}\right)\label{derbound1} \\
\frac{\partial{{G}_{N_u}(x)}}{\partial x}&=\frac{1}{x^2}\left(1-\frac{\frac{1}{N_u}\mathbb E[\norm{\eta_{u,x}'(\sqrt{\frac{1}{x}}\sv_u+\zv)}_{ F}^2]}{( \frac{1}{N_u}\mathbb E[{\rm tr}(\eta_{u,x}'(\sqrt{\frac 1 x}\sv_u+\zv))])^2 }\right)\label{derbound2}
\end{align}
Here, we define $\eta_{u,\tau}'$ stands for the Jacobian matrix of the posterior mean as 
\begin{align}
\eta'_{u,x}(\gammav)\doteq\frac{\partial }{\partial \gammav} \mathbb E[\sv_u\vert \gammav= \sqrt{\frac{1}{x}}\sv_u+\zv]\;.
\end{align}
As a matter of fact, \eqref{derbound2} is obtained by verifying that 
$\eta'_{u,x}(\gammav)$ coincides with the posterior covariance matrix 
$\sv_u$ given the observation $\gammav = \sqrt{\frac{1}{x}} \sv_u + \zv$, i.e.,
\begin{align}
\eta'_{u,x}(\gammav)
= \mathbb{E}\!\left[
(\sv_u - \mathbb{E}[\sv_u \mid \gammav])
(\sv_u - \mathbb{E}[\sv_u \mid \gammav])^\dagger
\,\middle|\, \gammav
\right].
\end{align}
We then invoke the classical result of the MMSE mapping \cite[Eq. (16)]{reeves2019understanding}:
\begin{align}
\frac{\partial \, \mathrm{mmse}(\sv_u \mid \sqrt{\rho} \sv_u + \zv)}{\partial \rho}
= - \mathbb{E}\!\left[\norm{
(\sv_u - \mathbb{E}[\sv_u \mid \sqrt{\rho} \sv_u + \zv])
(\sv_u - \mathbb{E}[\sv_u \mid \sqrt{\rho} \sv_u + \zv])^\dagger}_{\rm F}^2
\right]\;.
\end{align}
Note that the first partial derivative reads zero if and only if we have for all $\matr \Am_u\matr \Am_u^\dagger=c_u\Id_L$ for some constant $c_u$; otherwise, we have 
$
\frac{\partial{{F}_{N_u}(\xv)}}{\partial x_u}>0\;. $

Second, $\frac{\partial{{G}_{N_u}(x)}}{\partial x}=0$ if and only if $\eta_{u,x}'(\gammav)\equiv c_{u}\Id_{N_u}$, i.e., $\eta_{u,x}(\gammav)=\av_u+c_u\gammav$ for some constant vector $\av_u$. This can only happen if $\sv_u$ is an isotropic Gaussian distribution. Hence, unless $\sv_u$ is an isotropic Gaussian distribution, we have 
$
\frac{\partial{{G}_{N_u}(x)}}{\partial x}<0\;. $
This completes the proof. 
\subsection{The asymptotic equivalence with the RS fixed-point solution}
We now introduce the limiting mappings for $\xv \in [0,\infty)^U$ and $x \in (0,\infty)$,
\begin{align}
{F}_u(\xv) &\doteq \lim_{N_u,L\to\infty} { F}_{N_u}(\xv), \\
{G}_u(x) &\doteq \lim_{N_u\to\infty} { G}_{N_u}(x),
\end{align}
which exists by the existence of the limits in \eqref{limits}. Starting from
$\lambda_u^{\star(1)} = {\rm M}_u(0)$, we define the recursion for solving the RS fixed-point equations \eqref{fixed} as
\begin{align}
\tau_u^{\star(t)} &= {  F}_u\bigl(\matr\lambda^{\star(t)}\bigr), \\
\lambda_u^{\star(t+1)} &= \frac{1}{{G}_u(\tau_u^{\star(t)})} .
\end{align}
At this stage, we note that
\begin{equation}
F_u(\xv) \leq F_u(\xv'), \qquad \xv \leq \xv',
\end{equation}
which follows directly from the fact that
${F}_{N_u}(\xv) \leq { F}_{N_u}(\xv')$.
Recall also the assumption $
a_u \doteq \lim_{N_u,L\to\infty} \frac{1}{N_u}\norm{\Am_u}_{\rm F}^2 > 0$. Hence, we obtain the universal lower bound
\begin{equation}
\tau_u^\star,\; \tau_u^{\star(t)} \geq {F}_u(\matr 0)
\equiv \frac{\sigma^2}{a_u}.
\end{equation}

Next, we prove by induction over the iteration index $t'$ that
\begin{align}
\matr\lambda^{(t')} &= \matr\lambda^{\star(t')} + \matr\epsilon(1),
\label{hh11}\\
\matr\tau^{(t')} &= \matr\tau^{\star(t')} + \matr\epsilon(1),
\label{hh22}
\end{align}
where $\matr\epsilon$ denotes a sequence indexed by $(N,L_u)$ such that
\[
\norm{\matr\epsilon} \overset{\rm a.s.}{\longrightarrow} 0
\quad \text{as } N,L_u \to \infty.
\]
By the contraction property of $\matr\lambda^{(t)}$ and $\matr\tau^{(t)}$
(see Remark~\ref{contraction}), this will then implies
\begin{align}
\matr 0 \leq \matr\lambda^{\star(t+1)} \leq \matr\lambda^{\star(t)}, \\
\matr 0 < \matr\tau^{\star(t+1)} \leq \matr\tau^{\star(t)} .
\end{align}
Hence, as $t \to \infty$, the sequences
$\matr\lambda^{\star(t)}$ and $\matr\tau^{\star(t)}$
converge to a solution of the RS fixed-point equations \eqref{fixed}, namely,
\begin{align*}
\tau_u^\star &= { F}_u(\matr\lambda^\star), \\
\lambda_u^\star &= \frac{1}{{G}_u(\tau_u^\star)} .
\end{align*}

Equation \eqref{hh11} holds evidently for $t' = 1$.
Suppose now that \eqref{hh11} holds for $t' = t$.
Then, by the mean-value theorem,
\begin{align}
{ F}_{N_u}(\matr\lambda^{(t)}) - {F}_{N_u}(\matr\lambda^{\star(t)})
=
\sum_{u' \leq U}
[{ F}'_{N_u,}\!(\tilde{\matr\lambda}^{(t)})]_u
\epsilon_{u'},
\end{align}
where for short  we define $\tilde{\matr\lambda}^{(t)} \doteq
q\,\matr\lambda^{(t)} + (1-q)\,\matr\lambda^{\star(t)}$ with $q \equiv q_{N_u} \in (0,1)$ such that
$q_{N_u} \to 1$ as $N_u,L \to \infty$.

The partial derivatives can be bounded as (see \eqref{derbound1})
\begin{align}
\bigl| [{F}'_{N_u}(\tilde{\matr\lambda}^{(t)})]_{u'}\bigr|
<
\frac{\frac{1}{N_u}\norm{\Am_u \Am_{u'}^\dagger}_{\rm F}^2}
{\left(\frac{1}{N_u}\norm{\Am_u}_{\rm F}^2\right)^2}
\left(
1 + \frac{1}{\sigma^2}
\sum_{u \leq U} \tilde{\lambda}_u^{(t)} \norm{\Am_u}_2^2
\right)^2
- \delta_{u,u'} .
\end{align}
Since $\norm{\Am_u}_2$ is almost surely bounded as $N_u,L \to \infty$ and $
\lim_{N_u,L\to\infty} \frac{1}{N_u}\norm{\Am_u}_{\rm F}^2 > 0$, we obtain
\begin{equation}
\limsup_{N_u,L\to\infty}
\norm{{F}'_{N_u}(\tilde{\matr\lambda}^{(t)})}
< \infty
\quad \text{a.s.}
\end{equation}
This establishes \eqref{hh22}, i.e.,
\begin{equation}
\lim_{N_u,L\to\infty}
\Bigl(
{F}_{N_u}(\matr\lambda^{(t)})
- {F}_{N_u}(\matr\lambda^{\star(t)})
\Bigr)
= 0 .
\end{equation}

Similarly, by the mean-value theorem, we write
\begin{align}
{\rm M}_{N_u}\!\left({\tau_u^{\star(t)}}\right)
-
{\rm M}_{N_u}\!\left({\tau_u^{(t)}}\right)
=
\frac{1}{N_u}
\mathbb E \big[
\norm{
\eta_{u,1/\rho}\!\big(
\sqrt{\rho}\,\sv_u + \zv
\big)
}_{\rm F}^2
\big]
\epsilon_u\;,
\end{align}
where for short $\rho \doteq
\frac{q}{\tau_u^{(t)}} + \frac{1-q}{\tau_u^{\star(t)}}$ with $q \equiv q_{N_u} \in (0,1)$ such that $q_{N_u} \to 1$ as $N_u \to \infty$. From the regularity condition on the posterior mean denoiser, we get
\begin{equation}
\limsup_{N_u\to\infty}
\frac{1}{N_u}
\mathbb E\big[
\norm{
\eta_{u,1/\rho}'\!\big(
\sqrt{\rho}\,\sv_u + \zv
\big)
}_{\rm F}^2
\big]
< \infty\;.
\end{equation}
Hence, it follows that \eqref{hh11} holds for $t' = t+1$, completing the induction.

\bibliographystyle{IEEEtran}
\bibliography{report}
\end{document}

%% file: macros.tex
\usepackage{amsmath,amssymb,dsfont,stfloats,color,url,bbm}
\usepackage{subfigure}
\usepackage{dsfont}
\usepackage{cite}
\usepackage{xcolor}
\usepackage{pgfplots, pgfplotstable}
\renewcommand\fbox{\fcolorbox{red}{white}}
\allowdisplaybreaks
\usepackage{mathtools}
\usepackage{cite}
\usepackage[margin=1in]{geometry}
\usepackage{amsmath}
\usepackage{amssymb}
\usepackage{amsthm}
\usepackage{latexsym}
\usepackage{verbatim}
\usepackage{subfigure}
\usepackage{psfrag}
\usepackage{color}
\usepackage{enumitem}

\usepackage[utf8]{inputenc} 
\usepackage{amsfonts} 
\usepackage[T1]{fontenc}
\usepackage{url}              % provides \url{...}

\usepackage{hyperref}
\hypersetup{
	colorlinks=true,
	linkcolor=blue,
	filecolor=blue,
	citecolor = blue,      
	urlcolor=cyan,}

\usepackage{esdiff}
\usepackage{pgf,tikz,psfrag}
\usepackage{yfonts}

\newtheorem{theorem}{Theorem}

\newtheorem{lemma}{Lemma}

\newtheorem{assumption}{Assumption}
\newtheorem{corollary}{Corollary}

\newtheorem{definition}{Definition}
\newtheorem{remark}{Remark}

\def\mathlette#1#2{{\mathchoice{\mbox{#1$\displaystyle #2$}}%
{\mbox{#1$\textstyle #2$}}%
{\mbox{#1$\scriptstyle #2$}}%
{\mbox{#1$\scriptscriptstyle #2$}}}}
\newcommand{\matr}[1]{\mathlette{\boldmath}{#1}}

\providecommand{\GS}[2]{\mathcal{GS}(#1 \mid #2)}
\newcommand{\Op}[1]{\mathcal{O}(#1)}

\newcommand{\psiv}{\matr\psi}
\newcommand{\phiv}{\matr\phi}
\providecommand{\norm}[1]{\lVert#1\rVert}
\newfont{\bbb}{msbm10 scaled 700}

\newfont{\bb}{msbm10 scaled 1100}
\newcommand{\CC}{\mbox{\bb C}}

\newcommand{\RR}{\mbox{\bb R}}

\newcommand{\FF}{\mbox{\bb F}}

\newcommand{\NN}{\mbox{\bb N}}

\newcommand{\av}{{\bf a}}
\newcommand{\bv}{{\bf b}}

\newcommand{\ev}{{\bf e}}
\newcommand{\fv}{{\bf f}}
\newcommand{\gv}{{\bf g}}

\newcommand{\mv}{{\bf m}}
\newcommand{\nv}{{\bf n}}

\newcommand{\qv}{{\bf q}}
\newcommand{\rv}{{\bf r}}
\newcommand{\sv}{{\bf s}}

\newcommand{\vv}{{\bf v}}
\newcommand{\xv}{{\bf x}}
\newcommand{\yv}{{\bf y}}
\newcommand{\zv}{{\bf z}}
\newcommand{\zerov}{{\bf 0}}

\newcommand{\Am}{{\bf A}}

\newcommand{\Dm}{{\bf D}}

\newcommand{\Gm}{{\bf G}}
\newcommand{\Hm}{{\bf H}}
\newcommand{\Id}{{\bf I}}

\newcommand{\Mm}{{\bf M}}

\newcommand{\Om}{{\bf O}}
\newcommand{\Pm}{{\bf P}}
\newcommand{\Qm}{{\bf Q}}
\newcommand{\Rm}{{\bf R}}
\newcommand{\Sm}{{\bf S}}
\newcommand{\Tm}{{\bf T}}
\newcommand{\Um}{{\bf U}}

\newcommand{\Vm}{{\bf V}}
\newcommand{\Xm}{{\bf X}}

\newcommand{\Zm}{{\bf Z}}

\newcommand{\Cc}{{\cal C}}

\newcommand{\Nc}{{\cal N}}

\newcommand{\gammav}{\hbox{\boldmath$\gamma$}}
\newcommand{\deltav}{\hbox{\boldmath$\delta$}}

\newcommand{\muv}{\hbox{\boldmath$\mu$}}

\newcommand{\Sigmam}{\hbox{\boldmath$\Sigma$}}

\newcommand{\Xim}{\hbox{\boldmath$\Xi$}}

\renewcommand{\det}{{\hbox{det}}}

\newcommand{\eqdef}{\stackrel{\Delta}{=}}

\newcommand{\herm}{{\sf H}}

\newcommand{\transp}{{\sf T}}

\usepackage{stmaryrd} % for \mkv 